\documentclass[12pt,english]{article}
\usepackage{lmodern}
\usepackage[T1]{fontenc}
\usepackage[latin9]{inputenc}
\usepackage[letterpaper]{geometry}
\usepackage{color}
\usepackage{babel}
\usepackage{amsmath}
\usepackage{amsthm}
\usepackage{amssymb}
\usepackage{stmaryrd}
\usepackage{setspace}
\usepackage[authoryear]{natbib}
\usepackage[pdfusetitle,
 bookmarks=true,bookmarksnumbered=true,bookmarksopen=false,
 breaklinks=false,pdfborder={0 0 1},backref=false,colorlinks=true]
 {hyperref}
\hypersetup{
 colorlinks,pagebackref=false,citecolor=dark-blue,urlcolor=black,linkcolor=dark-blue}

\makeatletter
\theoremstyle{plain}
\newtheorem{assumption}{\protect\assumptionname}
\theoremstyle{plain}
\newtheorem{lem}{\protect\lemmaname}
\theoremstyle{remark}
\newtheorem{rem}{\protect\remarkname}
\theoremstyle{plain}
\newtheorem{prop}{\protect\propositionname}
\theoremstyle{plain}
\newtheorem{cor}{\protect\corollaryname}
\theoremstyle{remark}
\newtheorem{claim}{\protect\claimname}

\date{}

\usepackage{datetime}
\newdateformat{mydate}{\monthname[\THEMONTH] \THEYEAR}

\usepackage{amsfonts}
\usepackage{amsthm}
\usepackage{afterpage} 
\usepackage{color}
\usepackage{epsf}

\usepackage{fullpage}
\usepackage{latexsym}
\usepackage{bm}
\usepackage{bbm}

\usepackage[multiple]{footmisc}

\usepackage{floatrow}
\usepackage[margin=10pt,labelfont=bf,tableposition=bottom ]{caption}

\usepackage{tikz}
\usepackage{pgfplots}

\definecolor{dark-red}{rgb}{0.4,0.15,0.15}
\definecolor{dark-blue}{rgb}{0.15,0.15,0.75}
\definecolor{medium-blue}{rgb}{0,0,0.5}

\makeatother

\providecommand{\assumptionname}{Assumption}
\providecommand{\claimname}{Claim}
\providecommand{\corollaryname}{Corollary}
\providecommand{\lemmaname}{Lemma}
\providecommand{\propositionname}{Proposition}
\providecommand{\remarkname}{Remark}

\begin{document}
\title{Screening with Product Mismatch}
\author{Teck Yong Tan\thanks{Simon Business School, University of Rochester. Email: t.tan@rochester.edu.}}
\date{\mydate\today\\}
\maketitle
\begin{abstract}
A monopolist sells a product line whose variants are horizontally
differentiated from the buyers' perspective but ordered by production
cost. Buyers privately know their ideal product, and willingness to
pay may be correlated with horizontal need. The seller screens buyers
through product mismatch, and what she must screen determines whether
mismatch creates or reduces information rent. When buyers differ only
in horizontal need, mismatch creates rent: the seller induces less
mismatch, assigning served buyers products closer to their ideals
than under the first best. When willingness to pay is correlated with
horizontal need, mismatch instead reduces rent: the seller induces
more mismatch, sells the basic product to buyers whose efficient products
are advanced variants while excluding buyers better matched to it,
and stronger horizontal differentiation can expand coverage and raise
profit. Because mismatch is type-specific, optimal allocations are
determined by individual rationality rather than by incentive compatibility
alone.
\end{abstract}



\thispagestyle{empty}
\addtocounter{page}{-1}
\newpage

\section{Introduction}

Many firms sell product lines whose variants are designed for users
with different needs. A highly configurable cloud-computing service
gives a sophisticated user the flexibility needed for a complex deployment,
while imposing unnecessary complexity on a user with simpler needs.
Conversely, a preconfigured service is suitable for simple deployments
but may constrain users with specialized requirements. Similarly,
a lightweight AI model suffices for routine tasks but becomes inadequate
for demanding ones, while a more capable model imposes unnecessary
latency and cost on simple queries. The same pattern also appears
in tools, appliances, consumer electronics, and professional-service
plans designed for tasks of different scale or users with different
expertise. In each case, a variant well suited to one buyer can be
a poor fit for another, and buyers privately know which variants fit
their needs.

Such product lines have two common features. First, although buyers'
preferences are \emph{horizontally} differentiated, the product line
is ordered by cost from the seller's perspective: variants for more
demanding needs are costlier to supply. The seller may therefore prefer
to steer some buyers toward cheaper, more basic products rather than
supply each buyer his ideal variant. Second, buyers who need the costlier
variants also tend to value them more, because the scale or sophistication
of a demanding task raises the return from successfully completing
it. A buyer's ideal variant and his willingness to pay are thus naturally
correlated.

This paper develops a theory of product mismatch that arises from
screening in such product lines. Since the seminal analysis of \citet{Mussa_Rosen_1978(JET)}
and \citet{Maskin_Riley_1984(Rand)}, the theory of monopoly screening
has developed largely around vertical differentiation, where products
are ordered by a quality that all buyers rank alike. There, the seller
screens by degrading the quality each buyer receives, and stronger
buyer differentiation raises the cost of screening by increasing the
information rent conceded to buyers. In the settings above, however,
there is no common product ranking, so the seller can only screen
buyers by steering them toward products that fit them imperfectly,
rather than degrading quality.

The paper shows that when the screening instrument is product mismatch,
what the seller must screen determines whether that mismatch creates
or reduces information rent. Such duality has no immediate counterpart
in canonical vertical screening. When buyers differ only in their
horizontal need, steering a buyer from his ideal \emph{creates} information
rent for other buyers, so the seller induces less mismatch and provides
better product fit than the first-best allocation would. Stronger
horizontal differentiation then reduces coverage and profit. When
buyers' willingness to pay is correlated with their horizontal need,
mismatch instead \emph{reduces} information rent. The seller then
induces more mismatch, and stronger differentiation now induces the
seller to expand coverage and can even raise her profit. This case
also generates \emph{crowding out of the basic product}: the seller
always sells her most basic, lowest-cost product to some buyers whose
efficient product is a more advanced variant, while excluding buyers
better matched to the basic product. 

I study a model with a monopolist offering a product line. Each product
is indexed by $y\in\left[0,1\right]$ and ordered so that higher-indexed
products cost more to produce. There is a unit mass of buyers, each
demanding at most one unit. A buyer's private type $\theta\in\left[0,1\right]$
is his ideal product, and a type-$\theta$ buyer who consumes product
$y$ obtains gross utility $v\text{(\ensuremath{\theta}})-km\left(|y-\theta|\right)$.
The function $v$ is increasing---as in the applications, buyers
who prefer costlier-to-produce products have higher willingness to
pay. The \emph{mismatch }function $m$ is strictly increasing with
$m\left(0\right)=0$. A buyer's mismatch disutility is thus determined
by the distance $\left|y-\theta\right|$, and $k$ scales the strength
of horizontal differentiation. Section \ref{Section:Model} provides
further discussion on this formulation of horizontal differentiation
and how it relates to the applications above. Since production cost
increases in the product index, efficient (first-best) product allocation
must trade off product fit against production cost, so the efficient
allocation generally assigns each type a product below his ideal.

In a screening problem, the rent the seller concedes is governed by
how fast a buyer's equilibrium payoff $u\left(\theta\right)$ changes
with his type. Here, by the envelope theorem, this rate decomposes
into two channels:\footnote{The model assumes that $m'\left(0\right)=0$, so $m\left(\left|y-\theta\right|\right)$
is differentiable in $\theta$ everywhere.}
\begin{equation}
u'\left(\theta\right)=\underbrace{v'\left(\theta\right)}_{\text{WTP screening}}-\underbrace{k\partial_{\theta}m\left(\left|\alpha\left(\theta\right)-\theta\right|\right)}_{\text{product-fit screening}},\label{eq:Intro-u'}
\end{equation}
where $\alpha\left(\theta\right)$ is the product type $\theta$ consumes
in equilibrium. The first term is the vertical-value channel: it captures
how willingness to pay (WTP) rises with type. The second term is the
horizontal-fit channel: it captures how the buyer\textquoteright s
mismatch with the assigned product changes with type. The seller screens
both channels through the single instrument of product allocation
$\alpha$, and the interaction between the two determines whether
product mismatch creates information rent or reduces it.

\begin{figure}[h]
\caption{\protect\label{Figure:OptimalMechanism}Product Allocation \\ }

\begin{tikzpicture}[x=6.4cm, y=5.6cm, >=stealth]
  \def\bh{0.125}    
  \def\ths{0.25}    


 \fill[black!8] (1,0) rectangle (0.8,1);
  \node[align=center] at (0.9, 0.38) {\footnotesize excluded\\ \footnotesize types};

\draw[black!30, dashed] (0.8,0) -- (0.8,1);
\draw[black!30, dashed] (1,0) -- (1,1);

\node[above] at (0.5,1.2) {\underline {Pure horizontal differentiation (Section 3)}};
 
  \draw[->] (-0.02,0) -- (1.10,0) node[right] {$\theta$};
  \draw[->] (0,-0.02) -- (0,1.10) node[above] {$y$};
  \node[below left] at (0,0) {$0$};
  \draw (1,0.01) -- (1,-0.01) node[below] {$1$};
  \draw (0.01,1) -- (-0.01,1) node[left] {$1$};

  \draw[red, line width=1.7pt] (0,0) -- (\bh,0) -- (\ths,\ths) -- (0.8,0.8);
  \draw[decorate, decoration={brace, mirror, amplitude=4pt}]
    (0,-0.09) -- (\bh,-0.09);

\node[align=left] at (0.4,-0.16) {\footnotesize pooling at the  \footnotesize basic product $y=0$};
 
  
  \draw[green!60!black, densely dotted, line width=1.2 pt] (0,0) -- (1,1);

  \draw[blue, dashed, line width=1.2 pt] (0,0) -- (0.25,0) -- (1,0.75);

  \draw[green!60!black, densely dotted, line width=1.2 pt] (0.1,0.97) -- (0.21,0.97)
    node[right] {\small  ideal product $y=\theta$};
  \draw[blue, dashed, line width=1.2 pt] (0.1,0.89) -- (0.21,0.89)
    node[right] {\small  first-best product};
  \draw[red, line width=1.7pt] (0.1,0.81) -- (0.21,0.81)
 	node[right] {\small optimal mechanism};


\end{tikzpicture}$\ \ \ $\begin{tikzpicture}[x=6.4cm, y=5.6cm, >=stealth]
  \def\tauv{0.20882}   
  \def\rr{0.25}        
  \def\betat{0.42138}  
  \def\phit{0.50942}   
  \def\gam{0.91667}    
  \def\yphi{0.17609}   
  \def\ygam{0.58333}   
  \fill[black!8] (0,0) rectangle (\tauv,1);
  \node[align=center] at (0.5*\tauv, 0.58) {\footnotesize excluded\\ \footnotesize types};

  \foreach \x in {\tauv,\betat,\phit,\gam} \draw[black!30, dashed] (\x,0) -- (\x,1);

\node[above] at (0.5,1.2) {\underline {Correlated WTP and product fit (Section 4)}};
 
  \draw[->] (-0.02,0) -- (1.10,0) node[right] {$\theta$};
  \draw[->] (0,-0.02) -- (0,1.10) node[above] {$y$};
  \node[below left] at (0,0) {$0$};
  \draw (1,0.01) -- (1,-0.01) node[below] {$1$};
  \draw (0.01,1) -- (-0.01,1) node[left] {$1$};
  \draw (\tauv,0.01) -- (\tauv,-0.01) node[below] {\small $\tau$};
  \draw (\betat,0.01) -- (\betat,-0.01) node[below] {\small $\beta^{\tau}$};
  \draw (\phit,0.01) -- (\phit,-0.01) node[below] {\small $\varphi^{\tau}$};
  \draw (\gam,0.01) -- (\gam,-0.01) node[below] {\small $\gamma$};

  \draw[red, line width=1.7pt]
    (\tauv,0) -- (\betat,0) -- (\phit,\yphi) -- (\gam,\ygam) -- (1,0.75);
  \draw[decorate, decoration={brace, mirror, amplitude=4pt}]
    (\tauv,-0.09) -- (\betat,-0.09);
 
\node[align=left] at (0.55,-0.16) {\footnotesize pooling at the  \footnotesize basic product $y=0$};

  \draw[blue, dashed, line width=1.2 pt] (0,0) -- (0.25,0) -- (1,0.75);

  \draw[green!60!black, densely dotted, line width=1.2 pt] (0,0) -- (1,1);


\end{tikzpicture}
\raggedright{}\small{In this example, $\theta$ is uniformly distributed,
$m\left(z\right)=z^{2}/2$, $k=3$, and $c=3/4$. The derivations
are provided in appendix \ref{OA:Uniform-Distribution}. } 
\end{figure}

Section \ref{Section:PureHorizontal} isolates the horizontal-screening
force by assuming that $v$ is constant, so buyers differ only in
their horizontal product fit. In the envelope condition in (\ref{eq:Intro-u'}),
the WTP-screening channel is shut down, so all information rent comes
from the product-fit channel. Two distinctive features follow. First,
exclusion is \emph{from the top}, since higher types prefer products
that are costlier to produce and are therefore costlier to serve.
Second, each served type receives a product \emph{closer} to his ideal
than under the first-best allocation. The reason is that while downward
mismatch saves production cost, it also steepens the utility schedule
and increases the rent of lower types. The seller therefore treats
mismatch as more costly than the social planner does and induces less
of it. This logic also determines the comparative static with respect
to the intensity of horizontal differentiation: a higher $k$ raises
the rent that the seller concedes to the buyers, lowering her profit
and causing her to reduce coverage.

The left panel of Figure \ref{Figure:OptimalMechanism} illustrates
an example. The first-best allocation (dashed) is below the ideal
allocation (dotted) because of the production cost of any $y>0$,
and the optimal mechanism's allocation (solid) lies weakly within
the wedge between the two allocations. Served buyers therefore receive
products closer to their ideal variants than under the first best,
except for some low types receiving the basic product $y=0$. In the
example, the highest served types even receive their ideal products.

Section \ref{Section:Vertical} analyzes the case in which $v$ is
strictly increasing, so buyers who prefer the costlier-to-produce
variants also have higher willingness to pay. The WTP-screening channel
in (\ref{eq:Intro-u'}) is therefore active, and the right panel of
Figure \ref{Figure:OptimalMechanism} illustrates two reversals relative
to Section \ref{Section:PureHorizontal}. First, exclusion now starts
\emph{from the bottom} rather than from the top. Second, the optimal
mechanism's allocation now lies below the first-best allocation rather
than between the first-best and ideal allocations. Each served buyer
is therefore allocated a product \emph{further} below his first-best
allocation. The figure also shows crowding out of the basic product,
which never arises in the previous case: here, $y=0$ is sold to higher
types whose efficient product is non-basic $\left(y>0\right)$, whereas
lower types who are better matched to $y=0$ are excluded.

Both reversals are due to a change in which type earns rent because
of the active WTP-screening channel. Although a higher type still
prefers products that are costlier to produce, the seller can now
capture him with a low-cost product, since his willingness to pay
is higher. Higher types therefore crowd out the lower types better
matched to those low-cost products. This changes the role of downward
mismatch. In Section \ref{Section:PureHorizontal}, the low types
earn rent, so downward mismatch creates rent and the seller induces
less of it. Now, the high types earn rent, and downward mismatch \emph{reduces}
their rent by making lower-types' contracts less attractive to them.
Mismatch therefore turns from creating rent to reducing it, and the
seller thus induces more mismatch. The underlying logic also reverses
the comparative static relative to Section \ref{Section:PureHorizontal}:
stronger horizontal differentiation makes mismatch a stronger deterrent
against high types\textquoteright{} mimicking lower types, lowering
the rent the seller concedes, so a higher $k$ now expands coverage
and can raise the seller's profit.

More broadly, the analysis in sections \ref{Section:PureHorizontal}
and \ref{Section:Vertical} also illustrates why screening with product
mismatch is not merely vertical screening with a different allocation
variable. In vertical screening, higher types value any given quality
more than lower types, so no quality distortion can overturn the monotonicity
of utility in type. Thus, once the lowest served type\textquoteright s
individual rationality (IR) constraint is satisfied, all higher types'
IR are as well. With horizontal product mismatch, the mismatch disutility
is \emph{type-specific}, and whether a product becomes a better or
worse fit as the type changes depends on the allocation itself. The
IC-implied utility schedule is therefore not always monotone, so the
``worst type''---whose IR implies IR for the rest---cannot be
identified a priori. The analysis therefore must track every type's
IR constraint in the optimization. I solve the problem by constructing
the Lagrangian multiplier on the IR constraints and using it to adjust
the virtual surplus maximized by the seller. Economically, the multiplier
measures how much product mismatch must be reduced at each type to
maintain IR of \emph{other} types.\footnote{These required reductions are why the allocation schedule of the optimal
mechanism has kinks above $y=0$, as in Figure \ref{Figure:OptimalMechanism}.} The resulting adjustment produces pooling in utility rather than
the familiar pooling in allocation from ironing in vertical screening.
The broader takeaway is that in screening with horizontal product
mismatch, IR---rather than IC alone in vertical screening---also
shapes the seller-optimal allocations. 

\paragraph{Related Literature.}

The theory of monopoly screening has focused largely on vertically
differentiated products, whereas work on horizontal differentiation
remains sparse. \citet{Jiang2007opaque} and \citet{Fayxie2008probabilistic}
study a multiproduct monopoly facing horizontally differentiated buyers
on the Hotelling line and offering lotteries over the two variants
a buyer receives. \citet{Andersoncelik2020opaque}, and \citet{Balestrieri2021surprises}
adopt a screening approach and characterize the optimal lottery. \citet{Loertschermuir2025hotelling}
generalize the analysis to an auction environment with multiple buyers
and show that the optimal mechanism may require lottery-augmented
auctions. The current paper instead studies only posted-price mechanisms.
However, the product space is richer, which allows one to analyze
equilibrium product mismatch. Buyers here also exhibit both horizontal
and vertical heterogeneity. 

\citet{Rochet2002nonlinear} study multidimensional screening in which
buyers have a vertically differentiated taste for quality and a random,
privately known outside option. By interpreting the outside option
as an additive horizontal preference parameter, their model also features
screening under both vertical and horizontal differentiation. The
two screening problems differ in how the screening instrument interacts
with each dimension of heterogeneity. There the seller's allocation
is quality, which interacts only with the vertical taste parameter,
so screening is closer to \citet{Mussa_Rosen_1978(JET)}, with the
horizontal taste affecting participation. Here the seller's allocation
is the product variant, which enters the buyer's mismatch disutility,
so the seller directly screens the horizontal dimension. Vertical
and horizontal differentiation here are also correlated and determined
by a one-dimensional type $\theta$. Such a correlation arises naturally
in the intended applications.

The analysis also relates to the literature on ``countervailing incentives''
in mechanism design, beginning with \citet{LewisSappington1989}.
The difficulty in such problems is that participation is not pinned
down by a single extreme type. \citet{MaggiRodriguezClare1995(JET)}
and \citet{Jullien_2000(JET)} study countervailing incentives arising
from type-dependent outside options.\footnote{See also more recent work by \citet{kang2024optimal}, \citet{DworczakMuir2025}
and \citet{ValenzuelaStookey2025}, who develop ironing-based methods
for solving mechanism design problems in which type-dependent participation
constraints naturally arise.} Countervailing incentives also arise in the Hotelling setting in
\citet{Loertschermuir2025hotelling} due to the allocation space being
two-dimensional. Here, countervailing incentives arise because the
ideal allocation is type-specific due to horizontal differentiation.
This form of countervailing incentives is closer to those arising
in mechanism-design settings with costly misrepresentation, such as
\citet{Maggi_RodriguezClare_1995(Rand)} and \citet{TanManipulableObservables}.
However, whereas utility always remains weakly monotone in type under
the optimal mechanisms in those settings, it can be nonmonotone in
type under the optimal mechanism here.

\section{\protect\label{Section:Model}Model}

A monopolist offers a product line, with each product indexed by $y\in Y:=\left[0,1\right]$.
There is a unit mass of buyers.\footnote{Throughout, I use the female pronoun for the seller and the male pronoun
for buyers.} Each buyer demands at most one unit and has a private type $\theta\in\Theta:=\left[0,1\right]$,
drawn from the distribution $F$.
\begin{assumption}
\label{Assumption:F}$F$ has full support on $\Theta$ and a differentiable
density $f$. Both $F$ and $1-F$ are strictly log-concave.
\end{assumption}
If type $\theta$ buys product $y$ at price $p$, his utility is
\[
v\left(\theta\right)-kM\left(y,\theta\right)-p,\ \ \ \ \text{ where }M\left(y,\theta\right)=m\left(\left|y-\theta\right|\right).
\]
The function $v:\Theta\rightarrow\left[0,1\right]$ is nondecreasing
and captures the vertical component of buyers' values. The function
$M$ is the buyer's mismatch disutility, capturing horizontal preferences,
and $k>0$ parameterizes the intensity of horizontal differentiation. 
\begin{assumption}
\label{Assumption:m-1}$m:\mathbb{R}_{+}\rightarrow\mathbb{R}_{+}$
is thrice differentiable on $\mathbb{R}_{++}$, strictly increasing
and strictly convex, with $m\left(0\right)=m'\left(0\right)=0$. 
\end{assumption}
Type $\theta$'s ideal product is $y=\theta$, which generates no
mismatch disutility. Any product $y\ne\theta$ generates disutility
$kM\left(y,\theta\right)$, which is increasing in the distance $\left|y-\theta\right|$.
The condition $m'\left(0\right)=0$ implies that a marginal mismatch
is costless, while convexity of $m$ makes larger mismatch increasingly
costly. I impose further curvature assumptions on $m$ so that the
virtual surplus functions defined later are quasiconcave in $y$ and
have the relevant single-crossing properties.
\begin{assumption}
\label{Assumption:m-2}Both $m$ and $m'$ are strictly log-concave
on $\mathbb{R}_{++}$. Furthermore, $m''\left(0\right)>0$, and $0\le m'''\left(z\right)/m''\left(z\right)\le m''\left(0\right)/m'\left(1\right)$
for all $z\in(0,1]$.
\end{assumption}
Log-concavity of $m$ and $m'$ is readily satisfied. Examples include
$m\left(z\right)=z^{a}$ for $a\ge2$ and $m\left(z\right)=\exp\left(az\right)-az-1$
for $0<a\le\log2$. The second part of Assumption \ref{Assumption:m-2}
is more stringent but is required only in Section \ref{Section:PureHorizontal}.
Remark \ref{Remark:Horizontal-assumption} notes what happens if it
is violated. The exponential example continues to satisfy these conditions,
while the power family requires $a=2$ to satisfy $m''\left(0\right)>0$.\footnote{A simple way to generate examples is to specify $m''$ directly. For
example, let $m''\left(z\right)=\left(1+a_{1}z\right)^{a_{2}},$ with
$a_{2}>0$ and $0\le a_{1}\le\left(2+\frac{1}{a_{2}}\right)^{1/\left(a_{2}+1\right)}-1$,
which satisfies the second part of Assumption \ref{Assumption:m-2}.
Imposing $m'\left(0\right)=m\left(0\right)=0$ then yields $m\left(z\right)=\frac{\left(1+a_{1}z\right)^{a_{2}+2}-1-\left(a_{2}+2\right)a_{1}z}{a_{1}^{2}\left(a_{2}+1\right)\left(a_{2}+2\right)}$,
which satisfies the other conditions in Assumptions \ref{Assumption:m-1}
and \ref{Assumption:m-2}. Taking the limit of $a_{1}$ to zero, the
function becomes $m\left(z\right)=z^{2}/2$.}

The seller's cost of producing a unit of product $y$ is $cy$, where
$c\ge0$. Since producing $y=0$ is costless, I refer to $y=0$ as
the \emph{basic product} and to any $y>0$ as a non-basic product.
The seller's profit from selling product $y$ at price $p$ is therefore
$p-cy$. The seller publicly announces a price $p\left(y\right)$
for each $y\in Y$, after which each buyer either buys a product or
takes his outside option, normalized to zero.

\subsection{\protect\label{Subsection:ModelDiscussion}Discussion of Model}

Indexing products by $y$ and ordering them by production cost is
without loss of generality. Labeling each buyer by his ideal product
is also a convention, and $m\left(0\right)=0$ normalizes $v\left(\theta\right)$
as type $\theta$'s utility from consuming his ideal product at zero
price. 

The substantive restriction is that horizontal preferences have a
one-dimensional ordered structure, under which nearby types have nearby
ideal products, so similarity in horizontal taste is measured by distance
in type space. Relative fit is then \emph{symmetric} and\emph{ continuous}:
if type $\theta_{1}$ views type $\theta_{2}$'s ideal product as
a close substitute, then type $\theta_{2}$ must similarly view type
$\theta_{1}$'s ideal product as close; for any two products $y$
and $y'$, $M\left(y,\theta\right)-M\left(y',\theta\right)$ varies
continuously in $\theta$, so relative product fit cannot jump discontinuously
as type changes. The form $M\left(y,\theta\right)=m\left(\left|y-\theta\right|\right)$
implies that mismatch depends only on distance, not on the buyer's
type or on whether the mismatch is upward or downward. I impose this
structure for tractability; the main economic mechanisms do not rely
on it, and one could instead take $M\left(y,\theta\right)$ as the
primitive and impose analogous assumptions directly on it.

Section \ref{Section:PureHorizontal} assumes that $v$ is constant,
so buyers differ only in horizontal taste. Section \ref{Section:Vertical}
assumes that $v$ is strictly increasing, so buyers also differ in
vertical value, with each buyer's vertical value perfectly correlated
with his horizontal taste. This perfect correlation keeps the type
space one-dimensional and captures settings in which the same features
that raise willingness to pay---such as intensity of use, technical
sophistication, or scale of operations---also determine which product
specification best fits the buyer\textquoteright s needs.

On the supply side, $c=0$ means that all variants are equally costly
to produce, with the common cost normalized to zero. When $c>0$,
higher-indexed products are costlier to supply. Combined with strictly
increasing $v$, this means that the products preferred by higher-value
buyers are also more costly to produce. This captures settings in
which higher-value buyers have more demanding or specialized needs,
so their ideal products require costlier inputs, greater customization,
or more elaborate design.

Finally, the lower bound $y=0$ for the product space is substantive.
It represents the most basic product in the product line, with production
cost normalized to zero. This parallels the nonnegativity constraint
on quality in vertical screening: just as quality below zero has no
interpretation, a product below the most basic one is not well-defined.
By contrast, the upper bound $y=1$ is unimportant; the analysis is
unchanged if $Y$ extends above $1$.

\subsection{\protect\label{Subsection:Applications}Applications}

One application is pricing for digital infrastructure services such
as cloud-computing products and AI model lineups. Amazon Web Services
(AWS), for example, offers services that differ in how much control
and abstraction they provide. Elastic Compute Cloud (EC2) gives sophisticated
users fine-grained control over infrastructure, Lightsail offers simplified
preconfigured plans for users with basic deployment needs, and services
such as Elastic Beanstalk and Fargate lie in between, providing differing
degrees of abstraction over the infrastructure.\footnote{See https://docs.aws.amazon.com/whitepapers/latest/aws-overview/compute-services.html
for a description of the AWS services.} AI providers similarly offer model lineups designed for different
task requirements. A lightweight model is better suited for routine
classification, summarization, or drafting tasks, while a more capable
model is meant for complex reasoning or coding tasks.\footnote{For example, \citet{azure_choose_model} publishes guidance to help
buyers select among available AI models based on task fit, cost, context
window, and performance. } In these settings, mismatch costs run in both directions: a flexible
or high-capability variant imposes unnecessary complexity, latency,
or cost on a basic user or task, while a simplified or low-capability
variant constrains a user who requires more control or performance.\footnote{The existence of third-party platforms such as Heroku, which is built
on top of AWS and charges a markup to simplify application deployment,
illustrates the value of abstracting complexity for some users. } The model's correlation between value and cost is also natural in
these settings, since more sophisticated users tend to have higher
willingness to pay and require services that are costlier to provide.

A second application is tax-preparation services. Providers such as
TurboTax and H\&R Block offer product lines ranging from basic self-service
software, through assisted online filing, to full-service preparation
by an expert. A taxpayer\textquoteright s return complexity determines
which product provides the best fit, and mismatch runs in both directions.
Basic software may underserve a self-employed taxpayer, landlord,
or investor whose return requires additional schedules, deductions,
and judgment. Full-service may overshoot a wage earner with a single
W-2, who can avoid preparer interactions and filing delays by completing
the return independently. Value and cost also move together: expert
preparation is costlier to provide than software-based filing, and
taxpayers who prefer it have higher willingness to pay because the
stakes---tax liabilities, audit exposure, and missed deductions---rise
with complexity.

\subsection{Benchmark I: No Horizontal Differentiation}

If $k=0$, buyers are indifferent among products, so the seller offers
only a lowest-cost product: $y=0$ if $c>0$, and any product if $c=0$.
The problem reduces to a standard monopoly pricing of a single product
with buyer valuation $v\left(\theta\right)$. Thus, screening through
the product line is possible only when $k>0$.

\subsection{Benchmark II: First-best Product Allocation}

The \emph{first-best} product allocation for type $\theta$ maximizes
$v\left(\theta\right)-kM\left(y,\theta\right)-cy$. The solution lies
in $\left[0,\theta\right]$, since replacing any $y>\theta$ with
$y=\theta$ reduces both mismatch and production cost. On $\left[0,\theta\right]$,
the derivative with respect to $y$ is $km'\left(\theta-y\right)-c$.
Therefore, the first-best allocation is unique and given by
\begin{equation}
\alpha^{FB}\left(\theta\right):=\max\left\{ \theta-\delta^{FB},0\right\} \ ,\ \ \ \ \text{ where }\delta^{FB}=\left(m'\right)^{-1}\left(c/k\right).\label{eq:delta-FB}
\end{equation}
The first-best allocation does not depend on $v$, so this benchmark
applies for both sections \ref{Section:PureHorizontal} and \ref{Section:Vertical}.
Some downward mismatch is efficient because higher-indexed products
are costlier to produce, while marginal mismatch is costless by $m'\left(0\right)=0$.
Each type is mismatched downward from his ideal by the common factor
$\delta^{FB}$, subject to $y\ge0$. I rule out the case in which
the first best assigns every type to the most basic product $y=0$. 
\begin{assumption}
\label{Assumption:FB-ProductLine}$c/k<m'\left(1\right)$.
\end{assumption}

\section{\protect\label{Section:PureHorizontal}Pure Horizontal Preferences}

In this section, I assume that $v\left(\theta\right)=v_{0}$ for all
$\theta\in\Theta$, where $v_{0}\in\left(0,1\right)$ is commonly
known. Buyers then differ only in their horizontal tastes, which isolates
the role of horizontal preferences in the seller's screening problem.

If $c=0$, there is no cost variation across products, so the seller
has no reason to steer any buyer away from his ideal product. By setting
a uniform price $v_{0}$ across the product line, each type selects
his ideal product and the seller earns the entire first-best surplus.
\begin{lem}
\label{Lemma:Horizontal-c=00003D0}Suppose $v\left(\theta\right)=v_{0}$
for all $\theta$ and $c=0$. The uniform price schedule $p\left(y\right)=v_{0}$
is optimal. Under this schedule, each type $\theta$ purchases his
ideal product $y=\theta$, and the seller earns profit $v_{0}$.
\end{lem}
The remainder of the section addresses the $c>0$ case, where higher-indexed
products are costlier to produce. This is a screening problem. By
the revelation principle, it suffices to consider direct mechanisms.
A direct mechanism is a pair $\left(\alpha,t\right)$ consisting of
an allocation rule $\alpha:\Theta\rightarrow Y\cup\left\{ \varnothing\right\} $
and a transfer rule $t:\Theta\to\mathbb{R}$, with the convention
that $\alpha\left(\theta\right)=\varnothing$ denotes exclusion of
type $\theta$ and $t\left(\theta\right)=0$ in that case. Let $\Theta_{I}:=\left\{ \theta\in\Theta:\alpha\left(\theta\right)\in Y\right\} $
denote the \emph{inclusion set}. Type $\theta$'s utility under truthful
reporting is
\begin{equation}
u\left(\theta\right)=v_{0}-kM\left(\alpha\left(\theta\right),\theta\right)-t\left(\theta\right)\ \text{ if }\theta\in\Theta_{I}\ ;\ \ \ \ u\left(\theta\right)=0\ \text{ if }\theta\notin\Theta_{I}.\label{eq:u-h}
\end{equation}
The mechanism is \emph{incentive compatible (IC)} if every type weakly
prefers reporting truthfully to any alternative report, and \emph{individually
rational (IR)} if $u\left(\theta\right)\ge0$ for all $\theta\in\Theta$,
which holds automatically for excluded types. The seller chooses an
IC and IR mechanism to maximize 
\[
\int_{\Theta_{I}}\left[t\left(\theta\right)-c\alpha\left(\theta\right)\right]f\left(\theta\right)d\theta.
\]

Screening with horizontal differentiation differs from vertical screening
problems à la \citet{Mussa_Rosen_1978(JET)} in two ways that require
distinct treatment. First, exclusion cannot be folded into the allocation
space here. In vertical screening, exclusion of \emph{any} type can
be represented by assigning the contract with the zero-quality allocation
and zero price. Here, every $y\in Y$ is an actual product that some
type may strictly prefer to no purchase, so exclusion cannot be represented
by some common allocation in $Y$. The inclusion set $\Theta_{I}$
must therefore be specified explicitly, with $\varnothing$ treated
as a distinct option.

Second, the IC-implied utility schedule can be nonmonotone. Let subscripts
on $M$ denote partial derivatives. By the envelope theorem, any IC
mechanism satisfies 
\begin{equation}
u'\left(\theta\right)=-kM_{2}\left(\alpha\left(\theta\right),\theta\right)\ \ \ \ \text{a.e. on }\Theta_{I}.\label{eq:u'-horizontal}
\end{equation}
Since $M\left(y,\theta\right)=m\left(\left|y-\theta\right|\right)$
and $m'\left(0\right)=0$, 
\begin{equation}
M_{2}\left(y,\theta\right)=\begin{cases}
-m'\left(y-\theta\right) & \text{ if }y>\theta\\
m'\left(\theta-y\right) & \text{ if }y<\theta\\
0 & \text{ if }y=\theta
\end{cases}.\label{eq:M2}
\end{equation}
The sign of $u'\left(\theta\right)$ depends on the direction of mismatch:
$u'\left(\theta\right)>0$ when $\alpha\left(\theta\right)>\theta$,
and $u'\left(\theta\right)<0$ when $\alpha\left(\theta\right)<\theta$.
In vertical screening, $u'$ has constant sign---in \citet{Mussa_Rosen_1978(JET)},
$u'\left(\theta\right)$ equals the quality allocated to type $\theta$
and is therefore nonnegative---so a worst type can be identified,
and his IR constraint implies IR for all other types. Here, the sign
of $u'$ varies with the allocation, so the worst type cannot be identified
ex ante. 
\begin{lem}
\label{lem:h-IC}Suppose $v(\theta)=v_{0}$ for all $\theta\in\Theta$
and $c>0$. In any optimal mechanism, there exists $\tau\in\Theta$
such that $\Theta_{I}=\left[0,\tau\right]$. Moreover, fixing an inclusion
set $\left[0,\tau\right]$, an IR mechanism is IC if and only if the
envelope condition (\ref{eq:u'-horizontal}) holds on $\left[0,\tau\right]$,
$\alpha$ is nondecreasing on $\left[0,\tau\right]$, and, if $\tau<1$,
the cutoff type satisfies $u\left(\tau\right)=0$. 
\end{lem}
Higher types prefer higher-indexed products, which are costlier to
produce, so higher types are costlier to serve and exclusion thus
begins from the top. Given that $\Theta_{I}$ is an interval, the
IC conditions in the second part of Lemma \ref{lem:h-IC} follow from
standard arguments. The condition $u\left(\tau\right)=0$ is the IC
requirement for excluded types: if $u\left(\tau\right)>0$, then by
continuity types just above $\tau$ would obtain positive utility
from type $\tau$'s contract, contradicting exclusion. Given Lemma
\ref{lem:h-IC}, the analysis proceeds in two steps: Step 1 fixes
$\tau$ and characterizes the optimal mechanism that serves only types
in $\Theta_{I}=\left[0,\tau\right]$. Step 2 then optimizes over $\tau$.

\paragraph{Step 1. }

Substitute in $t\left(\theta\right)=v_{0}-kM\left(\alpha\left(\theta\right),\theta\right)-u\left(\theta\right)$.
The \emph{fixed-$\tau$ problem} is
\begin{equation}
\Pi^{h}\left(\tau\right)=\underset{\alpha,u\text{ s.t. IC, IR}}{\max}\ \int_{0}^{\tau}\left[v_{0}-kM\left(\alpha\left(\theta\right),\theta\right)-c\alpha\left(\theta\right)-u\left(\theta\right)\right]f\left(\theta\right)d\theta.\tag{\ensuremath{\mathcal{P}^{h}\left(\tau\right)}}\label{eq:Program-h}
\end{equation}
Fix $\tau<1$, so $u\left(\tau\right)=0$ by Lemma \ref{lem:h-IC};
the $\tau=1$ case is handled in the proof, where $u\left(\tau\right)$
is shown to be zero endogenously at the optimum. The envelope condition
(\ref{eq:u'-horizontal}) gives $u\left(\theta\right)=\int_{\theta}^{\tau}kM_{2}\left(\alpha\left(s\right),s\right)ds$.
Substituting $u$ into the objective and integrating by parts yields
\[
\Pi^{h}(\tau)\ =\ \max_{\alpha\ \text{nondecreasing}}\int_{0}^{\tau}\psi^{h}(\alpha\left(\theta\right),\theta)f(\theta)d\theta\ \ \ \text{ s.t. }\ \ \ u\left(\theta\right)\ge0\ \ \forall\theta\in\left[0,\tau\right],
\]
where 
\[
\psi^{h}\left(y,\theta\right):=v_{0}-kM\left(y,\theta\right)-cy-kM_{2}\left(y,\theta\right)\frac{F\left(\theta\right)}{f\left(\theta\right)}
\]
is the virtual surplus from type $\theta$ at allocation $y$. Define
\begin{equation}
\alpha^{h}\left(\theta\right):=\underset{y\in\left[0,\theta\right]}{\arg\max}\psi^{h}\left(y,\theta\right)\label{eq:alpha-h}
\end{equation}
Thus $\alpha^{h}\left(\theta\right)$ maximizes $\psi^{h}\left(\cdot,\theta\right)$
over only products weakly below type $\theta$'s ideal. The proof
of Lemma \ref{lem:h-step-1} shows that $\psi^{h}\left(\cdot,\theta\right)$
is strictly concave on $\left[0,\theta\right]$, so $\alpha^{h}\left(\theta\right)$
is well-defined.
\begin{lem}
\label{lem:h-step-1}For any $\tau>0$, the unique solution to the
fixed-$\tau$ problem \ref{eq:Program-h} is $\alpha^{h}$.
\end{lem}
Observe that $\alpha^{h}$ is independent of $\tau$. To develop intuitions
for how screening affects the product that buyers buy, I characterize
the properties of $\alpha^{h}$. Since $\alpha^{h}$ is the solution,
it must be nondecreasing, so $\beta^{h}:=\inf\left\{ \theta\in\Theta|\alpha^{h}\left(\theta\right)>0\right\} $
is well-defined. 
\begin{lem}
\label{lem:alpha-h}$\beta^{h}$ is in $\left(0,1\right)$, and $\alpha^{h}\left(\theta\right)=0$
for all $\theta\in\left[0,\beta^{h}\right]$. 
\begin{itemize}
\item Suppose $m''\left(0\right)/f\left(1\right)<c/k$.\footnote{If $f\left(1\right)=0$, define the left-hand side as $\lim_{s\uparrow1}m''\left(0\right)/f\left(s\right)$.}
Then $\alpha^{FB}\left(\theta\right)<\alpha^{h}\left(\theta\right)<\theta$
for all $\theta\in(\beta^{h},1]$.
\item Suppose $m''\left(0\right)/f\left(1\right)\ge c/k$. Then there exists
a unique $\theta^{*}\in(\beta^{h},1]$, characterized by 
\begin{equation}
m''\left(0\right)\frac{F\left(\theta^{*}\right)}{f\left(\theta^{*}\right)}=c/k,\label{eq:theta-*}
\end{equation}
 such that $\alpha^{FB}\left(\theta\right)<\alpha^{h}\left(\theta\right)<\theta$
for $\theta\in\left(\beta^{h},\theta^{*}\right)$, and $\alpha^{h}\left(\theta\right)=\theta$
for $\theta\ge\theta^{*}$. 
\end{itemize}
\end{lem}
Under $\alpha_{h}$, low types $\theta\le\beta^{h}$ are pooled at
the lowest-cost basic product $y=0$. Above the pooling region, every
served type receives a product \emph{closer} to his ideal than under
the first best, and types above $\theta^{*}$, when $\theta^{*}$
exists, receive their ideal product. These properties reflect the
interaction between cost saving and limiting information rent. As
in the first-best problem, the seller wants to induce downward mismatch
because higher-indexed products are costlier to produce. Unlike the
social planner, the seller's virtual surplus must also account for
information rent. Since $u'\left(\theta\right)=-km'\left(\theta-\alpha^{h}\left(\theta\right)\right)$,
increasing the downward mismatch makes the utility schedule decline
more steeply with type and raises the rent of the lower types. Downward
mismatch thus creates information rent. The seller consequently treats
downward mismatch as costlier than the social planner does and induces
less of it.

In fact, if $m''\left(0\right)/f\left(1\right)\ge c/k$, the seller
could further lower the low types' information rent by inducing\emph{
upward} mismatch for types above $\theta^{*}$---for these types,
their virtual surplus $\psi^{h}\left(\cdot,\theta\right)$ is maximized
by some $y>\theta$. However, this is never optimal once IR constraints
are taken into account. By (\ref{eq:u'-horizontal}) and (\ref{eq:M2}),
upward mismatch makes $u'\left(\theta\right)>0$, so utility rises
in type. Since the highest-served type $\tau$ must earn zero utility,
the utility schedule must then fall back over some region of $\left(\theta^{*},\tau\right)$
to reach zero at $\tau$. This means that the rent saved generated
by some upward mismatch at some $\hat{\theta}>\theta^{*}$ can only
be collected by inducing some downward mismatch above $\hat{\theta}$,
which again creates rent for lower types. The proof of Lemma \ref{lem:h-step-1}
shows that once this tradeoff is accounted for, the seller is better
off giving types above $\theta^{*}$ their ideal product instead of
inducing upward mismatch for them. The proof establishes this result
by constructing a Lagrangian multiplier on the IR constraints and
verifying optimality by weak duality. Because the same proof strategy
is used more substantially in Section \ref{Section:Vertical}, I defer
further discussion of this proof strategy to that section.
\begin{rem}
\label{Remark:Horizontal-assumption}The previous argument relies
on the curvature assumption on $m$ stated in the second part of Assumption
\ref{Assumption:m-2}. Without it, the solution may have to further
reduce the downward mismatch for some types below $\theta^{*}$ to
maintain IR for types above $\theta^{*}$, so the solution is comparably
less clean, which in turn complicates Step 2 below. Nevertheless,
the main economic forces in this section are unaffected: when buyers
differ only horizontally, exclusion must start from the high types,
and downward mismatch creates information rent for the low types,
so the seller generally induces less of it than the social planner
would. 
\end{rem}

\paragraph{Step 2. }

Since the solution for the fixed-$\tau$ problem \ref{eq:Program-h}
is independent of $\tau$, the seller's expected profit from serving
types $\left[0,\tau\right]$ is $\Pi^{h}(\tau)=\int_{0}^{\tau}\psi(\alpha^{h}(\theta),\theta)f(\theta)d\theta$,
so $\Pi^{h'}\left(\tau\right)=\Psi^{h}\left(\tau\right)f\left(\tau\right)$,
where $\Psi^{h}\left(\tau\right)=\psi(\alpha^{h}(\tau),\tau)$ is
the virtual surplus of the marginal included type. 
\begin{lem}
\label{lem:h-step-2}$\Pi^{h}$ is strictly quasiconcave on $[0,1]$
and admits a unique maximizer $\tau^{h}$.
\begin{itemize}
\item If $\Psi^{h}(1)\ge0$, then $\tau^{h}=1$.
\item If $\Psi^{h}(1)<0$, then $\tau^{h}\in(0,1)$ is the unique solution
to $\Psi^{h}(\tau)=0$. 
\end{itemize}
\end{lem}
The seller's optimal mechanism follows from Lemma \ref{lem:h-step-1}
and Lemma \ref{lem:h-step-2}.
\begin{prop}
\label{Proposition:H-solution}Suppose $v(\theta)=v_{0}$ for all
$\theta\in\Theta$ and $c>0$. The seller's optimal mechanism is unique: 
\begin{itemize}
\item The inclusion cutoff is $\tau^{h}$ from Lemma \ref{lem:h-step-2}.
If $v_{0}\ge c$, then $\tau^{h}=1$. 
\item The allocation on $\left[0,\tau^{h}\right]$ is $\alpha^{h}$ from
Lemma \ref{lem:h-step-1}. If $\beta^{h}<\tau^{h}$, then on $\left[\beta^{h},\tau^{h}\right]$,
product mismatch $\delta^{h}\left(\theta\right):=\theta-\alpha^{h}\left(\theta\right)$
is decreasing in $\theta$, strictly so for $\theta<\theta^{*}$ (if
$\theta^{*}$ exists).
\item The indirect utility schedule is $u\left(\theta\right)=\int_{\theta}^{\tau^{h}}km'\left(s-\alpha^{h}\left(s\right)\right)ds$.
It is decreasing in $\theta$, strictly so for $\theta<\theta^{*}$
(if $\theta^{*}$ exists). 
\item The transfer schedule is $t\left(\theta\right)=v_{0}-km\left(\theta-\alpha^{h}\left(\theta\right)\right)-u\left(\theta\right).$
$t\left(\theta\right)=v_{0}$ for $\theta\ge\theta^{*}$ (if $\theta^{*}$
exists), and $t\left(\theta\right)<v_{0}$ otherwise.
\end{itemize}
\end{prop}
Proposition \ref{Proposition:H-solution} completes the characterization
of the seller\textquoteright s optimal mechanism. The seller includes
every type whose virtual surplus is positive. Among the types receiving
non-basic products, product fit improves with type, though utility
declines. Types $\theta\ge\theta^{*}$ receive their ideal product
priced at their full valuation $v_{0}$, and so earn zero utility.
The next proposition provides comparative statics. 
\begin{prop}
\label{Proposition:h-CompStatics}Suppose $v(\theta)=v_{0}$ for all
$\theta\in\Theta$ and $c>0$. 
\begin{enumerate}
\item (On coverage) $\tau^{h}$ is weakly decreasing in $c$, weakly decreasing
in $k$, and weakly increasing in $v_{0}$. The increase in $v_{0}$
is strict if and only if $\tau^{h}<1$. Locally, the decrease in $c$
is strict if and only if $\beta^{h}<\tau^{h}<1$; and the decrease
in $k$ is strict if and only if $\tau^{h}<1$, and $\tau^{h}<\theta^{*}$
if $\theta^{*}$ exists.
\item (On mismatch) For each fixed $\theta$, $\delta^{h}(\theta)$ is weakly
increasing in $c/k$ and independent of $v_{0}$. Locally, the increase
is strict if and only if $\beta^{h}<\theta$, and $\theta<\theta^{*}$
if $\theta^{*}$ exists.
\item (On profit) The seller's profit $\Pi^{h}\left(\tau^{h}\right)$ is
weakly decreasing in $c$, strictly decreasing in $k$, and strictly
increasing in $v_{0}$. The decrease in $c$ is strict if and only
if $\tau^{h}>\beta^{h}$.
\end{enumerate}
\end{prop}
The inclusion set is $\left[0,\tau^{h}\right]$, so market coverage
is increasing in $\tau^{h}$. A higher $c$ shrinks coverage when
the marginal type receives a non-basic product, raising the cost of
serving him; if type $\tau^{h}$ is pooled at the costless basic product,
$c$ has no marginal effect. A higher $k$ shrinks coverage when the
marginal type is receives a mismatched product, reducing the surplus
from serving him; if $\tau^{h}\ge\theta^{*}$, he receives his ideal
product and incurs no mismatch cost, so $k$ has no marginal effect.
Therefore, $c$ is more likely to have effect when coverage is high,
whereas $k$ is more likely when coverage is low. A higher $v_{0}$
expands coverage simply by lifting every included type's willingness
to pay.

Product mismatch $\delta^{h}$ depends only on the ratio $c/k$, not
on $v_{0}$. This follows from the first-order condition of the virtual
surplus in (\ref{eq:delta-h}) in the appendix. Among types receiving
mismatched non-basic products, mismatch strictly increases when the
ideal product becomes more expensive to produce or when horizontal
fit becomes less important.

The profit effects of $v_{0}$ and $c$ are straightforward: $v_{0}$
raises willingness to pay and increases profit, while $c$ raises
the cost of non-basic products and lowers profit when the seller sells
some non-basic products. The effect of $k$ operates through mismatch.
When horizontal taste differentiation is stronger, any given mismatch
requires a larger price concession. The seller therefore reduces mismatch
and supplies products closer to buyers' ideal points, which are costlier
to supply. Profit thus falls with $k$. 

\section{\protect\label{Section:Vertical}Correlated Willingness to Pay and
Product Fit}

In this section, I assume $v\left(\theta\right)=\theta$, so that
a buyer's vertical value now also varies with his type. This captures
the correlation in the applications discussed in Section \ref{Subsection:Applications}.
Throughout this section, I restrict attention to $c<1$; otherwise,
every buyer's valuation of his ideal is less than its production cost. 

As in Section \ref{Section:PureHorizontal}, I work with direct mechanisms,
consisting of an allocation rule $\alpha$ and a transfer rule $t$.
Continue to let $\Theta_{I}$ denote the inclusion set. The only change
in the setup from Section \ref{Section:PureHorizontal} is that relative
to (\ref{eq:u-h}), an included type $\theta$'s utility under truth-telling
is now 
\begin{equation}
u\left(\theta\right)=\theta-kM\left(\alpha\left(\theta\right),\theta\right)-t\left(\theta\right).\label{eq:u-vertical}
\end{equation}
An IC mechanism must now satisfy
\begin{equation}
u'\left(\theta\right)=1-kM_{2}\left(\alpha\left(\theta\right),\theta\right)\ \ \ \ \text{a.e. on }\Theta_{I}.\label{eq:u'-v}
\end{equation}
In Section \ref{Section:PureHorizontal}, the envelope condition in
(\ref{eq:u'-horizontal}) implies that the sign of $u'$ is fully
determined by the sign of $M_{2}$, which depends only on whether
mismatch is upward or downward. Here, (\ref{eq:u'-v}) has an additional
``$+1$'' term, so the sign of $u'$ depends further on the magnitude
of $kM_{2}$ relative to $1$, which is the marginal vertical valuation.
\begin{lem}
\label{lem:v-CutoffMech}Suppose $v\left(\theta\right)=\theta$. In
any optimal mechanism, there exists $\tau\in\Theta$ such that $\Theta_{I}=\left[\tau,1\right]$.
Moreover, fixing an inclusion set $\left[\tau,1\right]$, an IR mechanism
is IC if and only if the envelope condition (\ref{eq:u'-v}) holds
on $\left[\tau,1\right]$, $\alpha$ is nondecreasing on $\left[\tau,1\right]$,
and if $\tau>0$, the cutoff type satisfies $u\left(\tau\right)=0$. 
\end{lem}
Lemma \ref{lem:v-CutoffMech} shows that the direction of exclusion
is reversed relative to Section \ref{Section:PureHorizontal}. Higher
types still prefer products that are costlier to produce, but they
now also have higher willingness to pay. Thus, serving a high type
no longer requires assigning him his ideal (or nearly ideal), high-cost
product. The seller can instead allocate him a lower-cost product,
let him bear some mismatch, and still profit, because his vertical
value is high. This changes the role of low-cost products. When buyers
differ only in horizontal tastes, a low-cost product is allocated
to the low type for whom it is a close match. With $v\left(\theta\right)=\theta$,
the same product may instead be more profitably allocated to a higher
type, despite the greater mismatch, because the higher type has greater
willingness to pay. Higher types therefore crowd out lower types even
at the low end of the product line. Lemma \ref{lem:v-CutoffMech}
formalizes this: if the seller excludes any buyers, she excludes low
types and serves an upper interval $\left[\tau,1\right]$.

Since the optimal inclusion set is still determined by a single cutoff,
I analyze the seller's problem in two steps, paralleling the approach
in Section \ref{Section:PureHorizontal}. Step 1 fixes $\tau$ and
characterizes the optimal mechanism for selling to types in $\left[\tau,1\right]$;
Step 2 then optimizes over $\tau$.

\subsection{Step 1: Optimal Mechanism for $\left[\tau,1\right]$}

Fix $\tau>0$, so $u\left(\tau\right)=0$ by Lemma \ref{lem:v-CutoffMech}.\footnote{\label{fn:=00005Ctau=00003D0}As in Section \ref{Section:PureHorizontal},
to simplify the exposition, I exclude the full coverage $\tau=0$,
under which the value of $u\left(\tau\right)$ is not pinned down
by Lemma \ref{lem:v-CutoffMech}. The solution characterization in
Corollary \ref{Corollary:V-Fixed-tau} includes the $\tau=0$ case.} The fixed-$\tau$ problem here is 
\begin{align}
\Pi^{v}\left(\tau\right) & =\underset{\alpha,u\text{ s.t. IC, IR}}{\max}\ \int_{\tau}^{1}\left[\theta-kM\left(\alpha\left(\theta\right),\theta\right)-c\alpha\left(\theta\right)-u\left(\theta\right)\right]f\left(\theta\right)d\theta,\tag{\ensuremath{\mathcal{P}^{v}\left(\tau\right)}}\label{eq:Program-v}\\
 & =\max_{\alpha\ \text{nondecreasing}}\int_{\tau}^{1}\psi^{v}(\alpha\left(\theta\right),\theta)f(\theta)d\theta\ \ \ \text{ s.t. }\ \ \ u\left(\theta\right)\ge0\ \ \forall\theta\in\left[\tau,1\right],
\end{align}
where $u\left(\theta\right)=\int_{\tau}^{\theta}\left[1-kM_{2}\left(\alpha\left(s\right),s\right)\right]ds$
by (\ref{eq:u'-v}), and $\psi^{v}$ is the virtual surplus function:
\begin{equation}
\psi^{v}\left(y,\theta\right):=\theta-kM\left(y,\theta\right)-cy-\left[1-kM_{2}\left(y,\theta\right)\right]\frac{1-F\left(\theta\right)}{f\left(\theta\right)}.\label{eq:V-VS-R}
\end{equation}

\subsubsection{\protect\label{Subsection:Relaxed-Program}A Relaxed Program: Ignore
IR}

To illustrate how the additional vertical valuation dimension affects
the screening forces and isolate these forces from issues related
to IR, I first consider the \emph{relaxed fixed-$\tau$ problem} that
ignores IR in \ref{eq:Program-v}. Define $d^{R}\left(\theta\right)$
by
\begin{equation}
m'\left(d^{R}\left(\theta\right)\right)-m''\left(d^{R}\left(\theta\right)\right)\frac{1-F\left(\theta\right)}{f\left(\theta\right)}=\frac{c}{k}.\label{eq:delta-R}
\end{equation}
(\ref{eq:delta-R}) is the first-order condition of $\psi^{v}$ with
respect to the mismatch $d=\theta-y$. Under Assumptions \ref{Assumption:F}
and \ref{Assumption:m-1}, (\ref{eq:delta-R}) has a unique solution
$d^{R}\left(\theta\right)$ for every $\theta$, and $d^{R}$ is strictly
decreasing in $\theta$. Let $b^{R}:=\sup\left\{ \theta\in\Theta|d^{R}\left(\theta\right)\ge\theta\right\} $.
\begin{lem}
\label{lem:Relaxed-program}The solution to the relaxed fixed-$\tau$
problem is uniquely $\alpha=a^{R}$, with
\begin{equation}
a^{R}\left(\theta\right)=\begin{cases}
0 & \text{ if }\theta\le b^{R}\\
\theta-d^{R}\left(\theta\right) & \text{ if }\theta>b^{R}
\end{cases},\label{eq:yR}
\end{equation}
Moreover, $b^{R}<1$, and $a^{R}\left(\theta\right)\le\alpha^{FB}\left(\theta\right)$
for all $\theta$, with strict inequality if $\theta<1$ and $\alpha^{FB}\left(\theta\right)>0$.
\end{lem}
In Section \ref{Section:PureHorizontal}, downward mismatch raises
the information rent conceded to lower types, so the seller induces
less mismatch than the first best, though never reversing the direction
of mismatch. With $v\left(\theta\right)=\theta$, the relevant incentive
constraint runs in the opposite direction: higher types must be deterred
from mimicking lower types, as in the argument for Lemma \ref{lem:v-CutoffMech}.
Downward mismatch now \emph{lowers} the rent conceded to higher types
because a lower-indexed product assigned to type $\theta$ is a worse
fit for types above $\theta$. The allocation $a^{R}$---hereafter
the ``\emph{relaxed allocation}''---therefore induces \emph{more}
downward mismatch than both the first-best and the Section \ref{Section:PureHorizontal}
optimum.

Since mismatch is always downward, $kM_{2}\left(a^{R}\left(\theta\right),\theta\right)$
is positive. Because (\ref{eq:u'-v}) also contains the marginal vertical-value
term ``$+1$'', the sign of $u'$ is then indeterminate. The IC-implied
utility can therefore be nonmonotone. In turn, even with $u\left(\tau\right)=0$,
IR may be violated for interior included types under $a^{R}$. I study
this issue next.

\subsubsection{\protect\label{Subsection:uR-shape}Overview of the Solution to \ref{eq:Program-v}}

I now return to the fixed-$\tau$ problem \ref{eq:Program-v} with
its IR constraints. This subsection provides a heuristic overview
of the solution, using the relaxed allocation $a^{R}$ as the building
block; the formal characterization is in the next subsection. 

If the utility schedule induced by $a^{R}$ is nonnegative for every
$\theta>\tau$, then $a^{R}$ is feasible for \ref{eq:Program-v}
and hence solves it. To see when this holds, define the relaxed mismatch
$\delta^{R}(\theta):=\theta-a^{R}(\theta).$ The \emph{relaxed utility
schedule} induced by $a^{R}$ is

\[
u^{R}(\theta;\tau):=\int_{\tau}^{\theta}\left[1-km'\left(\delta^{R}\left(s\right)\right)\right]ds,
\]
so its slope is
\begin{equation}
\partial_{\theta}u^{R}\left(\theta;\tau\right)=1-km'\left(\delta^{R}\left(\theta\right)\right)\underbrace{=}_{\text{by Lemma \ref{lem:Relaxed-program}}}\begin{cases}
1-km'\left(\theta\right) & \text{ if }\theta\le b^{R}\\
1-km'\left(d^{R}\left(\theta\right)\right) & \text{ if }\theta\ge b^{R}
\end{cases}.\label{eq:u-R-slope}
\end{equation}
Below $b^{R}$, where types are pooled at the basic product, mismatch
equals $\theta$, so the slope $\partial_{\theta}u^{R}$ is decreasing.
Above $b^{R}$, mismatch is $d^{R}\left(\theta\right)$, which is
decreasing in $\theta$, so the slope $\partial_{\theta}u^{R}$ is
instead increasing. Hence, $u^{R}$ is concave below $b^{R}$ and
convex above.

Suppose $\tau<b^{R}$, so the concave region $[\tau,b^{R})$ is nonempty.
Define the potential stationary point of this concave region by $\eta$:
\begin{equation}
1-km'\left(\eta\right)=0\ \ \ \iff\ \ \ \eta:=(m')^{-1}\left(1/k\right),\label{eq:eta}
\end{equation}
On $[\tau,b^{R})$, $\partial_{\theta}u^{R}$ is positive for $\theta<\eta$.
Hence, if $b^{R}\le\eta$, then $\partial_{\theta}u^{R}$ is positive
on the entire interval $[\tau,b^{R})$. Since $\partial_{\theta}u^{R}$
must remain increasing above $b^{R}$, $u^{R}(\tau;\tau)=0$ implies
that IR holds on $\left[\tau,1\right]$, so $a^{R}$ solves \ref{eq:Program-v}.

\bigskip

\begin{figure}[h]
\caption{\protect\label{Figure:uR}Shape of nonmonotone $u^{R}\left(\theta;\tau\right)$ }

\begin{tikzpicture}[xscale = 13,yscale=22]


\node   at (0.21,0.28) { \underline{ $\theta < b^R  $} };
\node at (0.2,0.24) { $a^R(\theta) = 0  $ ; $\delta^R(\theta) = \theta  $ };

\node   at (0.73,0.28) { \underline{ $\theta > b^R  $} };
\node at (0.75,0.24) { $a^R(\theta) > 0  $ ; $\delta^R(\theta) = d^R (\theta)  $ };

\draw [thick, ->] (-0.05,0)--(1.0,0); 
\draw [thick, ->] (-0.05,0)--(-0.05,0.32);
\node [right] at (1.0,0) {\footnotesize{$\theta$}};
\node [right] at (-0.05,0.32) {\footnotesize{$u^{R}(\theta;\tau)$}};

\draw[line width=0.5mm, blue,  domain = 0:0.5] 
plot(\x, {  \x - 1.5*(\x)^2  }); 
\node [below] at (0,0) {\footnotesize{$\tau_6$}};

\draw[line width=0.5mm, red,  domain = 0.5:1] 
plot(\x, { -2*\x+ 1.5*(\x)^2 +0.75  });

\draw [-, dotted, line width=0.2mm] (0.33,-0.1)--(0.33,0.16); 
\node [below] at (0.33,-0.1) {\footnotesize{$\eta $}};

\draw [-, dashed, line width=0.2mm] (0.5,-0.1)--(0.5,0.3); 
\node [below] at (0.5,-0.1) {\footnotesize{$b^R$}};

\draw [-, dotted, line width=0.2mm] (0.67,-0.1)--(0.67,0.09); 
\node [below] at (0.67,-0.1) {\footnotesize{$\gamma$}};



\draw[line width=0.5mm, red,   domain = 0.5:1] 
plot(\x, { -2*\x+ 1.5*(\x)^2 +0.75  - 0.082});

\draw[line width=0.5mm, blue,  domain = 0.095:0.5] 
plot(\x, {  \x - 1.5*(\x)^2  - 0.082 }); 
\node [below] at (0.1,0) {\footnotesize{$\tau_{5}$}};



\draw[line width=0.5mm, red,  domain = 0.5:1] 
plot(\x, { -2*\x+ 1.5*(\x)^2 +0.75  - 0.14});

\draw[line width=0.5mm, blue, domain = 0.2:0.5] 
plot(\x, {  \x - 1.5*(\x)^2  - 0.14 }); 

\node [below] at (0.2,0) {\footnotesize{$\tau_4$}};


\draw[line width=0.5mm, red,  domain = 0.5:1] 
plot(\x, { -2*\x+ 1.5*(\x)^2 +0.75  - 0.16});

\draw[line width=0.5mm, blue, domain = 0.4:0.5] 
plot(\x, {  \x - 1.5*(\x)^2  - 0.16 }); 
\node [below] at (0.4,0) {\footnotesize{$\tau_3$}};


\draw[line width=0.5mm, red,  domain = 0.55:1] 
plot(\x, { -2*\x+ 1.5*(\x)^2 +0.645});

\node [below] at (0.55,0) {\footnotesize{$\tau_2$}};

\draw[line width=0.5mm, red,  domain = 0.825:1] 
plot(\x, { -2*\x+ 1.5*(\x)^2 +0.63});

\node [below] at (0.82,0) {\footnotesize{$\tau_1$}};

\end{tikzpicture}
\end{figure}
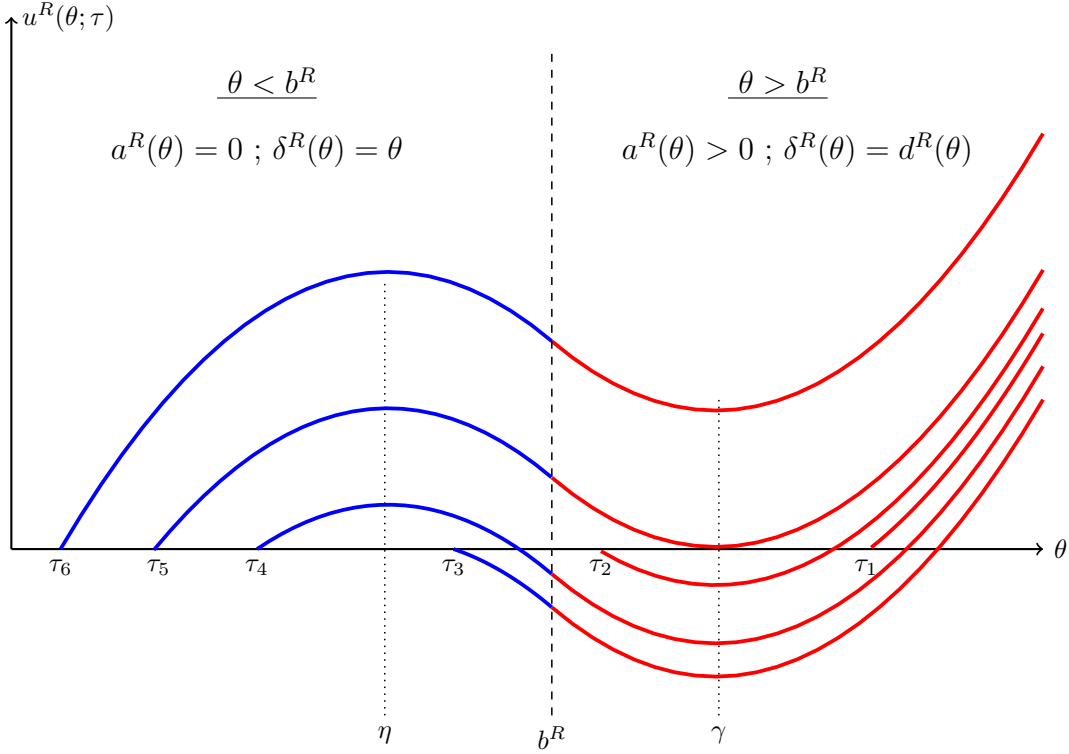

The less straightforward case is $\eta<b^{R}$, depicted in Figure
\ref{Figure:uR}. Now, $\eta$ is a local maximum point on $u^{R}$.
On the convex region $(b^{R},1]$, any stationary point must be a
local minimum. Let $\gamma$ denote this point:

\begin{equation}
1-km'\left(d^{R}\left(\gamma\right)\right)=0.\label{eq:gamma}
\end{equation}
Since $\gamma<1$,\footnote{This is because $m'\left(d^{R}\left(\gamma\right)\right)=1/k>c/k=m'\left(d^{R}\left(1\right)\right)$.}
$u^{R}$ is increasing near $\theta=1$. Hence, for $\eta<b^{R}$,
$u^{R}$ increases below $\eta$, decreases on $(\eta,\gamma)$, and
increases again above $\gamma$, as shown in Figure \ref{Figure:uR}.
The blue part is the concave region in which types are pooled at the
basic product, so $\delta^{R}(\theta)=\theta$; the red part is the
convex region in which types receive non-basic products, so $\delta^{R}(\theta)=d^{R}(\theta)$. 

The figure also shows how the cutoff $\tau$ matters. Because the
slope $\partial_{\theta}u^{R}$ does not depend on $\tau$, varying
$\tau$ shifts $u^{R}$ vertically to satisfy $u^{R}=0$ at type $\tau$.
If $\tau\ge\gamma$ (as with $\tau_{1}$), the included types lie
entirely on the increasing part, so IR holds. If $\tau\in\left(\eta,\gamma\right)$
(as with $\tau_{2}$ and $\tau_{3}$), $u^{R}$ decreases immediately
from zero at $\tau$, so IR must fail for types immediately above
$\tau$. If $\tau<\eta$ (as with $\tau_{4}$, $\tau_{5}$ and $\tau_{6}$),
$u^{R}$ first rises, falls over $\left(\eta,\gamma\right)$, then
rises on $\left(\gamma,1\right)$; IR fails only if the dip falls
below zero---as for $\tau_{4}$, but not $\tau_{5}$ and $\tau_{6}$. 

Thus, when $\eta<b^{R}$, $u^{R}$ is potentially nonmonotone, but
it violates IR only if the cutoff is low enough that type $\gamma$
is included, yet not so low that the earlier rise offsets the subsequent
dip. Equivalently, 

\begin{equation}
a^{R}\text{ fails IR }\ \ \ \ \ \ \iff\ \ \ \ \ \ \gamma\in(\tau,1)\quad\text{and}\quad u^{R}(\gamma;\tau)<0.\label{eq:aR-IR-failure}
\end{equation}
Note that the IR failure for cutoff $\tau_{4}$ is due to interior
types and presents new difficulties. It arises because the relaxed
allocation $a^{R}$ uses downward mismatch to reduce the rent conceded
to higher types. When this distortion is strong enough, the IC-implied
utility schedule can fall below zero.

\bigskip

When the relaxed allocation violates IR, the solution to \ref{eq:Program-v}
(formally characterized in the next subsection) can be understood
as raising the relaxed utility schedule $u^{R}\left(\cdot;\tau\right)$
just enough to keep it nonnegative. Since its \emph{slope} corresponds
one-to-one to $a^{R}$, which maximizes the virtual surplus, the optimal
utility schedule $u^{\tau}$ departs from that slope only where needed
to keep utility nonnegative. Figure \ref{Figure:problem-tau} shows
the two relevant cases.

\begin{figure}[h]
\caption{\protect\label{Figure:problem-tau}Correction when $a^{R}$ violates
IR.}

\begin{tikzpicture}[xscale = 6.7,yscale=22]

\node at (0.6,0.27) {(i): $\tau \in [\eta,\gamma)$  };

\draw [thick, ->] (0.25,0)--(1.2,0); 
\draw [thick, ->] (0.25,0)--(0.25,0.2);


\draw[line width=0.6mm, black, dashed, domain = 0.4:0.5] 
plot(\x, {  \x - 1.5*(\x)^2  - 0.16 }); 
\node [below] at (0.4,0) {\footnotesize{$\tau$}};
\node [right] at (0.99,0.075) {\footnotesize{$u^{R}(\theta ;\tau)$}};

\node [red] at (0.8,0.14) {\footnotesize{$a^R(\theta)$}};

\draw[line width=0.6mm, black, dashed,  domain = 0.5:1] 
plot(\x, { -2*\x+ 1.5*(\x)^2 +0.75  - 0.16});

\draw [-, dotted, line width=0.2mm] (0.67,-0.1)--(0.67,0.2); 
\node [below] at (0.67,-0.1) {\footnotesize{$\gamma$}};

\draw [line width=0.5mm, red,  -] (0.4,0)--(0.67,0); 
\draw[line width=0.5mm, red,   domain = 0.67:1] 
plot(\x, { -2*\x+ 1.5*(\x)^2 +0.75  - 0.16 + 0.07665});

\node [right, red] at (1.01,0.17) {\footnotesize{$u^{\tau}(\theta)$}};

\node [red] at (0.55,0.14) {\footnotesize{$a^0(\theta)$}};

\node [white] at (0.5,-0.13) {\footnotesize{$y^{\tau}(\theta)=y^0(\theta)$}};

\end{tikzpicture}$\ \ \ \ \ \ \ \ \ \ \ \ \ \ $\begin{tikzpicture}[xscale = 7,yscale=18]

\node at (0.43,0.32) {(ii): $\tau \in (\underline{\tau},\eta)$  };

\draw [thick, ->] (-0.1,0)--(1.2,0); 
\draw [thick, ->] (-0.1,0)--(-0.1,0.25);

\draw[line width=0.6mm, black, dashed,  domain = 0:0.5] 
plot(\x, {  \x - 2*(\x)^2  }); 
\node [below] at (0,0) {\footnotesize{$\tau$}};

\draw [-, dotted, line width=0.2mm] (0.25,.125)--(0.25,-0.15); 
\node [below] at (0.25,-0.15) {\footnotesize{$\eta$}};

\draw[line width=0.6mm, black, dashed,  domain = 0.5:1.05] 
plot(\x, { -3*\x+ 2*(\x)^2 +1  });
\node [right] at (1.05,0.06) {\footnotesize{$u^{R}(\theta;\tau)$}};

\draw[line width=0.4mm, red,    domain = 0.75:1.05] 
plot(\x, { -3*\x+ 2*(\x)^2 +1+0.125  });

\node [right, red] at (1.05,0.2) {\footnotesize{$u^{\tau}(\theta)$}}; 
\draw [-, dotted, line width=0.2mm] (0.75,0.2)--(0.75,-0.15); 
\node [below] at (0.75,-0.15) {\footnotesize{$\gamma$}};
\node [red] at (0.85,0.15) {\footnotesize{$a^R(\theta)$}};

\draw[line width=0.4mm, red,    domain = 0.4268:0.6036] 
plot(\x, {  2*(\x)^2  -2.4144*\x +0.72863});
\draw [line width=0.4mm, red,  -] (0.6036,0)--(0.75,0); 
\draw [-, dotted, line width=0.2mm] (0.6036,-0.15)--(0.6036,0.2); 
\node [red, below] at (0.6036,-0.15) {\footnotesize{$\phi^{\tau}$}};
\node [red] at (0.68,0.15) {\footnotesize{$a^0(\theta)$}};

\draw[line width=0.4mm, red,   domain = 0:0.4268] 
plot(\x, {  \x - 2*(\x)^2  });
\draw [-, dotted, line width=0.2mm] (0.4268,0.06)--(0.4268,-0.15); 
\node [below, red] at (0.4268,-0.15) {\footnotesize{$\beta^{\tau}$}};

\node [red] at (0.4,0.15) {\footnotesize{$a(\theta;\phi^{\tau})$}};

\draw [-, dotted, line width=0.2mm] (0.5,0)--(0.5,-0.10);
\node [below] at (0.5,-0.10) {\footnotesize{$\zeta$}};

\end{tikzpicture}
\end{figure}
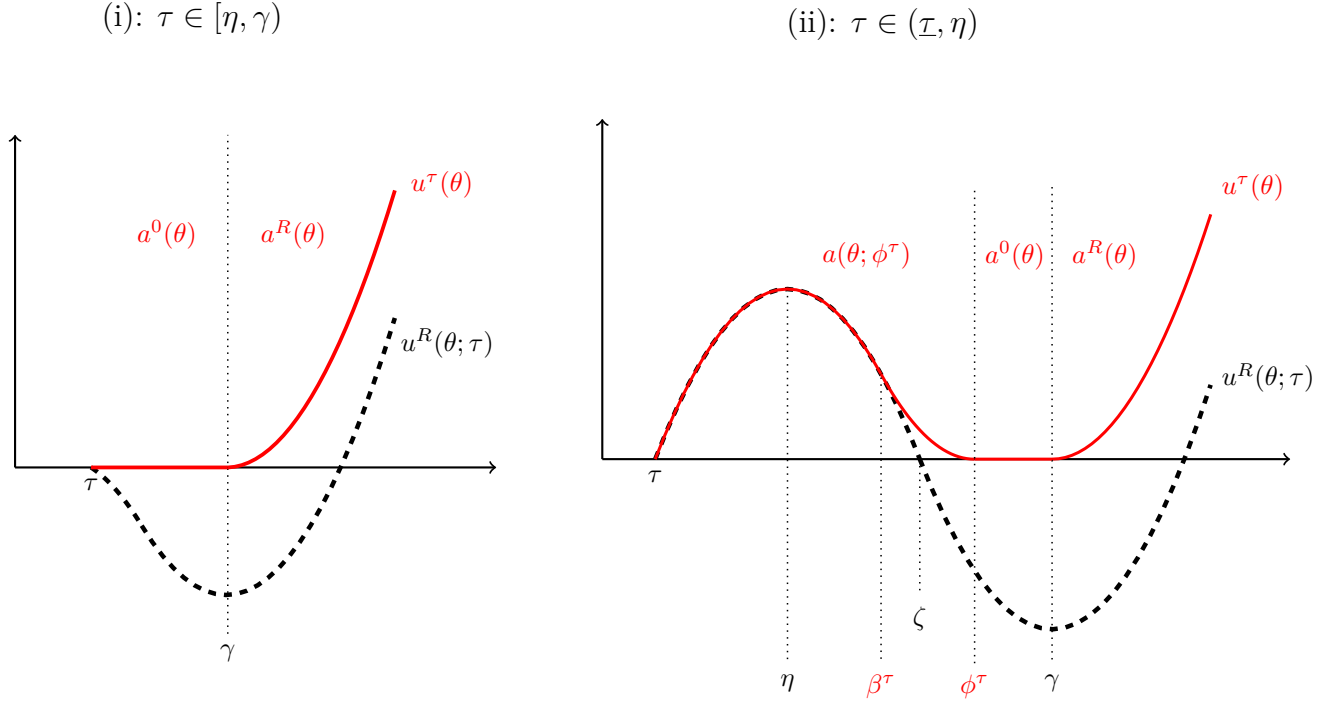

First, consider $\tau\in[\eta,\gamma)$, shown in Figure \ref{Figure:problem-tau}(i).
This case corresponds to cutoffs such as $\tau_{2}$ and $\tau_{3}$
in Figure \ref{Figure:uR}, where utility decreases immediately after
the cutoff. On $\left(\tau,\gamma\right)$, the slope of $u^{R}\left(\cdot;\tau\right)$
is negative, so the minimal correction raises it to zero: $u^{\tau}$
is flat at $0$ on $\left[\tau,\gamma\right]$. Above $\gamma$, where
the slope of $u^{R}$ is positive, retaining that slope now keeps
IR satisfied, since utility is already zero at $\gamma$. The allocation
schedule $\alpha^{\tau}$ associated with $u^{\tau}$ is pinned down
by its slope through the envelope condition in (\ref{eq:u'-v}). On
$\left[\gamma,1\right]$, $\alpha^{\tau}$ remains the relaxed allocation,
since $u^{\tau}$ has the same slope as $u^{R}\left(\cdot;\tau\right)$.
For types below $\gamma$, $\alpha^{\tau}\left(\theta\right)$ is
the allocation that makes the slope equal to zero:
\begin{equation}
1-kM_{2}\left(\alpha^{\tau}\left(\theta\right),\theta\right)=0\ \ \iff\ \ \alpha^{\tau}\left(\theta\right)=a^{0}\left(\theta\right):=\theta-(m')^{-1}\left(1/k\right).\label{eq:y0}
\end{equation}
Mismatch remains downward because $M_{2}\left(a^{0}\left(\theta\right),\theta\right)=1/k>0$,
so $a^{0}\left(\theta\right)<\theta$. Raising the slope therefore
means reducing mismatch: on $\left(\tau,\gamma\right)$, buyers receive
products closer to their ideal under $a^{0}$ than under $a^{R}$.

The second case is $\tau<\eta$ and $u^{R}\left(\gamma;\tau\right)<0$,
shown in Figure \ref{Figure:problem-tau}(ii). This case corresponds
to cutoffs such as $\tau_{4}$ in Figure \ref{Figure:uR}. The idea
is again to raise $u^{R}\left(\cdot;\tau\right)$ above zero with
the least distortion to its slope. In the figure, $\zeta$ is the
first type whose IR constraint fails under the relaxed allocation.
The correction, however, begins \emph{before} $\zeta$, at type $\beta^{\tau}$,
where utility is still positive. The reason is that $u^{R}$ is steep
near $\zeta$, so starting the correction only at $\zeta$ would require
forcing the slope to zero at exactly where $u^{R}$ is falling rapidly,
creating a large distortion. By starting at $\beta^{\tau}$, the seller
reduces downward mismatch gradually relative to $a^{R}$ until type
$\phi^{\tau}$, so utility falls more slowly than under $u^{R}$ and
reaches zero only at type $\phi^{\tau}>\zeta$. From $\phi^{\tau}$
to $\gamma$, utility is held at zero with $\alpha^{\tau}\left(\theta\right)=a^{0}\left(\theta\right)$
in (\ref{eq:y0}). Above $\gamma$, the slope of $u^{R}$ is positive
again, so the allocation returns to $a^{R}$. The types $\beta^{\tau}$
and $\phi^{\tau}$ are chosen to optimally smooth out the distortion
to the slope of $u^{R}$ over $\left(\beta^{\tau},\gamma\right)$,
and are formally derived in the next subsection.

The economic takeaway is the same in both cases. Ignoring IR, the
relaxed allocation $a^{R}$ maximizes the virtual surplus, using downward
mismatch to reduce buyers' information rent. When IR fails, the seller
must reduce that mismatch and concede more rent. However, match improvement
need not track IR failure of $a^{R}$. As Figure \ref{Figure:problem-tau}(ii)
shows, types just below $\zeta$ receive better-match products and
higher utility, even though their IR already holds under $a^{R}$.
At the same time, types just above $\gamma$---whose IR fails---receive
the same allocation as under $a^{R}$; their utility is restored through
lower prices rather than better product fit.

\bigskip

I close this subsection by noting when the nonmonotonicity problem
is more likely to arise.
\begin{lem}
\label{Lemma:-eta-gamma}For all $c$ and $k$, one of the following
three cases must hold: (i) $\eta>b^{R}>\gamma$, (ii) $\eta=b^{R}=\gamma$,
or (iii) $\eta<b^{R}<\gamma$. Moreover, $\gamma-\eta$ is strictly
increasing in both $k$ and $c$. 
\end{lem}
The nonmonotonicity problem can only rise when $\eta<b^{R}$. By Lemma
\ref{Lemma:-eta-gamma}, this is equivalent to $\eta<\gamma$, which
becomes more likely as $k$ or $c$ increases. Intuitively, the relaxed
problem satisfies IR only at $\tau$ and lets IC determine the utilities
of higher types. As in vertical screening (\citealp{Mussa_Rosen_1978(JET)}),
the seller distorts each type's allocation to reduce information rent.
There, a higher type values any given quality more, so the utility
is always monotone in type, regardless of the quality distortion.
Here, the ``quality'' when moving along the product line is horizontally
differentiated, so a higher type can bear enough extra mismatch disutility
to outweigh his vertical advantage. A higher $c$ pushes the relaxed
allocation toward lower-indexed products, increasing downward mismatch.
A higher $k$ magnifies the disutility of any given mismatch. Both
forces make it more likely that the mismatch disutility dominates
the vertical-value term, causing the IC-implied utility schedule to
decrease.

\subsubsection{Characterization of the Solution to \ref{eq:Program-v}}

Let $\lambda$ be a nonnegative multiplier measure for the participation
constraints $u(\theta)\geq0$, and define its cumulative multiplier
by $\Lambda(\theta):=\lambda\left(\left[\theta,1\right]\right)$.
The Lagrangian is 
\[
\int_{\tau}^{1}\left[\theta-kM\left(\alpha\left(\theta\right),\theta\right)-c\alpha\left(\theta\right)-u\left(\theta\right)\right]f\left(\theta\right)d\theta\ +\ \int_{\tau}^{1}u(\theta)\lambda(d\theta).
\]
Substituting in the envelope condition (\ref{eq:u'-v}) and $u\left(\tau\right)=0$,
and integrating by parts gives
\begin{equation}
\mathcal{L}\left(\alpha,\Lambda\right):=\int_{\tau}^{1}\tilde{\psi}^{v}\left(\alpha\left(\theta\right),\theta;\Lambda\left(\theta\right)\right)f\left(\theta\right)d\theta\label{eq:Lagrangian}
\end{equation}
where 
\begin{equation}
\tilde{\psi}^{v}(y,\theta;L):=\theta-kM(y,\theta)-cy-\bigl[1-kM_{2}(y,\theta)\bigr]\frac{1-F(\theta)-L}{f(\theta)}.\label{eq:V-Adjusted-VS}
\end{equation}
I call $\tilde{\psi}^{v}$ the \emph{adjusted virtual surplus}. Compared
to the relaxed virtual surplus $\psi^{v}$ in (\ref{eq:V-VS-R}),
the inverse hazard rate $\left(1-F\left(\theta\right)\right)/f\left(\theta\right)$
is replaced by the lower \emph{effective inverse hazard rate} $\left(1-F\left(\theta\right)-\Lambda\left(\theta\right)\right)/f\left(\theta\right)$,
determined by the multiplier $\Lambda\left(\theta\right)$, which
takes into account the IR of types above $\theta$. When $\Lambda\left(\theta\right)=0$,
$\tilde{\psi}^{v}$ reduces to $\psi^{v}$ at $\theta$. 

I derive the solution to \ref{eq:Program-v} by constructing a pair
$\left(\alpha,\Lambda\right)$ such that $\alpha$ is a pointwise
maximizer of $\tilde{\psi}^{v}\left(\cdot,\theta;\Lambda\left(\theta\right)\right)$,
so $\alpha$ attains $\sup_{\alpha'}\mathcal{L}\left(\alpha',\Lambda\right)$.
By weak duality, this maximized Lagrangian is an upper bound on the
seller's profit for any feasible $\Lambda$. If the same $\alpha$
is also feasible for \ref{eq:Program-v} and satisfies complementary
slackness against $\Lambda$,\footnote{Since $\lambda$ is a nonnegative measure, $\Lambda$ is nonincreasing.
Complementary slackness is equivalent to $\Lambda$ being constant
on every interval on which $u>0$.} then this upper bound is attained, so $\alpha$ solves \ref{eq:Program-v}. 

Let $\alpha^{\tau}$ denote the allocation that solves \ref{eq:Program-v},
with induced utility schedule $u^{\tau}\left(\theta\right):=\int_{\tau}^{\theta}\left[1-kM_{2}\left(\alpha^{\tau}\left(s\right),s\right)\right]ds$.
The characterization has three cases. 

\bigskip

The first case is when $a^{R}$ is already feasible. Subsection \ref{Subsection:uR-shape}
showed that this holds whenever $\eta\ge\gamma$, and, when $\eta<\gamma$,
for cutoffs at which $u^{R}$ stays nonnegative. To delineate those
cutoffs, suppose $\eta<\gamma$, as in Figure \ref{Figure:uR}, and
define
\begin{equation}
\underline{\tau}:=\inf\{\tau\in[0,\eta]:u^{R}(\gamma;\tau)<0\}.\label{eq:tau_underline}
\end{equation}
In Figure \ref{Figure:uR}, $\underline{\tau}$ is the cutoff $\tau_{5}$,
at which the dip just reaches zero at $\gamma$. With $\underline{\tau}$
as defined, (\ref{eq:aR-IR-failure}) becomes as follows: $a^{R}$
fails IR if and only if $\tau\in\left(\underline{\tau},\gamma\right)$.
\begin{lem}
\label{lem:p-tau-case-1}Suppose either $\eta\ge\gamma$, or $\eta<\gamma$
and $\tau\notin(\underline{\tau},\gamma)$. Then the relaxed allocation
$a^{R}$ solves the fixed-$\tau$ problem \ref{eq:Program-v}: $\alpha^{\tau}\left(\theta\right)=a^{R}\left(\theta\right)$
for all $\theta\in\left[\tau,1\right]$. 
\end{lem}
Because $a^{R}$ already satisfies IR, the supporting multiplier is
$\Lambda\left(\theta\right)=0$ for all $\theta\in\left[\tau,1\right]$. 

\bigskip

The second case is $\eta<\gamma$ and $\tau\in[\eta,\gamma)$, corresponding
to cutoffs such as $\tau_{2}$ and $\tau_{3}$ in Figure \ref{Figure:uR},
with the optimal correction in Figure \ref{Figure:problem-tau}(i). 
\begin{lem}
\label{lem:p-tau-case-2}Suppose $\eta<\gamma$ and $\tau\in[\eta,\gamma)$.
Then
\begin{equation}
\alpha^{\tau}\left(\theta\right)=\begin{cases}
a^{0}\left(\theta\right)\ \text{ in (\ref{eq:y0})} & ,\text{ if }\theta\in\left[\tau,\gamma\right]\\
a^{R}\left(\theta\right) & ,\text{ if }\theta\in\left[\gamma,1\right]
\end{cases}.\label{eq:y-tau-case2}
\end{equation}
\end{lem}
On $\left[\gamma,1\right]$, the supporting multiplier is $\Lambda\left(\theta\right)=0$,
since the allocation coincides with $a^{R}$. On $[\tau,\gamma)$,
it must support $a^{0}$ as the pointwise maximizer of the adjusted
virtual surplus. Differentiating $\tilde{\psi}^{v}$ with respect
to the allocation yields

\begin{equation}
\frac{\partial\tilde{\psi}^{v}(y,\theta;\Lambda\left(\theta\right))}{\partial y}=-kM_{1}(y,\theta)-c+kM_{21}(y,\theta)\frac{1-F(\theta)-\Lambda\left(\theta\right)}{f(\theta)},\label{eq:V-FOC}
\end{equation}
so $\left(a^{0}\left(\theta\right),\Lambda\left(\theta\right)\right)$
satisfies the first-order condition if and only if
\begin{equation}
\frac{\partial\tilde{\psi}^{v}(y,\theta;\Lambda\left(\theta\right))}{\partial y}\Big|_{y=a^{0}\left(\theta\right)}=0\ \ \ \iff\ \ \ \Lambda\left(\theta\right)=L^{0}\left(\theta\right):=1-F\left(\theta\right)-f\left(\theta\right)\frac{1-F\left(\gamma\right)}{f\left(\gamma\right)}.\label{eq:Lambda-0}
\end{equation}
Therefore, the supporting multiplier for the case in Lemma \ref{lem:p-tau-case-2}
is

\begin{equation}
\Lambda(\theta)=\begin{cases}
L^{0}(\theta), & \theta\in[\tau,\gamma],\\
0, & \theta\in[\gamma,1].
\end{cases}\label{eq:Lambda-Case-2}
\end{equation}
When $\Lambda\left(\theta\right)=L^{0}\left(\theta\right)$, the binding
IR constraints reduce the effective hazard rate to the constant $(1-F(\gamma))/f(\gamma)$
on $[\tau,\gamma]$. This reduction is exactly what makes the zero-slope
allocation $a^{0}$ optimal there.

\bigskip

The third case is $\eta<\gamma$ and $\tau\in(\underline{\tau},\eta)$,
corresponding to a cutoff such as $\tau_{4}$ in Figure \ref{Figure:uR},
with the optimal correction in Figure \ref{Figure:problem-tau}(ii).
To locate $\phi^{\tau}$, parameterize the cumulative multiplier by
$\phi\in\left(\tau,\gamma\right)$:
\begin{equation}
L\left(\theta;\phi\right):=\begin{cases}
0 & \text{ if }\theta\in\left[\gamma,1\right]\\
L^{0}\left(\theta\right) & \text{ if }\theta\in\left[\phi,\gamma\right]\\
L^{0}\left(\phi\right) & \text{ if }\theta\in\left[\tau,\phi\right]
\end{cases}\label{eq:L(theta;phi)}
\end{equation}
Here $\phi$ is the first type whose IR binds. For $\theta\ge\phi$,
the solution is as in Lemma \ref{lem:p-tau-case-2}: allocation $a^{0}$
and multiplier $L^{0}$ on $\left[\phi,\gamma\right]$, and allocation
$a^{R}$ and multiplier $\Lambda=0$ above $\gamma$. For $\theta<\phi$,
IR is slack, so $\Lambda$ is constant at $L^{0}\left(\phi\right)$.
Define 
\[
a\left(\theta;\phi\right):=\underset{y\in\left[0,1\right]}{\arg\max}\ \tilde{\psi}^{v}\left(y,\theta;L^{0}\left(\phi\right)\right),
\]
the pointwise maximizer of the adjusted virtual surplus for $\theta<\phi$.
The proof of Lemma \ref{lem:p-tau-case-3} shows that $\tilde{\psi}^{v}\left(\cdot,\theta;L^{0}\left(\phi\right)\right)$
is strictly quasiconcave on $Y$, so $a\left(\theta;\phi\right)$
is well-defined. Let 
\[
U\left(\phi;\tau\right)=\int_{\tau}^{\phi}\left[1-km'\left(s-a\left(s;\phi\right)\right)\right]ds
\]
be type $\phi$'s utility from allocation $a\left(\cdot;\phi\right)$.
Since $\phi$ is the first type at which IR binds, it must satisfy
$U(\phi;\tau)=0$. The proof of Lemma \ref{lem:p-tau-case-3} shows
that a unique $\phi^{\tau}\in(\eta,\gamma)$ at which $U\left(\phi^{\tau};\tau\right)=0$,
thereby identifying $\phi^{\tau}$.
\begin{lem}
\label{lem:p-tau-case-3}Suppose $\eta<\gamma$ and $\tau\in\left(\underline{\tau},\eta\right)$.
Then
\begin{equation}
\alpha^{\tau}\left(\theta\right)=\begin{cases}
a\left(\theta;\phi^{\tau}\right) & ,\text{ if }\theta\in[\tau,\phi^{\tau})\\
a^{0}\left(\theta\right) & ,\text{ if }\theta\in\left[\phi^{\tau},\gamma\right]\\
a^{R}\left(\theta\right) & ,\text{ if }\theta\in\left[\gamma,1\right]
\end{cases}.\label{eq:alpha-tau-case-3}
\end{equation}
Furthermore, there exists $\beta^{\tau}\in\left(\tau,b^{R}\right)$
such that $a\left(\theta;\phi^{\tau}\right)=0$ for $\theta\le\beta^{\tau}$,
while for $\theta>\beta^{\tau}$, $a\left(\theta;\phi^{\tau}\right)$
is determined by the first-order condition 
\[
\frac{\partial\tilde{\psi}^{v}(y,\theta;L^{0}\left(\phi^{\tau}\right))}{\partial y}\Big|_{y=a\left(\theta;\phi^{\tau}\right)}=0
\]
\end{lem}
\bigskip

Under the convention $\phi^{\tau}=\tau$ for $\tau\in[\eta,\gamma$),
the schedule in (\ref{eq:alpha-tau-case-3}) subsumes (\ref{eq:y-tau-case2}),
so Lemma \ref{lem:p-tau-case-3} covers every case in which $a^{R}$
violates IR, while Lemma \ref{lem:p-tau-case-1} covers the rest.
I adopt this convention henceforth. The appendix also establishes
that the solution to $P^{v}(\tau)$ is unique for every $\tau$ (see
the proof of Corollary \ref{Corollary:V-Fixed-tau}). The following
collects the characterization:
\begin{cor}
\label{Corollary:V-Fixed-tau}Fix $\tau\ge0$. The fixed-$\tau$ problem
\ref{eq:Program-v} has a unique solution. 
\begin{itemize}
\item If $\eta\ge\gamma$, or $\eta<\gamma$ and $\tau\in[0,\underline{\tau}]\cup[\gamma,1)$,
then $\alpha^{\tau}=a^{R}$. 
\item If $\eta<\gamma$ and $\tau\in(\underline{\tau},\gamma)$, then $\alpha^{\tau}$
is the schedule in (\ref{eq:alpha-tau-case-3}). 
\end{itemize}
The supporting multiplier $\Lambda^{\tau}$ such that $\mathcal{L}\left(\alpha^{\tau},\Lambda^{\tau}\right)=\sup_{\alpha'}\mathcal{L}\left(\alpha',\Lambda^{\tau}\right)$
is 
\begin{equation}
\Lambda^{\tau}\left(\theta\right)=\begin{cases}
0 & \text{if }\eta\ge\gamma,\text{ or }\eta<\gamma\text{ and }\tau\notin\left(\underline{\tau},\gamma\right)\\
L\left(\theta;\phi^{\tau}\right)\text{ in (\ref{eq:L(theta;phi)})} & \text{if }\eta<\gamma\text{ and }\tau\in\left(\underline{\tau},\gamma\right)
\end{cases}.\label{eq:Lambda-tau}
\end{equation}

\end{cor}

\subsection{Optimal Mechanism}

Given the fixed-$\tau$ \ref{eq:Program-v} solution, it remains to
do step 2: optimize over the cutoff $\tau$. 
\begin{lem}
\label{Lemma:Pi-v(v)-quasiconcave}$\Pi^{v}$ is strictly quasiconcave
on $\Theta$. Its derivative is $\Pi^{v'}\left(\tau\right)=-f\left(\tau\right)\Psi^{v}\left(\tau\right)$,
where $\Psi^{v}(\tau):=\tilde{\psi}^{v}\left(\alpha^{\tau}(\tau),\tau;\Lambda^{\tau}(\tau)\right)$
is the adjusted virtual surplus of the marginal included type $\tau$,
with $\alpha^{\tau}$ and $\Lambda^{\tau}$ from Corollary \ref{Corollary:V-Fixed-tau}.
It holds that $\Pi^{v'}\left(0\right)>0>\Pi^{v'}\left(1\right)$. 
\end{lem}
\begin{prop}
\label{Proposition:Mechanism-V}Suppose $v\left(\theta\right)=\theta$.
The seller's optimal mechanism is unique. 
\begin{itemize}
\item The inclusion set is $\left[\tau^{v},1\right]$, where $\tau^{v}\in\left(0,1\right)$
is the unique solution to $\Psi^{v}\left(\tau^{v}\right)=0$.
\item The allocation schedule on  $[\tau^{v},1]$ is $\alpha^{\tau^{v}}$,
as characterized in Corollary \ref{Corollary:V-Fixed-tau}.
\item The utility schedule is $u^{\tau^{v}}(\theta)=\int_{\tau^{v}}^{\theta}\left[1-kM_{2}(\alpha^{\tau^{v}}(s),s)\right]ds$,
and the transfer schedule is $t(\theta)=\theta-kM(\alpha^{\tau^{v}}(\theta),\theta)-u^{\tau^{v}}(\theta).$
\begin{itemize}
\item If $\eta\ge\gamma$, then $u^{\tau^{v}}$ is strictly increasing on
$\left[\tau^{v},1\right]$.
\item If $\eta<\gamma$, then $u^{\tau^{v}}$ is nonmonotone on $\left[\tau^{v},1\right]$. 
\end{itemize}
\end{itemize}
On $\left[\tau^{v},1\right]$, $\alpha^{\tau^{v}}\left(\theta\right)\le\alpha^{FB}\left(\theta\right)$,
with strict inequality if and only if $\theta<1$ and $\alpha^{FB}\left(\theta\right)>0$.
\end{prop}
Proposition \ref{Proposition:Mechanism-V} completes the characterization
of the seller's optimal mechanism. I highlight four features.

First, the seller excludes types whose virtual surplus is negative,
following familiar logic from mechanism design. The departure here
is the multiplier $\Lambda^{\tau}\left(\tau^{v}\right)$ term inside
the virtual surplus of the marginal type, which lowers its effective
hazard rate to reflect the IR constraints of higher types. Thus, the
need to maintain IR affects not only product allocation for interior
served types, but also the cutoff determining which types are served.

Second, when $\eta<\gamma$, the optimal utility schedule must be
nonmonotone. The optimal cutoff \emph{never} lies in the Lemma \ref{lem:p-tau-case-2}
region $[\eta,\gamma)$ in this case. For such cutoffs, utility is
flat at zero just above $\tau$, so including the marginal type requires
no additional rent for higher types. Lowering $\tau$ is then unambiguously
profitable---the seller captures the full trade surplus of the added
types without raising rent elsewhere---so $\tau^{v}$ lies strictly
below $\eta$.

Third, product mismatch exceeds the first-best level. To reduce information
rent, the relaxed allocation $a^{R}$ induces more downward mismatch
than the first best. When some types receive negative utility under
$a^{R}$, the seller improves product fit for some types to restore
IR. This correction softens the screening distortion but never overturns
it: each served type's allocated product remains below his first-best,
and strictly so unless he is the highest type or his first-best is
already $y=0$.

Fourth, the optimal mechanism always generates basic-product crowd-out.
\begin{cor}
\label{Corollary:Basic}Under the optimal mechanism, there is always
a positive measure of served types with $\alpha^{\tau^{v}}\left(\theta\right)=0<\alpha^{FB}\left(\theta\right)$.
\end{cor}
Since $\tau^{v}>0$, the seller excludes some low types. At the same
time, Corollary \ref{Corollary:Basic} shows that the basic product
is sold to a positive measure of higher types whose efficient allocation
is nonbasic. Thus the product best suited to the lowest types is instead
consumed by higher types, while those lower types are excluded. Such
crowd-out does not arise in Section \ref{Section:PureHorizontal}:
there, the buyers assigned to the basic product would receive it under
the first best as well. Here, correlation between product fit and
willingness to pay leads the seller to allocate the basic product
to buyers who are less well matched to it but more valuable to serve.

\subsection{Comparative Statics and the Effects of Vertical Heterogeneity}
\begin{prop}
\label{Proposition:Vertical-CS}Suppose that $v\left(\theta\right)=\theta$.
\begin{enumerate}
\item $\tau^{v}$ is strictly decreasing in $k$. If $c=0$, the seller's
profit is strictly increasing in $k$.
\item $\tau^{v}$ is weakly decreasing in $c$, strictly so if and only
if a positive measure of included types has binding IR. The seller's
profit is strictly decreasing in $c$.
\end{enumerate}
\end{prop}
Relative to Section \ref{Section:PureHorizontal}, a higher $k$ now
expands coverage and can raise profit, reversing the comparative static
there. A higher $k$ amplifies the rent effect of any given product
mismatch; what differs from Section \ref{Section:PureHorizontal}
is that mismatch here reduces the rent conceded to higher types, rather
than creating rent for lower types. Assigning a lower type a lower-indexed
product makes that contract a worse fit for higher types, so a larger
$k$ makes it a stronger deterrent, and the seller concedes less rent
to higher types. The seller therefore finds it profitable to serve
additional low types, so $\tau^{v}$ falls.

The same stronger screening role of mismatch also tends to raise profit,
with one caveat. Because mismatch screens more effectively, the seller
induces less of it while preserving the screening effect, assigning
better-matched---hence higher-indexed---products. When $c=0$, this
is costless, so profit rises with $k$. When $c>0$, however, the
better matched products are costlier to produce, which can outweigh
the rent savings. Appendix \ref{OA:Uniform-Profit} illustrates this
tradeoff with an example: when $F$ is the uniform distribution and
$m\left(z\right)=z^{2}/2$, profit rises with $k$ if $c<1/2$ but
falls if $c>1/2$.

The effect of $c$ works differently from $k$. Since the marginal
included type receives the basic product $y=0$, a higher $c$ has
no direct cost effect at the exclusion margin, unlike Section \ref{Section:PureHorizontal},
where the marginal type may receive a non-basic product. Instead,
its effect on coverage operates only through the IR-correction region.
When some included types have binding IR, the seller must improve
product fit relative to the relaxed allocation, and a higher $c$
makes these improvements more expensive. Expanding coverage by lowering
$\tau^{v}$ raises utility before the binding-IR region and reduces
the required correction. Hence $\tau^{v}$ weakly falls with $c$,
strictly so exactly when IR-correction is required. The profit effect
of $c$ is more straightforward: a higher $c$ raises the production
cost of the non-basic products in the optimal menu, reducing profit.

\section{Conclusion}

This paper develops a theory of monopoly screening for product lines
that are horizontally differentiated from the buyers' perspective
but ordered by production cost from the seller's. In such markets,
the seller cannot screen simply by degrading a commonly ranked quality,
and product mismatch then arises endogenously as a screening instrument.
Whether mismatch creates or reduces information rent depends on what
the seller must screen.

When buyers differ only in horizontal fit, mismatch creates rent,
so the seller induces less of it and assigns served buyers products
closer to their ideals than under the first best. When horizontal
fit is correlated with willingness to pay, mismatch instead reduces
rent by deterring higher-value buyers from mimicking lower-value ones,
so the seller induces more of it. Product mismatch thus worsens beyond
the first best, the basic product is sold to buyers whose efficient
product is more advanced while buyers better matched to it are excluded,
and stronger horizontal differentiation expands coverage and can even
raise the seller's profit. These outcomes are opposite to what happens
under pure horizontal differentiation.

The analysis also shows why screening with product mismatch is not
simply vertical screening with a different allocation variable. Because
mismatch disutility is type-specific, incentive compatibility need
not make utility monotone in type, and the types with binding participation
constraints need not be extreme types. The optimal allocation is therefore
shaped not only by incentive compatibility, but also by individual
rationality constraints throughout the included set, unlike in vertical
screening.

I conclude with two directions for further research. First, the model
imposes a one-dimensional type structure in which a buyer's ideal
product and willingness to pay are perfectly correlated. Allowing
these dimensions to be imperfectly correlated would clarify how the
strength of the correlation determines whether mismatch primarily
creates or reduces information rent. Second, the model studies a monopolist,
while many of the motivating applications feature competing sellers
with overlapping product lines. Extending the analysis to competition
would show how strategic interaction changes the use of product mismatch
as a screening instrument and its implications for coverage, product-line
design, and profit.

\bibliographystyle{econ}
\bibliography{ref}

\appendix

\section{\protect\label{Appendix:Proof}Proofs for Main Results}

This appendix gives the proofs of the results related to the characterization
of the optimal mechanisms. The proofs for the other results are collected
in appendix \ref{OA:Proofs-Auxiliary-Results}, and appendix \ref{OA:Supporting-Anaylsis}
provides supporting analysis.

\subsection{Proof of Lemma \ref{lem:h-IC}}
\begin{proof}
\emph{The inclusion interval is $\Theta_{I}=\left[0,\tau\right]$
for some $\tau\in\Theta_{I}$.} 

Suppose that $\mathcal{M}$ is an optimal mechanism and let $u\left(\theta\right)$
denote type $\theta$'s indirect utility. $u$ is continuous, since
$m$ is continuously differentiable on the compact interval $[0,1]$.
Suppose, for contradiction, that \emph{$\Theta_{I}$ }is not an interval.
This implies there exist $\theta_{1}<\theta_{2}$ with $(\theta_{1},\theta_{2})\cap\Theta_{I}=\emptyset$
and included types accumulating at both endpoints. Continuity of $u$
then gives $u(\theta_{2})=0$. Some included type $s>\theta_{2}$
must have $t(s)-c\alpha(s)\geq0$; otherwise, removing his contract
would raise profit and weakly relax remaining incentive constraints.
Combining $s$'s IR constraint with $u(\theta_{2})=0$, 
\begin{eqnarray*}
 & v_{0}-t(s)\geq km(|\alpha(s)-s|),\ \ \ \ \text{ and }\ \ \ \ v_{0}-t(s)\leq km(|\alpha(s)-\theta_{2}|),
\end{eqnarray*}
gives $|\alpha(s)-s|\leq|\alpha(s)-\theta_{2}|$, so $\alpha(s)\geq(s+\theta_{2})/2>\theta_{2}$.
Then $c\alpha(s)\leq t(s)\leq v_{0}$ implies $c\theta<c\alpha(s)\leq v_{0}$
for every $\theta\in(\theta_{1},\theta_{2})$.

Augment the menu by adding the contract $(\theta,v_{0})$ for each
$\theta\in(\theta_{1},\theta_{2})$. Each gap type obtains $0$ from
his own added contract, strictly negative utility from any other added
contract, and at most $0$ from the original menu (since $u=0$ on
the gap), so accepts his own added contract at zero utility. Types
in the original $\Theta_{I}$ obtain at most $v_{0}-v_{0}-km(|\theta-\theta'|)\leq0$
from any added contract and so will not deviate. Thus IC and IR are
preserved, and the added contracts contribute strictly positive aggregate
profit $\int_{\theta_{1}}^{\theta_{2}}(v_{0}-c\theta)\,f(\theta)\,d\theta>0,$
contradicting optimality. Hence $\Theta_{I}$ is an interval.

Let the lower bound of $\Theta_{I}$ be $\underline{\theta}$. If
$\underline{\theta}>0$, then repeating the argument above with $\theta_{1}=0$
and $\theta_{2}=\underline{\theta}$ leads to a contradiction. Therefore,
$\Theta_{I}$ must have lower bound $0$. 

\bigskip

\emph{Necessary and sufficient conditions for IC. }This follows from
arguments similar to those used in vertical screening. Details are
provided in appendix \ref{OA:Lemma-H-Cutoff}.
\end{proof}

\subsection{Proof of Lemmas \ref{lem:h-step-1} and \ref{lem:alpha-h}}
\begin{proof}
Write $R(\theta):=F(\theta)/f(\theta)$, with $R(0)=0$ and $R'>0$
by strict log-concavity of $F$.

\bigskip

\emph{Step 1: $\alpha^{h}$ is well-defined and nondecreasing.} On
$[0,\theta]$, $M(y,\theta)=m(\theta-y)$ and $M_{2}(y,\theta)=m'(\theta-y)$,
so
\begin{equation}
\frac{\partial\psi^{h}}{\partial y}=km'(\theta-y)+km''(\theta-y)\frac{F(\theta)}{f(\theta)}-c,\qquad\frac{\partial^{2}\psi^{h}}{\partial y^{2}}=-km''(\theta-y)-km'''(\theta-y)\frac{F(\theta)}{f(\theta)}.\label{eq:=00005Cpsi^h-derivative}
\end{equation}
For $y<\theta$, $m''(\theta-y)>0$ and $m'''(\theta-y)\geq0$ (Assumption
\ref{Assumption:m-2}), so $\partial^{2}\psi^{h}/\partial y^{2}<0$.
Thus $\psi^{h}(\cdot,\theta)$ is strictly concave on the compact
interval $[0,\theta]$, and $\alpha^{h}(\theta)=\arg\max_{y\in[0,\theta]}\psi^{h}(y,\theta)$
is uniquely attained. The maintained assumptions imply that $\psi^{h}$
is supermodular, which then implies that $\alpha^{h}$ is nondecreasing.
Details are provided in appendix \ref{OA:Lemma-H-Step-1}.

\bigskip

\emph{Step 2: Properties of $\alpha^{h}$ in Lemma }\ref{lem:alpha-h}.
By strict concavity (Step 1), $\alpha^{h}(\theta)$ is determined
by the sign of $\partial_{y}\psi^{h}$ at the two endpoints of $[0,\theta]$.
Let $A(\theta):=m'(\theta)+m''(\theta)R(\theta)$. At the lower corner
$y=0$, $\left.\partial_{y}\psi^{h}\right|_{y=0}=k\big[A(\theta)-c/k\big],$
so $\alpha^{h}(\theta)=0$ iff $A(\theta)\leq c/k$. Since $A'=m''+m'''R+m''R'>0$
on $(0,1]$ (Assumption \ref{Assumption:m-2} and $R'>0$), $A$ is
strictly increasing with $A(0)=0$ and $A(1)=m'(1)+m''(1)/f(1)>m'(1)>c/k$
by Assumption \ref{Assumption:FB-ProductLine}. Hence there is a unique
$\beta^{h}\in(0,1)$ with $A(\beta^{h})=c/k$; $\alpha^{h}=0$ on
$[0,\beta^{h}]$ and $\alpha^{h}>0$ on $(\beta^{h},1]$.

At the upper corner $y=\theta$, using $m'(0)=0$, $\left.\partial_{y}\psi^{h}\right|_{y=\theta}=k\big[m''(0)R(\theta)-c/k\big],$
so $\alpha^{h}(\theta)=\theta$ iff $m''(0)R(\theta)\geq c/k$. If
$m''(0)/f(1)<c/k$, then $m''(0)R(\theta)\leq m''(0)/f(1)<c/k$ for
all $\theta\leq1$, so the upper corner never binds and $\alpha^{h}(\theta)<\theta$
throughout $(\beta^{h},1]$. If $m''(0)/f(1)\geq c/k$, then since
$m''(0)R(\cdot)$ is continuous, strictly increasing, and ranges from
$0$ to $m''(0)/f(1)\geq c/k$, there is a unique $\theta^{*}\in(\beta^{h},1]$
solving $m''(0)R(\theta^{*})=c/k$, which is (\ref{eq:theta-*}).
To see $\theta^{*}>\beta^{h}$, subtract this equation from $A(\theta^{*})$:
\[
A(\theta^{*})-c/k=m'(\theta^{*})+\big[m''(\theta^{*})-m''(0)\big]R(\theta^{*})>0,
\]
using $m'(\theta^{*})>0$, $m''(\theta^{*})\geq m''(0)$ ($m'''\geq0$),
and $R(\theta^{*})>0$; thus $A(\theta^{*})>c/k=A(\beta^{h})$, and
monotonicity of $A$ gives $\theta^{*}>\beta^{h}$. For $\theta\geq\theta^{*}$
the upper corner binds, so $\alpha^{h}(\theta)=\theta$ and, since
$M_{2}(s,s)=0$ for all $s\geq\theta^{*}$, $t(\theta)=v_{0}-km(0)-0=v_{0}$.

The interior optimum $d^{h}(\theta):=\theta-\alpha^{h}(\theta)>0$
satisfies
\begin{equation}
m'\left(d^{h}\left(\theta\right)\right)+m''\left(d^{h}\left(\theta\right)\right)R\left(\theta\right)=c/k.\label{eq:delta-h}
\end{equation}
As $m''(d^{h})R(\theta)>0$, $m'(d^{h})<c/k=m'(\delta^{FB})$, so
$d^{h}<\delta^{FB}$ by strict monotonicity of $m'$. If $\alpha^{FB}(\theta)>0$
then $\alpha^{h}(\theta)=\theta-d^{h}>\theta-\delta^{FB}=\alpha^{FB}(\theta)$;
if $\alpha^{FB}(\theta)=0$ then $\alpha^{h}(\theta)>0=\alpha^{FB}(\theta)$.
In either case $\alpha^{FB}(\theta)<\alpha^{h}(\theta)<\theta$.

\bigskip

\emph{Step 3: }\textit{$\alpha^{h}$ is the unique solution to \ref{eq:Program-h}.}

\textit{(3a) Set up the Lagrangian.} Fix $\tau>0$ and let $u_{\tau}$
denote $u\left(\tau\right)$. Ignore the constraint that $u_{\tau}$
must be zero if $\tau<1$ and treat $u_{\tau}$ as an endogenous choice;
it will be verified later (in step 3c below) that $u_{\tau}$ is always
optimally chosen to be zero. Let 
\[
\tilde{\Pi}\left(\alpha,u_{\tau}\right)=\int_{0}^{\tau}\left[v_{0}-kM(\alpha(\theta),\theta)-c\alpha(\theta)-u(\theta)\right]f(\theta)\,d\theta;\ \ u(\theta)=u_{\tau}+\int_{\theta}^{\tau}kM_{2}(\alpha(s),s)\,ds.
\]
Let $\lambda$ be a nonnegative multiplier measure on the participation
constraints $u(\theta)\geq0$, and define the cumulative multiplier
$\Lambda(\theta):=\lambda([0,\theta]).$ The Lagrangian is 
\[
\int_{0}^{\tau}\left[v_{0}-kM(\alpha(\theta),\theta)-c\alpha(\theta)-u(\theta)\right]f(\theta)\,d\theta+\int_{0}^{\tau}u(\theta)\,\lambda(d\theta).
\]
Substituting $u$ and integrating by parts yields
\begin{equation}
\mathcal{L}(\alpha,u_{\tau};\Lambda)=\int_{0}^{\tau}\tilde{\psi}^{h}(\alpha(\theta),\theta,\Lambda(\theta))f(\theta)\,d\theta+u_{\tau}\left[\Lambda\left(\tau\right)-F\left(\tau\right)\right]\label{eq:Lagrangian-h}
\end{equation}
where
\[
\tilde{\psi}^{h}(y,\theta,L)=v_{0}-kM(y,\theta)-cy-kM_{2}(y,\theta)\frac{F(\theta)-L}{f(\theta)}.
\]
Consider the following cumulative multiplier: (\emph{case (a)}) if
$m''\left(0\right)/f\left(1\right)<c/k$, then $\Lambda\left(\theta\right)=0$
for all $\theta\in\left[0,\tau\right]$; (\emph{case (b)}) if $m''\left(0\right)/f\left(1\right)\ge c/k$,
then
\[
\Lambda\left(\theta\right)=\begin{cases}
0 & \text{ if }\theta\le\theta^{*}\\
F\left(\theta\right)-\frac{cf\left(\theta\right)}{km''\left(0\right)} & \text{ if }\theta>\theta^{*}
\end{cases},
\]
where $R(\theta^{*})=c/\left(km''\left(0\right)\right)$ from (\ref{eq:theta-*}). 

\medskip

\textit{(3b) $\alpha^{h}(\theta)$ uniquely maximizes $\tilde{\psi^{h}}(\cdot,\theta,\Lambda(\theta))$
over $Y$.} Fix $\theta$. Write the effective inverse hazard rate
$\rho(\theta):=\left[F(\theta)-\Lambda(\theta)\right]/f\left(\theta\right)$.
For the proposed multiplier, $\rho(\theta)=R(\theta)$ wherever $\Lambda=0$,
and $\rho(\theta)=R\left(\theta^{*}\right)$ wherever $\Lambda>0$.
In both cases (a) and (b),
\begin{equation}
\rho(\theta)\leq\frac{c}{km''\left(0\right)}\qquad\text{for all }\theta\in\Theta.\label{eq:rho-inequality}
\end{equation}
In case (a), $\rho(\theta)=R\left(\theta\right)\leq R(1)=1/f(1)<c/\left(km''\left(0\right)\right)$
since $m''(0)/f(1)<c/k$. In case (b), $\rho\left(\theta\right)=R\left(\theta^{*}\right)=c/\left(km''\left(0\right)\right)$
for $\theta>\theta^{*}$, and $\rho\left(\theta\right)=R\left(\theta\right)\le R\left(\theta^{*}\right)$
for $\theta\le\theta^{*}$. 

For case (a), $\Lambda\left(\theta\right)=0$, so \textit{\emph{$\tilde{\psi^{h}}(\cdot,\theta,\Lambda(\theta))=\psi^{h}\left(\cdot,\theta\right)$,
which is uniquely maximized by $\alpha^{h}\left(\theta\right)$ on
$\left[0,\theta\right]$.}} For case (b) with $\theta\le\theta^{*}$,
the previous statement applies. For $\theta>\theta^{*}$, \textit{\emph{using
the same algebra in (\ref{eq:=00005Cpsi^h-derivative}) but replacing
$F/f$ with $\rho\ge0$, $\tilde{\psi^{h}}(\cdot,\theta,\Lambda(\theta))$
is still strictly concave on $\left[0,\theta\right]$. Since $\partial_{y}\tilde{\psi^{h}}(\theta,\theta,\Lambda(\theta))=0$,
$\alpha^{h}\left(\theta\right)=\theta$ uniquely maximizes $\tilde{\psi^{h}}(\cdot,\theta,\Lambda(\theta))$
on $\left[0,\theta\right]$. It remains to show that $\tilde{\psi^{h}}(y,\theta,\Lambda(\theta))<\tilde{\psi^{h}}(\theta,\theta,\Lambda(\theta))$
for all $y\in(\theta,1]$. Consider $y=\theta+d$, with $d>0$. Here
}}$M(y,\theta)=m(d)$, $M_{2}(y,\theta)=-m'(d)$, so
\begin{equation}
\partial_{y}\tilde{\psi}^{h}(y,\theta,\Lambda)=-km'(d)-c+km''(d)\,\rho(\theta)\ \le\ -km'(d)-c+\frac{c}{m''(0)}m''(d),\label{eq:partial-tilde-psi}
\end{equation}
where inequality follows from (\ref{eq:rho-inequality}). By the fundamental
theorem of calculus,
\[
m''(d)-m''(0)=\int_{0}^{d}m'''(z)\,dz\leq\frac{m''(0)}{m'(1)}\int_{0}^{d}m''(s)\,ds=\frac{m''(0)}{m'(1)}\,m'(d).
\]
where the inequality follows from $m'''(z)\leq m''(0)m''(z)/m'\left(1\right)$
(Assumption \ref{Assumption:m-2}). Hence 
\[
\frac{c}{km''(0)}\big[m''(d)-m''(0)\big]\leq\frac{c}{km'(1)}\,m'(d)<m'(d),
\]
where the last inequality is Assumption \ref{Assumption:FB-ProductLine}
($c/k<m'(1)$). Multiplying by $k$ throughout, the last inequality
implies that $\partial_{y}\tilde{\psi}^{h}(y,\theta,\Lambda)$ in
(\ref{eq:partial-tilde-psi}) is negative. Therefore,\textit{\emph{
$\tilde{\psi^{h}}(y,\theta,\Lambda(\theta))<\tilde{\psi^{h}}(\theta,\theta,\Lambda(\theta))$
for all $y\in(\theta,1]$.}}

\medskip

\textit{(3c) $u_{\tau}=0$ is always optimal.} The coefficient on
$u_{\tau}$ in the Lagrangian in (\ref{eq:Lagrangian-h}) is $\Lambda(\tau)-F(\tau)$.
In case (a), it is $-F(\tau)<0$; in case (b), it is $-cf(\tau)/\left(km''\left(0\right)\right)<0$
(or $-F(\tau)<0$ if $\tau\leq\theta^{*}$). Since $\mathcal{L}$
is linear in $u_{\tau}$ with negative slope and $u_{\tau}\geq0$,
the maximizer is $u_{\tau}=0$.

\medskip

\textit{(3d) $\Lambda$ is a valid nonnegative measure.} In case (a),
$\Lambda\equiv0$. In case (b), $\Lambda=0$ on $[0,\theta^{*}]$
and $\Lambda(\theta)=F(\theta)-cf(\theta)/\left(km''(0)\right)$ on
$(\theta^{*},\tau]$. It is continuous at $\theta^{*}$, since $\Lambda(\theta^{*})=F(\theta^{*})-f(\theta^{*})R(\theta^{*})=0$.
On $(\theta^{*},\tau]$,
\[
\Lambda'(\theta)=f(\theta)\Big[1-\tfrac{c}{km''(0)}\tfrac{f'(\theta)}{f(\theta)}\Big]=f(\theta)\Big[1-R(\theta^{*})\tfrac{f'(\theta)}{f(\theta)}\Big]\geq0,
\]
because for $\theta\geq\theta^{*}$, $R(\theta^{*})\leq R(\theta)$,
so if $f'>0$, then $R(\theta^{*})f'/f\leq Ff'/f^{2}<1$ by strict
log-concavity of $F$; if $f'\le0$, then the inequality above is
trivial. Hence $\Lambda$ is nondecreasing, so $\lambda=d\Lambda\geq0$.

\medskip

\textit{(3e) IR and complementary slackness. }\textit{\emph{Let $u^{h}\left(\theta\right)=\int_{\theta}^{\tau}kM_{2}\left(\alpha^{h}\left(z\right),z\right)dz$
denote the indirect utility under $\alpha^{h}$. For $\theta>\theta^{*}$
in case (b),}} $\alpha^{h}\left(\theta\right)=\theta$ gives $u^{h\prime}\left(\theta\right)=-kM_{2}(\theta,\theta)=0$,
and with $u_{\tau}=0$, $u^{h}\left(\theta\right)=0$. Otherwise,
for all other $\theta$, $\alpha^{h}\left(\theta\right)<\theta$,
so $u^{h\prime}\left(\theta\right)=-km'(\theta-\alpha^{h}\left(\theta\right))<0$,
so $u^{h}\left(\theta\right)>0$ if $\theta\le\theta^{*}$ or $\theta<\tau$.
Therefore, IR holds. Furthermore, $\Lambda\left(\theta\right)>0$
only in case (b) and with $\theta>\theta^{*}$. $\int_{0}^{\theta}u^{h}\left(z\right)d\lambda\left(z\right)$
is thus zero for all $\theta$.

\medskip

\textit{(3f) $\alpha^{h}$ and $u_{\tau}=0$ uniquely solve \ref{eq:Program-h}}\textit{\emph{.
For all feasible $\alpha:\Theta\rightarrow Y$ and $u_{\tau}\ge0$,
\[
\tilde{\Pi}\left(\alpha,u_{\tau}\right)\le\mathcal{L}(\alpha,u_{\tau};\Lambda)\le\mathcal{L}(\alpha^{h},0;\Lambda)=\tilde{\Pi}\left(\alpha^{h},0\right).
\]
The first inequality is weak duality. The second inequality follows
from steps (3b) and (3c). The last equality follows from complementary
slackness. From steps (3b) and (3c), since $\alpha^{h}(\theta)$ is
the unique maximizer of $\tilde{\psi^{h}}(\cdot,\theta,\Lambda(\theta))$,
the second inequality is strict if $\alpha\ne\alpha^{h}$ or $u_{\tau}\ne0$.
This establishes solution uniqueness.}}
\end{proof}

\subsection{Proof of Lemma \ref{lem:h-step-2}}
\begin{proof}
$\Psi^{h}$ is continuous on $[0,1]$ because $\alpha^{h}$ is, with
$\Psi^{h}(0)=v_{0}>0$. By straightforward algebra, $\Psi^{h}$ is
strictly decreasing on $(0,1]$ under the maintained assumptions.
Details are provided in appendix \ref{OA:Lemma-H-Step-2}. Since $(\Pi^{h})'(\tau)=\Psi^{h}(\tau)f(\tau)$
and $f>0$, $(\Pi^{h})'$ changes sign at most once on $(0,1)$, from
positive to negative, so $\Pi^{h}$ is strictly quasiconcave on $[0,1]$.
If $\Psi^{h}(1)\geq0$, strict monotonicity gives $\Psi^{h}(\tau)>0$
for all $\tau\in(0,1)$, so $\Pi^{h}$ is strictly increasing and
$\tau^{h}=1$. If $\Psi^{h}(1)<0$, then, combined with $\Psi^{h}(0)=v_{0}>0$,
continuity and the intermediate value theorem give a solution to $\Psi^{h}(\tau)=0$,
and strict monotonicity gives its uniqueness; $\Pi^{h}$ is strictly
increasing below this solution and strictly decreasing above it, so
$\tau^{h}\in(0,1)$ is the unique maximizer.
\end{proof}

\subsection{Proof of Proposition \ref{Proposition:H-solution}}
\begin{proof}
The optimal mechanism follows from Lemma \ref{lem:h-step-1} and Lemma
\ref{lem:h-step-2}. It remains to prove that $\tau^{h}=1$ if $v_{0}\ge c$.
If $m''(0)/f(1)\geq c/k$, then $\theta^{*}\leq1$ exists and $\alpha^{h}(1)=1$,
so $\Psi^{h}(1)=v_{0}-c\ge0$ and thus $\tau^{h}=1$. If $m''(0)/f(1)<c/k$,
then $\alpha^{h}\left(\tau^{h}\right)<\tau^{h}$ and thus $\Psi^{h}\left(\tau^{h}\right)=\psi^{h}\left(\alpha^{h}\left(\tau^{h}\right),\tau^{h}\right)>\psi^{h}\left(\tau^{h},\tau^{h}\right)=v_{0}-c\tau^{h}\ge0$.
$\Psi^{h}\left(\tau^{h}\right)>0$ implies $\tau^{h}=1$. 
\end{proof}

\subsection{Proof of Lemma \ref{lem:v-CutoffMech}}
\begin{proof}
Let $\mathcal{M}=\{\Theta_{I};\alpha(\cdot),t(\cdot)\}$ be an optimal
mechanism. Suppose that there exist $\theta_{1}\in\Theta_{I}$ and
$\theta_{2}\in(\theta_{1},1]$ with $(\theta_{1},\theta_{2})\cap\Theta_{I}=\emptyset$,
where either $\theta_{2}\in\Theta_{I}$ (\emph{interior gap}) or $\theta_{2}=1$
(\emph{top gap}). I construct another feasible mechanism $\widetilde{\mathcal{M}}$
that improves the seller's profit. This is a contradiction, which
implies then that $\Theta_{I}=\left[\tau,1\right]$ for some $\tau$. 

Define $\delta^{0}:=(m')^{-1}(1/k)$, $a^{0}(\theta):=\max\{\theta-\delta^{0},0\}$,
$t^{0}(\theta):=\theta-kM(a^{0}(\theta),\theta)$, and $\pi^{0}(\theta):=t^{0}(\theta)-ca^{0}(\theta)$.
Define $\widetilde{\mathcal{M}}=\{\widetilde{\Theta}_{I};\widetilde{\alpha}(\cdot),\widetilde{t}(\cdot)\}$
by $\widetilde{\Theta}_{I}:=\Theta_{I}\cup(\theta_{1},\theta_{2})$,
$\widetilde{\alpha}=\alpha$ and $\widetilde{t}=t$ on $\Theta_{I}$,
and $\widetilde{\alpha}=a^{0}$, $\widetilde{t}=t^{0}$ on $(\theta_{1},\theta_{2})$.

\emph{$\theta_{1}\ge\delta^{0}$.} Types in $(\theta_{1},\theta_{2})$
are excluded under $\mathcal{M}$. For any $\varepsilon>0$ small
enough that $\theta_{1}+\varepsilon\in(\theta_{1},\theta_{2})$, IC
for the excluded type $\theta_{1}+\varepsilon$ requires $(\theta_{1}+\varepsilon)-kM(\alpha(\theta_{1}),\theta_{1}+\varepsilon)-t(\theta_{1})\leq0.$
Letting $\varepsilon\downarrow0$, the left-hand side converges to
$u(\theta_{1})\geq0$ (by IR), so $u(\theta_{1})=0$ and the right-derivative
at $\varepsilon=0$ is nonpositive: $1-kM_{2}(\alpha(\theta_{1}),\theta_{1})\leq0.$
Since $M_{2}(y,\theta)\leq0$ when $y\geq\theta$, this inequality
forces $\alpha(\theta_{1})<\theta_{1}$, in which case $M_{2}(\alpha(\theta_{1}),\theta_{1})=m'(\theta_{1}-\alpha(\theta_{1}))$.
The inequality above then becomes $m'(\theta_{1}-\alpha(\theta_{1}))\geq1/k$,
i.e., $\theta_{1}-\alpha(\theta_{1})\geq\delta^{0}$. Combined with
$\alpha(\theta_{1})\geq0$, this gives $\theta_{1}\geq\delta^{0}$.

\emph{$\widetilde{\mathcal{M}}$ satisfies IR. }On $\Theta_{I}$,
$\widetilde{u}(\theta)=u(\theta)\geq0$ by IR of $\mathcal{M}$. On
$(\theta_{1},\theta_{2})$, $\widetilde{u}(\theta)=0$.

\emph{$\widetilde{\mathcal{M}}$ satisfies IC. }I check separately
that (a) no type wants to mimic a contract in $\Theta_{I}$, and (b)
no type wants to mimic a contract in $(\theta_{1},\theta_{2})$. For
(a): for any type $\theta\in\Theta$ and any $\theta'\in\Theta_{I}$,
the contract $(\alpha(\theta'),t(\theta'))$ is unchanged, so IC of
$\mathcal{M}$ already ensures $\theta$ does not strictly prefer
it to her own contract under $\mathcal{M}$. If $\theta\in\Theta_{I}$,
her contract is unchanged and (a) follows directly. If $\theta\in(\theta_{1},\theta_{2})$,
her new utility is $0$, while the utility from mimicking any $\theta'\in\Theta_{I}$
is at most $0$ by IC of $\mathcal{M}$ (since IC of $\mathcal{M}$
requires excluded types to weakly prefer the outside option). If $\theta\notin\widetilde{\Theta}_{I}$,
her outside option is $0$, and IC of $\mathcal{M}$ applied at this
$\theta$ gives the bound. For (b): the contracts in $(\theta_{1},\theta_{2})$
are $(a^{0}(\theta'),t^{0}(\theta'))$ with $\theta'\geq\delta^{0}$
(since $\theta_{1}\ge\delta^{0}$ from above). For any $\theta'\geq\delta^{0}$,
the function $\theta\mapsto\theta-kM(a^{0}(\theta'),\theta)-t^{0}(\theta')$
is concave in $\theta$ with a unique maximum at $\theta=\theta'$,
where its value is zero. So no type $\theta$ strictly prefers the
contract of $\theta'$ to his own.

\emph{$\widetilde{\mathcal{M}}$ is more profitable than $\mathcal{M}$.
}The profit difference is $\int_{\theta_{1}}^{\theta_{2}}\pi^{0}(\theta)f(\theta)d\theta$.
For $\theta\in(0,\delta^{0}]$, $a^{0}(\theta)=0$ and $\pi^{0}(\theta)=\theta-km(\theta)$.
$\pi^{0}(0)=0$ and $\pi^{0'}(\theta)=1-km'(\theta)>0$ on $[0,\delta^{0})$
since $km'(\delta^{0})=1$ and $m'$ is strictly increasing. For $\theta\in[\delta^{0},1]$,
$a^{0}(\theta)=\theta-\delta^{0}$ and $\pi^{0}(\theta)=(1-c)\theta+c\delta^{0}-km(\delta^{0})$.
$\pi^{0'}(\delta^{0})=\delta^{0}-km(\delta^{0})>0$ by the previous
case. For $\theta>\delta^{0}$, $\pi^{0}$ is increasing in $\theta$
since $c<1$. Therefore, $\pi^{0}(\theta)>0$ for all $\theta\in(0,1]$.

\bigskip

\emph{Necessary and sufficient conditions for IC. }This follows from
arguments similar to those used in vertical screening. Details are
provided in appendix \ref{OA:Lemma-Vertical}.
\end{proof}

\subsection{Proof of Lemma \ref{lem:Relaxed-program}}
\begin{proof}
Fix $\theta$. I first characterize the pointwise maximizer of $\psi^{v}(\cdot,\theta)$
over $Y=[0,1]$. If $y>\theta$, then $M(y,\theta)=m(y-\theta)$ and
$M_{2}(y,\theta)=-m'(y-\theta)$. Hence $\partial_{y}\psi^{v}(y,\theta)=-km'(y-\theta)-c-km''(y-\theta)(1-F(\theta))/f(\theta)<0$.
Thus no pointwise maximizer can involve upward mismatch; it is enough
to maximize over $y\in[0,\theta]$. On this region, let $d=\theta-y\in[0,\theta]$
and $\hat{\psi}(d,\theta):=\psi^{v}(\theta-d,\theta)$. Then $\partial_{d}\hat{\psi}=c-km'(d)+km''(d)\cdot(1-F(\theta))/f(\theta)$
has the same sign as $\left(1-F\left(\theta\right)\right)/f\left(\theta\right)-G\left(d\right)$,
where $G\left(d\right):=[m'(d)-c/k]/m''(d).$ 
\begin{equation}
G'(d)=\frac{[m''(d)]^{2}-[m'(d)-c/k]m'''(d)}{[m''(d)]^{2}}.\label{eq:G'>0}
\end{equation}
$G'(d)>0$: If $m'(d)-c/k\le0$, $G'(d)>0$ because $m''>0$ and $m'''\ge0$.
If $m'(d)-c/k>0$, log-concavity of $m'$ implies $m'(d)m'''(d)<[m''(d)]^{2}$;
hence $[m'(d)-c/k]m'''(d)<[m''(d)]^{2}$. 

It follows that $\hat{\psi}(d,\theta)$ is single-peaked in $d$.
The unconstrained optimum is the unique $d$ satisfying $G(d)=\left(1-F\left(\theta\right)\right)/f\left(\theta\right)$,
equivalently, $d^{R}\left(\theta\right)$ in (\ref{eq:delta-R}).
If $d^{R}(\theta)\le\theta$, this unconstrained optimum is feasible
and the pointwise maximizer is $a^{R}(\theta)=\theta-d^{R}(\theta)$.
If $d^{R}(\theta)>\theta$, then the objective is increasing in $d$
throughout the feasible interval $[0,\theta]$, so the constrained
maximizer is $d=\theta$, or equivalently $a^{R}(\theta)=0$. This
gives the stated allocation rule. $b^{R}<1$ because at $\theta=1$,
(\ref{eq:delta-R}) gives $m'\left(d^{R}\left(1\right)\right)=c/k<m'\left(1\right)$
(Assumption \ref{Assumption:FB-ProductLine}), so $d^{R}\left(1\right)<1$. 

It remains to check that $a^{R}$ is monotonic. On $[\tau,b^{R}]$,
$a^{R}\left(\theta\right)=0$. On $(b^{R},1]$, $G(d^{R}\left(\theta\right))=\left(1-F\left(\theta\right)\right)/f\left(\theta\right)$,
where the RHS is decreasing in $\theta$ (Assumption \ref{Assumption:F}).
Since $G'>0$, $d^{R}$ is strictly decreasing on $(b^{R},1]$, implying
that $a^{R}(\theta)=\theta-d^{R}(\theta)$ is strictly increasing. 

Finally, compare $a^{R}$ with the first-best allocation. The first-best
mismatch $\delta^{FB}$ satisfies $m'(\delta^{FB})=c/k$. From (\ref{eq:delta-R}),
$m'(d^{R}(\theta))=c/k+m''(d^{R}(\theta))\cdot(1-F(\theta))/f(\theta)\geq c/k$,
with strict inequality if $\theta<1$. Strict convexity of $m$ gives
$d^{R}(\theta)\geq\delta^{FB}$, with strict inequality if $\theta<1$.
Hence for $\theta\in(b^{R},1)$, $a^{R}(\theta)=\theta-d^{R}(\theta)<\theta-\delta^{FB}\leq\alpha^{FB}(\theta)$;
at $\theta=1$, $a^{R}(1)=\alpha^{FB}(1)$. For $\theta\in[\tau,b^{R}]$,
$a^{R}(\theta)=0\leq\alpha^{FB}(\theta)$, with strict inequality
whenever $\alpha^{FB}(\theta)>0$.
\end{proof}

\subsection{Proof of Lemma \ref{lem:p-tau-case-1}}
\begin{proof}
Take $\Lambda=0$. Then $\tilde{\psi}^{v}\left(\cdot,\theta;0\right)=\psi^{v}\left(\cdot,\theta\right)$,
and Lemma \ref{lem:Relaxed-program} implies that $a^{R}$ is the
\emph{unique} pointwise maximizer of $\tilde{\psi}^{v}\left(\cdot,\theta;\Lambda\left(\theta\right)\right)$
at every $\theta$. The pair $\left(a^{R},\Lambda=0\right)$ satisfies
complementary slackness trivially. $a^{R}$ is nondecreasing (Lemma
\ref{lem:Relaxed-program}), so IC holds. It remains to verify $u^{R}\left(\theta;\tau\right)\geq0$
on $\left[\tau,1\right]$. The slope of $u^{R}\left(\cdot;\tau\right)$
is $1-km'\left(\delta^{R}\left(\theta\right)\right)$ where $\delta^{R}\left(\theta\right)=\theta$
on $\left[0,b^{R}\right]$ and $\delta^{R}\left(\theta\right)=d^{R}\left(\theta\right)$
on $\left[b^{R},1\right]$; hence $\delta^{R}$ peaks at $b^{R}$.

\emph{Case $\eta\geq\gamma$:} Then $\delta^{R}\left(\theta\right)\leq b^{R}\leq\eta$
on $\left[0,1\right]$, so the slope is non-negative throughout. Hence
$u^{R}\left(\theta;\tau\right)\geq u^{R}\left(\tau;\tau\right)=0$.

\emph{Case $\eta<\gamma$:} The slope is positive on $\left[0,\eta\right)$,
negative on $\left(\eta,\gamma\right)$, and positive on $\left(\gamma,1\right]$.
If $\tau<\gamma$, $u^{R}\left(\cdot;\tau\right)$ has its unique
local minimum on $\left[\tau,1\right]$ at $\gamma$. $\tau\notin\left(\underline{\tau},\gamma\right)$
has two cases. If $\tau\ge\gamma$, the slope is positive on $(\tau,1)$,
so $u^{R}\left(\theta;\tau\right)\geq u^{R}\left(\tau;\tau\right)=0$.
If $\tau\le\underline{\tau}$, then $u^{R}\left(\gamma;\tau\right)\ge0$.
As $\gamma$ is the only local minimum in $\left(\tau,1\right)$,
$u^{R}\left(\cdot;\tau\right)$ must be nonnegative on $\left[\tau,1\right]$.
\end{proof}

\subsection{Proof of Lemma \ref{lem:p-tau-case-2}}
\begin{proof}
Take the $\Lambda$ in (\ref{eq:Lambda-Case-2}).

\emph{$\Lambda$ is a valid cumulative multiplier. }Write $L^{0}\left(\theta\right)=\left(1-F\left(\theta\right)\right)H\left(\theta\right)$
with $H\left(\theta\right):=1-h\left(\theta\right)\left(1-F\left(\gamma\right)\right)/f\left(\gamma\right)$
and $h:=f/\left(1-F\right)$. By Assumption \ref{Assumption:F}, $h$
is nondecreasing, so $H$ is nonincreasing; since $H\left(\gamma\right)=0$,
$H\geq0$ on $\left[\tau,\gamma\right]$. Then $\left(L^{0}\right)'=-fH+\left(1-F\right)H'\leq0$
on $\left[\tau,\gamma\right]$, so $L^{0}$ is nonincreasing with
$L^{0}\left(\gamma\right)=0$. Thus $\Lambda$ is nonincreasing on
$\left[\tau,1\right]$ with $\Lambda\left(1\right)=0$.

\emph{$\alpha^{\tau}$ in (\ref{eq:y-tau-case2}) pointwise maximizes
$\tilde{\psi}^{v}\left(\cdot,\theta;\Lambda\left(\theta\right)\right)$}.
On $\left(\gamma,1\right]$, $\Lambda=0$ and $\tilde{\psi}^{v}=\psi^{v}$;
$a^{R}$ is the unique maximizer by Lemma \ref{lem:Relaxed-program}.
On $\left[\tau,\gamma\right]$, the FOC at $y=a^{0}\left(\theta\right)$
reduces to
\[
1-c-km''\left(\eta\right)\frac{1-F\left(\theta\right)-\Lambda\left(\theta\right)}{f\left(\theta\right)}=0.
\]
Equation (\ref{eq:delta-R}) evaluated at $\theta=\gamma$ gives $\left(1-F\left(\gamma\right)\right)/f\left(\gamma\right)=\left(1-c\right)/\left(km''\left(\eta\right)\right)$,
so the FOC is satisfied exactly when $\Lambda\left(\theta\right)=L^{0}\left(\theta\right)$.
The Lemma \ref{lem:Relaxed-program} argument, applied with $\left(1-F-L^{0}\right)/f$
in place of $\left(1-F\right)/f$, shows $\tilde{\psi}^{v}\left(\cdot,\theta;L^{0}\left(\theta\right)\right)$
is strictly quasiconcave in $y$, so $a^{0}\left(\theta\right)$ is
the unique maximizer on $\left[0,1\right]$ provided it lies in $\left[0,1\right]$.
The latter is immediate: $\tau\geq\eta$ gives $a^{0}\left(\tau\right)\geq0$,
and $\gamma\leq1$ gives $a^{0}\left(\gamma\right)=\gamma-\eta<1$.

\emph{(IC) $\alpha^{\tau}$ is nondecreasing}. $a^{0}\left(\theta\right)=\theta-\eta$
is strictly increasing on $\left[\tau,\gamma\right]$; $a^{R}$ is
nondecreasing on $\left[\gamma,1\right]$ (Lemma \ref{lem:Relaxed-program});
and $a^{0}\left(\gamma\right)=\gamma-\eta=\gamma-d^{R}\left(\gamma\right)=a^{R}\left(\gamma\right)$.

\emph{IR:} On $\left[\tau,\gamma\right]$, the envelope condition
(\ref{eq:u'-v}) gives $u^{\tau\prime}=1-km'\left(\eta\right)=0$,
so $u^{\tau}\equiv0$ on $\left[\tau,\gamma\right]$. On $\left[\gamma,1\right]$,
$d^{R}\left(\theta\right)<d^{R}\left(\gamma\right)=\eta$ for $\theta>\gamma$,
so $u^{\tau\prime}>0$ and $u^{\tau}\geq u^{\tau}\left(\gamma\right)=0$.

\emph{Complementary slackness:} $\Lambda$ is nonincreasing on $\left[\tau,\gamma\right]$
where $u^{\tau}=0$, and is constant ($=0$) on $\left(\gamma,1\right]$
where $u^{\tau}>0$.
\end{proof}

\subsection{Proof of Lemma \ref{lem:p-tau-case-3}}
\begin{proof}
\emph{For $\phi\in\left(\tau,\gamma\right)$, $\tilde{\psi}^{v}\left(\cdot,\theta;L^{0}\left(\phi\right)\right)$
is strictly quasiconcave on $Y$}. $L^{0}\left(\phi\right)\in\left(0,L^{0}\left(\tau\right)\right]$
by the monotonicity of $L^{0}$ established in the proof of Lemma
\ref{lem:p-tau-case-2}. The Lemma \ref{lem:Relaxed-program} argument
with $\left(1-F-L^{0}\left(\phi\right)\right)/f$ replacing $\left(1-F\right)/f$
shows $\tilde{\psi}^{v}\left(\cdot,\theta;L^{0}\left(\phi\right)\right)$
is strictly quasiconcave, so $a\left(\theta;\phi\right)$ is the unique
maximizer of $\tilde{\psi}^{v}\left(\cdot,\theta;L^{0}\left(\phi\right)\right)$.
Recall that $G\left(\delta\right):=[m'(\delta)-c/k]/m''(\delta)$,
and $G$ is strictly increasing by (\ref{eq:G'>0}). By Assumption
\ref{Assumption:F}, $(1-F-L^{0}(\phi))/f$ is nonincreasing on $[\tau,\phi]$
(see appendix \ref{OA:Adjusted-hazardrate}). Define $\beta(\phi):=\sup\{\theta:a(\theta;\phi)=0\}$.
$G(\theta)-(1-F(\theta)-L^{0}(\phi))/f(\theta)$ is then strictly
increasing in $\theta$, so $\{\theta:a(\theta;\phi)=0\}=[\tau,\beta(\phi)]$
and $a(\cdot;\phi)$ is nondecreasing.

\emph{Existence of $\beta^{\tau}$ and $\phi^{\tau}$}: Reparameterizing
by $\delta=\theta-y$, the maximizer of $\tilde{\psi}^{v}\left(\cdot,\theta;L^{0}\left(\phi\right)\right)$
satisfies the corner $a\left(\theta;\phi\right)=0$ if $G\left(\theta\right)\leq\left(1-F\left(\theta\right)-L^{0}\left(\phi\right)\right)/f\left(\theta\right)$
and the interior FOC 
\begin{equation}
G\left(\theta-a\left(\theta;\phi\right)\right)=\left(1-F\left(\theta\right)-L^{0}\left(\phi\right)\right)/f\left(\theta\right)\label{eq:a(theta,phi)-FOC}
\end{equation}
 otherwise. Define $\Phi\left(\theta\right):=1-F\left(\theta\right)-G\left(\theta\right)f\left(\theta\right)$.
By the corner condition, $\beta\left(\phi\right)$ is the largest
$\theta$ with $\Phi\left(\theta\right)\geq L^{0}\left(\phi\right)$.
On $\left[0,b^{R}\right]$, $\Phi\geq0$ since $G\left(\theta\right)\leq\left(1-F\left(\theta\right)\right)/f\left(\theta\right)$
in the pooled region (Lemma \ref{lem:Relaxed-program}), with $\Phi\left(b^{R}\right)=0$.
Continuity of $\Phi$ and the intermediate value theorem then give
$\beta\left(\phi\right)<b^{R}$ for every $\phi\in\left(\tau,\gamma\right)$. 

Define $U\left(\phi;\tau\right)$ as in the main text. Whenever $\phi>\beta(\phi)$,
$a(\phi;\phi)=a^{0}(\phi)=\phi-\eta$, so $\phi-a(\phi;\phi)=\eta$,
$m'(\phi-a(\phi;\phi))=1/k$, and the boundary term in $\partial U/\partial\phi$
vanishes. For $\phi\leq\beta\left(\phi\right)$, $a\left(s;\phi\right)=0$
on $\left[\tau,\phi\right]$, so $U\left(\phi;\tau\right)=\int_{\tau}^{\phi}\left[1-km'\left(s\right)\right]ds$,
which is independent of $\phi$ beyond its upper limit and satisfies
$\partial U\left(\phi;\tau\right)/\partial\phi=1-km'\left(\phi\right)>0$
for $\phi\in\left(\tau,\eta\right)$. For $\phi\in\left(\eta,\gamma\right)$,
$\phi>\beta\left(\phi\right)$ and the boundary term vanishes. Differentiating
under the integral and using $\partial a/\partial\phi\leq0$ (since
$L^{0}$ is nonincreasing and $a$ is nondecreasing in $L^{0}$ by
the FOC), the integrand $km''\left(s-a\right)\partial a/\partial\phi$
is negative on $\left[\beta\left(\phi\right),\phi\right]$, with strict
inequality on a set of positive measure. Hence $\partial U\left(\phi;\tau\right)/\partial\phi<0$
for $\phi\in\left(\eta,\gamma\right)$. Thus $U\left(\cdot;\tau\right)$
is strictly increasing on $\left(\tau,\eta\right)$ and strictly decreasing
on $\left(\eta,\gamma\right)$, with $U\left(\tau;\tau\right)=0$,
$U\left(\eta;\tau\right)>0$, and $U\left(\gamma;\tau\right)=u^{R}\left(\gamma;\tau\right)<0$
(the last equality holds because $L^{0}\left(\gamma\right)=0$, so
$a\left(\cdot;\gamma\right)=a^{R}$; the strict negativity is by $\tau>\underline{\tau}$).
By the intermediate value theorem, there is a unique $\phi^{\tau}\in\left(\eta,\gamma\right)$
with $U\left(\phi^{\tau};\tau\right)=0$. Set $\beta^{\tau}:=\beta\left(\phi^{\tau}\right)$.
Since $L^{0}\left(\phi^{\tau}\right)>0=\Phi\left(b^{R}\right)$, $\beta^{\tau}<b^{R}$.
Since $\tau<\eta$, $\Phi(\tau)>L^{0}(\tau)$ (equivalently $G(\tau)<G(\eta)$);
since $\phi^{\tau}>\tau$ and $L^{0}$ is strictly decreasing, $L^{0}(\phi^{\tau})<L^{0}(\tau)<\Phi(\tau)$.
Hence $\beta^{\tau}=\beta(\phi^{\tau})>\tau$, and $\beta^{\tau}\in(\tau,b^{R})$.

\emph{(IC) $\alpha^{\tau}$ is nondecreasing}. On $\left[\tau,\beta^{\tau}\right]$,
$a\left(\cdot;\phi^{\tau}\right)$ is the constant zero. On $\left(\beta^{\tau},\phi^{\tau}\right)$,
since $G$ is increasing, $a\left(\cdot;\phi^{\tau}\right)$ is strictly
increasing by the implicit function theorem applied to the interior
FOC in (\ref{eq:a(theta,phi)-FOC}). Equation (\ref{eq:a(theta,phi)-FOC})
applied at $\theta=\beta^{\tau}$ gives $a\left(\beta^{\tau};\phi^{\tau}\right)=0$
by definition, so $a\left(\cdot;\phi^{\tau}\right)$ is nondecreasing
on $[\tau,\phi^{\tau})$. Equation (\ref{eq:a(theta,phi)-FOC}) applied
at $\theta=\phi^{\tau}$ gives $a\left(\phi^{\tau};\phi^{\tau}\right)=\phi^{\tau}-\eta=a^{0}\left(\phi^{\tau}\right).$
By IC argument in the proof of Lemma \ref{lem:p-tau-case-2}, $\alpha^{\tau}$
is also nondecreasing on $\left[\phi^{\tau},1\right]$.

\emph{IR.} On $\left[\tau,\phi^{\tau}\right]$, $u^{\tau}\left(\theta\right)=\int_{\tau}^{\theta}\left[1-km'\left(s-a\left(s;\phi^{\tau}\right)\right)\right]ds$.
The integrand is positive for $s<\eta$. For $s\in\left(\eta,\phi^{\tau}\right)$,
the mismatch is at least $\eta$, so the integrand is nonpositive.
Hence $u^{\tau}$first increases and then decreases on $\left[\tau,\phi^{\tau}\right]$.
Since $u^{\tau}\left(\tau\right)=0$ and $u^{\tau}\left(\phi^{\tau}\right)=U\left(\phi^{\tau};\tau\right)=0$,
it follows that $u^{\tau}\left(\theta\right)\ge0$ on $\left[\tau,\phi^{\tau}\right]$.
On $\left[\phi^{\tau},\gamma\right]$, $u^{\tau\prime}=0$ and $u^{\tau}\left(\phi^{\tau}\right)=0$,
so $u^{\tau}=0$. On $\left[\gamma,1\right]$, $u^{\tau\prime}\geq0$
and $u^{\tau}\left(\gamma\right)=0$, so $u^{\tau}\geq0$.

Let $\Lambda^{\tau}\left(\theta\right)=L\left(\theta;\phi^{\tau}\right)$
in (\ref{eq:L(theta;phi)}).

\emph{$\alpha^{\tau}\left(\theta\right)$ in (\ref{eq:alpha-tau-case-3})
uniquely pointwise maximizes $\tilde{\psi}^{v}\left(\cdot,\theta,\Lambda^{\tau}\left(\theta\right)\right)$
on $\left[\tau,1\right]$: }On $\left[\tau,\phi^{\tau}\right]$, $\alpha^{\tau}\left(\theta\right)=a\left(\theta;\phi^{\tau}\right)$,
which uniquely maximizes $\tilde{\psi}^{v}\left(\cdot,\theta,\Lambda^{\tau}\left(\theta\right)\right)=\tilde{\psi}^{v}\left(\cdot,\theta,L^{0}\left(\phi^{\tau}\right)\right)$
by definition. On $\left(\phi^{\tau},\gamma\right)$, $\tilde{\psi}^{v}\left(y,\theta,\Lambda^{\tau}\left(\theta\right)\right)=\tilde{\psi}^{v}\left(y,\theta,L^{0}\left(\theta\right)\right)$
and $\alpha^{\tau}\left(\theta\right)=a^{0}\left(\theta\right)$.
As shown in the proof of Lemma \ref{lem:p-tau-case-2}, $a^{0}\left(\theta\right)$
uniquely maximizes $\tilde{\psi}^{v}\left(\cdot,\theta,L^{0}\left(\theta\right)\right)$.
On $\left[\gamma,1\right]$, $\tilde{\psi}^{v}\left(y,\theta,\Lambda^{\tau}\left(\theta\right)\right)=\psi^{v}\left(y,\theta\right)$
and $\alpha^{\tau}\left(\theta\right)=a^{R}\left(\theta\right)$.
As shown in the proof of Lemma \ref{lem:Relaxed-program}, $a^{R}\left(\theta\right)$
uniquely maximizes $\psi^{v}\left(\cdot,\theta\right)$.

\emph{$\Lambda^{\tau}$ is a valid cumulative measure: }On $\left[\tau,\phi^{\tau}\right]$
and $\left[\gamma,1\right]$, $\Lambda^{\tau}$ is constant. On $\left[\phi^{\tau},\gamma\right]$,
$\Lambda^{\tau}=L^{0}$, which is nonincreasing. At $\phi^{\tau}$
and $\gamma$, $\Lambda^{\tau}$ is continuous (with values $L^{0}\left(\phi^{\tau}\right)$
and $0$).

\emph{Complementary slackness:} On $\left(\tau,\phi^{\tau}\right)$,
$u^{\tau}>0$ and $\Lambda^{\tau}$ is constant. On $\left[\phi^{\tau},\gamma\right]$,
$u^{\tau}=0$, where $\Lambda^{\tau}$ may vary. On $\left(\gamma,1\right]$,
$u^{\tau}>0$ and $\Lambda^{\tau}=0$ is constant.
\end{proof}

\subsection{Proof of Corollary \ref{Corollary:V-Fixed-tau}}
\begin{proof}
Existence and optimality follow from Lemmas \ref{lem:p-tau-case-1}--\ref{lem:p-tau-case-3}.
It remains to prove uniqueness and consider the case $\tau=0$. Fix
$\tau>0$. Let $\Pi\left(\alpha,u\right):=\int_{\tau}^{1}\left[\theta-kM\left(\alpha\left(\theta\right),\theta\right)-c\alpha\left(\theta\right)-u\left(\theta\right)\right]f\left(\theta\right)d\theta$.
If $\left(\tilde{\alpha},\tilde{u}\right)$ is feasible for \ref{eq:Program-v},
then $\Pi\left(\tilde{\alpha},\tilde{u}\right)\le\mathcal{L}\left(\tilde{\alpha},\Lambda^{\tau}\right)\le\mathcal{L}\left(\alpha^{\tau},\Lambda^{\tau}\right)=\Pi^{v}\left(\tau\right),$
where the last equality follows from Lemmas \ref{lem:p-tau-case-1}--\ref{lem:p-tau-case-3}.
If $\left(\tilde{\alpha},\tilde{u}\right)$ is also a solution, then
$\Pi\left(\tilde{\alpha},\tilde{u}\right)=\Pi^{v}\left(\tau\right)$,
which means that both inequalities above must hold with equality.
In particular, $\mathcal{L}\left(\tilde{\alpha},\Lambda^{\tau}\right)=\mathcal{L}\left(\alpha^{\tau},\Lambda^{\tau}\right)$.
The proofs of Lemmas \ref{lem:p-tau-case-1}--\ref{lem:p-tau-case-3}
establish that, for every $\theta\in[\tau,1]$, $\alpha^{\tau}(\theta)$
is the unique pointwise maximizer of $\tilde{\psi}^{v}(\cdot,\theta;\Lambda^{\tau}(\theta))$
on $[0,1]$, so $\tilde{\alpha}$ must equal $\alpha^{\tau}$ almost
everywhere. By the envelope condition in (\ref{eq:u'-v}) and $u\left(\tau\right)=0$,
$\tilde{u}$ must then also equal $u^{\tau}$ almost everywhere. Therefore,
the solution to \ref{eq:Program-v} is unique and equals $\left(\alpha^{\tau},u^{\tau}\right)$.

Finally, consider $\tau=0$. Repeating the derivation of (\ref{eq:Lagrangian})
without imposing $u(\tau)=0$ yields, for any feasible $(\alpha,u)$,
$\Pi(\alpha,u)\leq u(0)\bigl[\Lambda^{0}(0)-1\bigr]+\mathcal{L}(\alpha,\Lambda^{0}),$
with $\Lambda^{0}$ being $\Lambda^{\tau}$ in (\ref{eq:Lambda-tau})
set at $\tau=0$. I claim that $\Lambda^{0}\left(0\right)<1$. If
$\Lambda^{0}\left(0\right)=0$, this is immediate. Otherwise, $\Lambda^{0}=L\left(\cdot,\phi^{0}\right)$
for some $\phi^{0}\in\left(0,\gamma\right)$, so $\Lambda^{0}(0)\leq L^{0}(0)=1-f(0)(1-F(\gamma))/f(\gamma)<1$.
Let $\alpha^{0}$ be the allocation $\alpha^{\tau}$ defined in the
corollary with $\tau$ set at $0$, and $u^{0}\left(\theta\right)=\int_{0}^{\theta}1-kM_{2}\left(\alpha^{0}\left(s\right),s\right)ds$---i.e.,
$u^{0}\left(0\right)$ is set at $0$. For any feasible $\left(\alpha,u\right)$
with $u\left(0\right)>0$, $\Pi(\alpha,u)\leq u(0)\bigl[\Lambda^{0}(0)-1\bigr]+\mathcal{L}(\alpha,\Lambda^{0})<\mathcal{L}(\alpha,\Lambda^{0})\le\mathcal{L}(\alpha^{0},\Lambda^{0})=\Pi(\alpha^{0},u^{0}),$
where the strict inequality follows from $\Lambda^{0}(0)<1$ and $u\left(0\right)>0$.
Hence, any $(\alpha,u)$ with $u(0)>0$ is strictly suboptimal. The
$\tau>0$ argument above then applies for $\tau=0$ as well. 
\end{proof}

\subsection{Proof of Lemma \ref{Lemma:Pi-v(v)-quasiconcave}}
\begin{proof}
Since $f>0$ on $\Theta$ and $\Pi^{v'}\left(\tau\right)=-f\left(\tau\right)\Psi^{v}\left(\tau\right)$,
$\Pi^{v}$ is strictly quasiconcave if $\Psi^{v}$ is strictly increasing.
This follows from differentiating $\Psi^{v}$ while taking into account
the three possible cases in Lemmas \ref{lem:p-tau-case-1}-\ref{lem:p-tau-case-3}.
Details are provided in appendix \ref{OA:Quasiconcavity-Pi-v}. It
remains to show that $\Psi^{v}(0)<0<\Psi^{v}(1)$. First,\emph{ }$\Psi^{v}(0)=-1/f(0)<0$.
Next, $\alpha^{\tau}\left(1\right)=\alpha^{FB}\left(1\right)$ for
all $\tau$. Therefore, $\Psi^{v}(1)=(1-c)+c\delta^{FB}-km(\delta^{FB})>1-c>0$
(strict convexity with $m(0)=m'(0)=0$ gives $km(\delta^{FB})<km'(\delta^{FB})\delta^{FB}=c\delta^{FB}$). 
\end{proof}

\subsection{Proof of Proposition \ref{Proposition:Mechanism-V} and Corollary
\ref{Corollary:Basic}}
\begin{proof}
By Lemma \ref{Lemma:Pi-v(v)-quasiconcave}, $\Pi^{v}$ has an interior
global maximizer characterized by the FOC, which is $\Psi^{v}\left(\tau^{v}\right)=0$.
The allocation and utility schedules follow from Corollary \ref{Corollary:V-Fixed-tau}. 

\bigskip

\emph{Case $\eta\geq\gamma$.} Corollary \ref{Corollary:V-Fixed-tau}
gives $\alpha^{\tau^{v}}=a^{R}$. By Lemma \ref{lem:Relaxed-program},
$a^{R}\left(\theta\right)\le\alpha^{FB}\left(\theta\right)$ for all
$\theta\in\left[\tau,1\right]$, with strict inequality whenever $\theta<1$
and $\alpha^{FB}\left(\theta\right)>0$. 

By the envelope condition (\ref{eq:u'-v}), $u^{\tau^{v}\prime}(\theta)=1-km'(\delta^{R}(\theta))$.
By Lemma \ref{Lemma:-eta-gamma}, $b^{R}\leq\eta$. For $\theta\in[\tau^{v},b^{R}]$,
$\delta^{R}(\theta)=\theta\leq b^{R}\leq\eta$; for $\theta\in(b^{R},1]$,
$\delta^{R}(\theta)=d^{R}(\theta)$ is strictly decreasing from $b^{R}\leq\eta$.
Hence $km'(\delta^{R}(\theta))\leq1$ with equality at most at $\theta=b^{R}=\eta$.
So $u^{\tau^{v}}$ is strictly increasing on $[\tau^{v},1]$.

Next, I show that $\tau^{v}<b^{R}$, which implies a positive measure
of types pooling at $y=0$. Let $J=(1-F)/f$, so $m'(b^{R})-m''(b^{R})J(b^{R})=c/k$.
Hence
\[
J(b^{R})=\frac{m'(b^{R})-c/k}{m''(b^{R})}\le\frac{m'(b^{R})}{m''(b^{R})}\le b^{R},
\]
where the last inequality follows from $m'(0)=0$ and $m'''\ge0$.
Since $\eta\ge b^{R}$, we also have $km'(b^{R})\le1$. Therefore,
\[
\begin{aligned}\Psi^{v}(b^{R}) & =b^{R}-km(b^{R})-\left[1-km'(b^{R})\right]J(b^{R})\\
 & >b^{R}\left[1-km'(b^{R})\right]-\left[1-km'(b^{R})\right]J(b^{R})\ =\ \left[1-km'(b^{R})\right]\left[b^{R}-J(b^{R})\right]\ge0.
\end{aligned}
\]
The strict inequality uses strict convexity of $m$ and $m(0)=0$,
which imply $m(b^{R})<b^{R}m'(b^{R})$. Thus $\Psi^{v}(b^{R})>0$,
and so $\tau^{v}<b^{R}$.

\bigskip

\emph{Case $\eta<\gamma$.} I first show that $\tau^{v}<\eta<b^{R}$.
If $\tau^{v}\in[\eta,\gamma)$, from the proof of Lemma \ref{Lemma:Pi-v(v)-quasiconcave}
case 2(ii), $\Psi^{v}(\tau^{v})=(1-c)\tau^{v}+c\eta-km(\eta).$ Since
$km'(\eta)=1$ and strict convexity with $m(0)=0$ implies $m(\eta)<\eta m'(\eta)=\eta/k$,
it implies $\Psi^{v}(\eta)=\eta-km(\eta)>0$. Therefore, $\tau^{v}<\eta$.
By Lemma \ref{Lemma:-eta-gamma}, $\eta<\gamma$ implies $\eta<b^{R}$.
Therefore, $\tau^{v}<b^{R}$ as well.

If $\tau^{v}\le\underline{\tau}$, then $\alpha^{\tau^{v}}=a^{R}$
by Corollary \ref{Corollary:V-Fixed-tau}. Since $\tau^{v}<b^{R}$,
there is a positive measure of types pooling at $y=0$. Since $\tau^{v}<\eta<\gamma$,
the slope of the induced utility schedule is positive on $(\tau^{v},\eta)$,
negative on $(\eta,\gamma)$, and positive again on $(\gamma,1)$.
Thus $u^{\tau^{v}}$ is nonmonotone. The property that $a^{R}\left(\theta\right)\le\alpha^{FB}\left(\theta\right)$
for all $\theta\in\left[\tau,1\right]$ has been established for the
case $\eta\ge\gamma$.

If $\tau^{v}\in\left(\underline{\tau},\eta\right)$, then Lemma \ref{lem:p-tau-case-3}
applies and $\alpha^{\tau^{v}}$ is characterized in (\ref{eq:alpha-tau-case-3}),
with $\tau=\tau^{v}$. In such cases, the proof of Lemma \ref{lem:p-tau-case-3}
gives $u^{\tau}(\theta)>0$ in $\left(\tau,\phi^{\tau}\right)$, $u^{\tau}\left(\theta\right)=0$
on $\left[\phi^{\tau},\gamma\right]$, and $u^{\tau}\left(\theta\right)>0$
on $(\gamma,1]$, so $u^{\tau^{v}}$ is nonmonotone. Lemma \ref{lem:p-tau-case-3}
also states that $\beta^{\tau}>\tau$. Therefore, the set $[\tau^{v},\beta^{\tau^{v}})$
is a positive measure set of types pooling at $y=0$. 

For $\theta\in\left(\beta^{\tau^{v}},\gamma\right)$, $\alpha^{\tau^{v}}\left(\theta\right)=a\left(\theta;\phi^{\tau^{v}}\right)$,
which is characterized by the FOC on the adjusted virtual surplus:
\[
m'(\theta-\alpha^{\tau^{v}}\left(\theta\right))-m''(\theta-\alpha^{\tau^{v}}\left(\theta\right))\frac{1-F(\theta)-\Lambda^{\tau}\left(\theta\right)}{f(\theta)}=\frac{c}{k}.
\]
$\left(1-F-\Lambda^{\tau}\right)/f$ is decreasing for $\tau\in\left(\underline{\tau},\eta\right)$
from the proof of Lemmas \ref{lem:p-tau-case-2}-\ref{lem:p-tau-case-3}.
Following the argument in Lemma \ref{lem:Relaxed-program} to establish
$a^{R}\left(\theta\right)<\alpha^{FB}\left(\theta\right)$, with $\left(1-F\right)/f$
replaced by $\left(1-F-\Lambda^{\tau}\right)/f$, yields $\alpha^{\tau^{v}}\left(\theta\right)<\alpha^{FB}\left(\theta\right)$. 

For $\theta\in\left[\gamma,1\right]$, $\alpha^{\tau^{v}}\left(\theta\right)=a^{R}\left(\theta\right)\le\alpha^{FB}\left(\theta\right)$
with equality only at $\theta=1$. 

Finally, $\beta^{\tau^{v}}<b^{R}$ from Lemma \ref{lem:p-tau-case-3},
so $\alpha^{\tau^{v}}\left(\theta\right)<\alpha^{FB}\left(\theta\right)$
whenever $\alpha^{FB}\left(\theta\right)>0$.
\end{proof}
\newpage 

\begin{center} 

{\Large \bf 
Online Appendix }

\bigskip

\end{center}

\section{\protect\label{OA:Proofs-Auxiliary-Results}Proofs for Auxiliary
Results }

\subsection{Proof of Proposition \ref{Proposition:h-CompStatics}}
\begin{proof}
Throughout, statements involving $\theta^{*}$ apply when $\theta^{*}$
exists. 

\textit{Part (1). }\textit{\emph{If $\Psi^{h}(1)>0$ and $\tau^{h}=1$,
local changes in $c$, $k$ or $v_{0}$ imply that $\Psi^{h}(1)$
is still strictly positive and thus $\tau^{h}$ is still 1. Henceforth,
assume that $\tau^{h}<1$; thus $\Psi^{h}(\tau^{h})=0$ and $\Psi^{h'}(\tau^{h})<0$.
By the implicit function theorem, $\operatorname{sign}\!\left(d\tau^{h}/dz\right)=\operatorname{sign}\!\left(\partial\Psi^{h}(\tau^{h})/\partial z\right)$
for $z\in\left\{ c,k,v_{0}\right\} $. }}Since $\alpha^{h}(\tau)$
pointwise maximizes $\psi^{h}(\cdot,\tau)$ on $[0,\tau]$, the envelope
theorem gives
\[
\frac{\partial\Psi^{h}}{\partial c}=-\alpha^{h}(\tau),\qquad\frac{\partial\Psi^{h}}{\partial k}=-m\!\left(\tau-\alpha^{h}(\tau)\right)-m'\!\left(\tau-\alpha^{h}(\tau)\right)\frac{F(\tau)}{f(\tau)},\qquad\frac{\partial\Psi^{h}}{\partial v_{0}}=1.
\]
The first is strictly negative if $\tau>\beta^{h}$ and is zero otherwise.
The second is strictly negative if $\tau>\alpha^{h}\left(\tau\right)$---iff
$\tau<\theta^{*}$---and zero otherwise. 

\textit{Part (2). }Fix $\theta>0$, write $R(\theta)=F(\theta)/f(\theta)$
and $r:=c/k$. By Lemma \ref{lem:alpha-h} and the characterizations
of $\beta^{h}$ and $\theta^{*}$, $\theta\geq\theta^{*}$ if and
only if $r\leq m''(0)R(\theta)$, and $\theta\leq\beta^{h}$ if and
only if $r\geq m'(\theta)+m''(\theta)R(\theta)$. Hence, as a function
of $r$,
\[
\delta^{h}(\theta)=\begin{cases}
0 & \text{if }r\leq m''(0)R(\theta),\\
d^{h}(\theta) & \text{if }m''(0)R(\theta)<r<m'(\theta)+m''(\theta)R(\theta),\\
\theta & \text{if }r\geq m'(\theta)+m''(\theta)R(\theta),
\end{cases}
\]
where $d^{h}(\theta)$ solves (\ref{eq:delta-h}). The left-hand side
of (\ref{eq:delta-h}) is strictly increasing in $d^{h}$, its derivative
being $m''(d^{h})+m'''(d^{h})R(\theta)>0$, so in the middle regime
$d^{h}(\theta)$ is strictly increasing in $r$, with $d^{h}(\theta)\downarrow0$
as $r\downarrow m''(0)R(\theta)$ and $d^{h}(\theta)\uparrow\theta$
as $r\uparrow m'(\theta)+m''(\theta)R(\theta)$. Therefore $\delta^{h}(\theta)$
is continuous and weakly increasing in $r$ on its whole range, strictly
increasing exactly in the middle regime $\beta^{h}<\theta<\theta^{*}$,
and locally constant in the two corner regimes: at $0$ when $\theta>\theta^{*}$,
since $m''(0)R(\theta)>r$ holds strictly and is preserved by small
changes in $r$, and at $\theta$ when $\theta<\beta^{h}$, by the
symmetric argument. Since neither $\beta^{h}$, $\theta^{*}$, nor
(\ref{eq:delta-h}) involves $v_{0}$, $\delta^{h}(\theta)$ is independent
of $v_{0}$.

\textit{Part (3). }$\Pi^{h}(\tau^{h})=\max_{\tau\in[0,1]}\int_{0}^{\tau}\psi^{h}(\alpha^{h}(\theta),\theta)f(\theta)d\theta.$
By the envelope theorem,
\[
\frac{\partial\Pi^{h}\left(\tau^{h}\right)}{\partial v_{0}}=F(\tau^{h})>0,\qquad\frac{\partial\Pi^{h}\left(\tau^{h}\right)}{\partial c}=-\int_{0}^{\tau^{h}}\alpha^{h}(\theta)f(\theta)\,d\theta\le0
\]
\[
\frac{\partial\Pi^{h}\left(\tau^{h}\right)}{\partial k}=-\int_{0}^{\tau^{h}}\left[m(\delta^{h}(\theta))+m'(\delta^{h}(\theta))\frac{F(\theta)}{f(\theta)}\right]f(\theta)\,d\theta<0
\]
The first inequality is trivial. For the second, the inequality is
strict if $\tau^{h}>\beta^{h}$, implying a positive measure of types
below $\tau^{h}$ with $\alpha^{h}\left(\theta\right)>0$; if $\tau^{h}\le\beta^{h}$,
then $\alpha^{h}\left(\theta\right)=0$ on $\left[0,\tau^{h}\right]$.
For the third, the integrand is nonnegative everywhere and strictly
positive on the pooling region, where $\alpha^{h}\left(\theta\right)=0$
and $\delta^{h}(\theta)=\theta>0$; this region has positive measure
since $\beta^{h}>0$. 
\end{proof}

\subsection{Proof of Lemma \ref{Lemma:-eta-gamma}}
\begin{proof}
Equations (\ref{eq:eta}) and (\ref{eq:gamma}) yield $m'(\eta)=1/k=m'(d^{R}(\gamma))$.
Since $m'$ is strictly increasing, this implies $\eta=d^{R}(\gamma)$.
By definition, $d^{R}(b^{R})=b^{R}$. Comparing $\gamma$ to $b^{R}$
via these two properties, and noting that $d^{R}$ is strictly decreasing
yields the following three mutually exclusive and exhaustive cases:
(1) if $\gamma<b^{R}$, then $\eta=d^{R}(\gamma)>d^{R}(b^{R})=b^{R}$,
yielding $\gamma<b^{R}<\eta$; (2) if $\gamma>b^{R}$, then $\eta=d^{R}(\gamma)<d^{R}(b^{R})=b^{R}$,
yielding $\eta<b^{R}<\gamma$; (3) if $\gamma=b^{R}$, then $\eta=d^{R}(b^{R})=b^{R}$,
yielding $\eta=b^{R}=\gamma$.

It remains to show the comparative statics of $\gamma-\eta$. $m'(\eta)=1/k$,
so $\eta$ is independent of $c$ and decreasing in $k$. As noted
above, $\eta=d^{R}(\gamma)$. Substituting into (\ref{eq:delta-R})
yields $\left[1-F\left(\gamma\right)\right]/f\left(\gamma\right)=(1-c)/[km''(\eta)]$.
The RHS decreases in $c$ and since the LHS is decreasing, $\gamma$
increases with $c$. Thus, $\gamma-\eta$ increases with $c$. 

For $k$, $km''(\eta(k))$ is strictly increasing. Its derivative
is $m''(\eta)-m'''(\eta)m'(\eta)/m''(\eta)>0$, where the inequality
follows from the strict log-concavity of $m'$. Hence $(1-c)/[km''(\eta(k))]$
decreases in $k$. Since $(1-F)/f$  is decreasing, $\gamma$ increases
with $k$. Together with $\eta$ decreasing with $k$, this shows
that $\gamma-\eta$ increases with $k$. 
\end{proof}

\subsection{Proof of Proposition \ref{Proposition:Vertical-CS}}

\emph{Preliminaries.} Throughout, I slightly abuse notations by adding
the parameter $c$ and $k$ into existing notations as necessary---i.e.,
$\Psi^{v}\left(\tau;c\right)$, $\gamma\left(c\right)$, $L^{0}\left(\theta;c\right)$,
$\phi^{\tau}\left(c\right)$, and so on when $c$ varies with $k$
fixed, and analogously when $k$ varies with $c$ fixed. Write $J:=(1-F)/f$.
Recall that $\eta=(m')^{-1}(1/k)$ does not depend on $c$, and define
\[
\Psi_{0}(\tau;k):=\psi^{v}(0,\tau)=\tau-km(\tau)-\left[1-km'(\tau)\right]J(\tau),
\]
which does not depend on $c$ because the production-cost term vanishes
at $y=0$.

For a constant $L\in[0,1-F(\theta))$, the proof of Lemma \ref{lem:p-tau-case-3}
shows $\tilde{\psi}^{v}(\cdot,\theta;L)$ has a unique maximizer on
$Y$, attained in $[0,\theta]$. Write $\delta(\theta;L)$ for $\theta$
minus this maximizer and $\hat{J}(\theta;L):=(1-F(\theta)-L)/f(\theta)$.
With $G(d)=[m'(d)-c/k]/m''(d)$ from the proof of Lemma \ref{lem:Relaxed-program},
which is strictly increasing by (\ref{eq:G'>0}), 
\begin{equation}
\delta(\theta;L)=\theta\text{ if }G(\theta)\le\hat{J}(\theta;L),\text{ and }G(\delta(\theta;L))=\hat{J}(\theta;L)\text{ otherwise}.\label{eq:CS-Prop-1}
\end{equation}
In this notation, $\delta(\theta;0)=\delta^{R}(\theta)$ and $\delta(\theta;L^{0}(\phi))=\theta-a(\theta;\phi)$.
I first establish five claims.
\begin{claim}
\label{clm:psi0}(i) $\Psi_{0}(\cdot;k)$ is strictly increasing on
$[0,\eta]$ with $\Psi_{0}(0;k)<0<\Psi_{0}(\eta;k)$; it has a unique
zero $\tau_{0}(k)\in(0,\eta)$; and $\tau_{0}(k)<b^{R}$ for every
admissible $c$. (ii) At the optimum, $\tau^{v}<\eta$, $\alpha^{\tau^{v}}(\tau^{v})=0$,
and 
\begin{equation}
\Psi^{v}(\tau^{v})=\Psi_{0}(\tau^{v})+\left[1-km'(\tau^{v})\right]\frac{\Lambda^{\tau^{v}}(\tau^{v})}{f(\tau^{v})}=0,\label{eq:CS-Prop-2}
\end{equation}
where $\Lambda^{\tau^{v}}(\tau^{v})=0$ if the Lemma \ref{lem:p-tau-case-1}
case applies at $\tau^{v}$, and $\Lambda^{\tau^{v}}(\tau^{v})=L^{0}(\phi^{\tau^{v}})>0$
if the Lemma \ref{lem:p-tau-case-3} case applies; the Lemma \ref{lem:p-tau-case-2}
case never applies at $\tau^{v}$. (iii) The Lemma \ref{lem:p-tau-case-3}
case applies at $\tau^{v}$ if and only if a positive measure of included
types have binding IR; in that case $\tau^{v}<\tau_{0}$, and otherwise
$\tau^{v}=\tau_{0}$.
\end{claim}
\begin{proof}
(i) Differentiating, $\Psi_{0}'(\tau)=[1-km'(\tau)][1-J'(\tau)]+km''(\tau)J(\tau)$.
On $[0,\eta]$, $1-km'(\tau)\ge0$ and $J'<0$ (Assumption \ref{Assumption:F}),
so the first term is nonnegative, while $km''(\tau)J(\tau)>0$ for
$\tau<1$. Hence $\Psi_{0}$ is strictly increasing on $[0,\eta]$.
$\Psi_{0}(0)=-J(0)<0$. Strict convexity with $m(0)=0$ gives $km(\eta)<\eta km'(\eta)=\eta$,
so $\Psi_{0}(\eta)=\eta-km(\eta)>0$; hence $\tau_{0}\in(0,\eta)$
exists and is unique. For $\tau_{0}<b^{R}$: if $b^{R}>\eta$, this
is immediate from $\tau_{0}<\eta$. If $b^{R}\le\eta$, then $d^{R}(b^{R})=b^{R}$
gives $J(b^{R})=G(b^{R})\le m'(b^{R})/m''(b^{R})\le b^{R}$, where
the last step uses $m'(z)=\int_{0}^{z}m''\le zm''(z)$ (since $m'''\ge0$).
With $km(b^{R})<b^{R}km'(b^{R})$ and $km'(b^{R})\le1$, 
\[
\Psi_{0}(b^{R})=b^{R}-km(b^{R})-\left[1-km'(b^{R})\right]J(b^{R})>\left[1-km'(b^{R})\right]\left[b^{R}-J(b^{R})\right]\ge0,
\]
so $\tau_{0}<b^{R}$. 

(ii) The proof of Proposition \ref{Proposition:Mechanism-V} establishes
$\tau^{v}<\eta$ and $\tau^{v}<b^{R}$; the Lemma \ref{lem:p-tau-case-2}
case requires $\tau\ge\eta$ and so cannot apply at $\tau^{v}$. In
the Lemma \ref{lem:p-tau-case-1} case, $\alpha^{\tau^{v}}(\tau^{v})=a^{R}(\tau^{v})=0$
(since $\tau^{v}<b^{R}$) and $\Lambda^{\tau^{v}}=0$, so $\Psi^{v}(\tau^{v})=\psi^{v}(0,\tau^{v})=\Psi_{0}(\tau^{v})$.
In the Lemma \ref{lem:p-tau-case-3} case, $\alpha^{\tau^{v}}(\tau^{v})=0$
(since $\beta^{\tau}>\tau$) and $\Lambda^{\tau^{v}}(\tau^{v})=L^{0}(\phi^{\tau^{v}})$
by (\ref{eq:L(theta;phi)}); evaluating $\tilde{\psi}^{v}$ at $y=0$
gives (\ref{eq:CS-Prop-2}). For positivity: write $L^{0}(\theta)=(1-F(\theta))\left[1-h(\theta)/h(\gamma)\right]$
with $h:=f/(1-F)$, as in the proof of Lemma \ref{lem:p-tau-case-2};
strict log-concavity of $1-F$ makes $h$ strictly increasing, so
$L^{0}(\theta)>0$ for $\theta<\gamma$, and $\phi^{\tau^{v}}<\gamma$.

(iii) In the Lemma \ref{lem:p-tau-case-3} case, $u^{\tau^{v}}=0$
on $[\phi^{\tau^{v}},\gamma]$, which is a positive-measure set. In
the Lemma \ref{lem:p-tau-case-1} case, $u^{\tau^{v}}=u^{R}(\cdot;\tau^{v})$,
which is continuously differentiable. If $u^{R}=0$ on a set $E$
of positive measure, then almost every point of $E$ is an accumulation
point of $E$, at which $(u^{R})'=0$, i.e., $\delta^{R}(\theta)=\eta$.
But $\delta^{R}(\theta)=\theta$ is strictly increasing on $[0,b^{R}]$
and $\delta^{R}(\theta)=d^{R}(\theta)$ is strictly decreasing on
$(b^{R},1]$ (proof of Lemma \ref{lem:Relaxed-program}), so $\{\delta^{R}=\eta\}$
contains at most two points---a contradiction. Hence IR binds on
a null set in the Lemma \ref{lem:p-tau-case-1} case. Finally, by
(\ref{eq:CS-Prop-2}), $\Psi_{0}(\tau^{v})=0$ in the Lemma \ref{lem:p-tau-case-1}
case, so $\tau^{v}=\tau_{0}$ by (i); and $\Psi_{0}(\tau^{v})=-\left[1-km'(\tau^{v})\right]L^{0}(\phi^{\tau^{v}})/f(\tau^{v})<0$
in the Lemma \ref{lem:p-tau-case-3} case, so $\tau^{v}<\tau_{0}$.
\end{proof}
\begin{claim}
\label{clm:pointwise}Fix $\theta\in(0,1)$ and $L,L'\in[0,1-F(\theta))$
with $\hat{J}(\theta;L)>0$. (i) If $L'\ge L$, then $\delta(\theta;L')\le\delta(\theta;L)$.
(ii) At fixed $k$, if $c_{2}>c_{1}$, then $\delta(\theta;L;c_{2})\ge\delta(\theta;L;c_{1})$,
with strict inequality if $\delta(\theta;L;c_{1})<\theta$. (iii)
At fixed $c$, if $k_{2}>k_{1}$, then $k_{2}m'(\delta(\theta;L;k_{2}))\ge k_{1}m'(\delta(\theta;L;k_{1}))$,
with strict inequality if $\delta(\theta;L;k_{1})>0$. 
\end{claim}
\begin{proof}
(i) $\hat{J}(\theta;L')\le\hat{J}(\theta;L)$. If $G(\theta)\le\hat{J}(\theta;L')$,
both are at the corner $\theta$. If only $\delta(\theta;L)$ is at
the corner, $\delta(\theta;L')<\theta=\delta(\theta;L)$. If both
are interior, $G(\delta(\theta;L'))=\hat{J}(\theta;L')\le\hat{J}(\theta;L)=G(\delta(\theta;L))$
and $G$ strictly increasing give the result.

(ii) $G(d;c)$ is strictly decreasing in $c$ at every $d$, while
$\hat{J}(\theta;L)$ does not vary with $c$. If $\delta(\theta;L;c_{1})=\theta$,
then $G(\theta;c_{2})<G(\theta;c_{1})\le\hat{J}(\theta;L)$, so $\delta(\theta;L;c_{2})=\theta$.
If $\delta(\theta;L;c_{1})$ is interior, then either $\delta(\theta;L;c_{2})=\theta>\delta(\theta;L;c_{1})$,
or both are interior and $G(\delta(\theta;L;c_{2});c_{2})=\hat{J}(\theta;L)=G(\delta(\theta;L;c_{1});c_{1})>G(\delta(\theta;L;c_{1});c_{2})$
forces $\delta(\theta;L;c_{2})>\delta(\theta;L;c_{1})$.

(iii) Write $\delta_{i}:=\delta(\theta;L;k_{i})$ and $x:=k_{1}m'(\delta_{1})$.
If $\delta_{1}=0$ the claim is trivial. Suppose $\delta_{1}>0$ and
define $\hat{\delta}:=(m')^{-1}(x/k_{2})\in(0,\delta_{1})$, so $k_{2}m'(\hat{\delta})=k_{1}m'(\delta_{1})$.
Strict log-concavity of $m'$ means $m''/m'$ is strictly decreasing,
so 
\begin{equation}
k_{2}m''(\hat{\delta})=k_{1}m'(\delta_{1})\cdot\frac{m''(\hat{\delta})}{m'(\hat{\delta})}>k_{1}m'(\delta_{1})\cdot\frac{m''(\delta_{1})}{m'(\delta_{1})}=k_{1}m''(\delta_{1}).\label{eq:CS-Prop-3}
\end{equation}
I show $\delta_{2}>\hat{\delta}$, which yields $k_{2}m'(\delta_{2})>k_{2}m'(\hat{\delta})=k_{1}m'(\delta_{1})$.
If $\delta_{2}=\theta$, then $\delta_{2}\ge\delta_{1}>\hat{\delta}$.
Otherwise $\delta_{2}$ is interior, so $G(\delta_{2};k_{2})=\hat{J}(\theta;L)>0$,
and it suffices to show $G(\hat{\delta};k_{2})<\hat{J}(\theta;L)$.
If $x\le c$, then $G(\hat{\delta};k_{2})=(x-c)/(k_{2}m''(\hat{\delta}))\le0<\hat{J}(\theta;L)$.
If $x>c$, then by (\ref{eq:CS-Prop-3}), 
\[
G(\hat{\delta};k_{2})=\frac{x-c}{k_{2}m''(\hat{\delta})}<\frac{x-c}{k_{1}m''(\delta_{1})}=G(\delta_{1};k_{1})\le\hat{J}(\theta;L),
\]
where the last inequality is the first-order condition if $\delta_{1}$
is interior and the corner condition if $\delta_{1}=\theta$. 
\end{proof}
\begin{claim}
\label{clm:irmonotone}Suppose the Lemma \ref{lem:p-tau-case-3} case
applies to \ref{eq:Program-v} at parameters $p_{1}=(c_{1},k_{1})$,
and let $p_{2}$ raise exactly one of $c$ or $k$. Then $\eta(k_{2})<\gamma(p_{2})$
and $u^{R}(\gamma(p_{2});\tau;p_{2})<0$. In particular, if $\tau<\eta(k_{2})$,
the Lemma \ref{lem:p-tau-case-3} case applies to $P^{v}(\tau)$ at
$p_{2}$. 
\end{claim}
\begin{proof}
Write $\eta_{i}:=\eta(k_{i})$ and $\gamma_{i}:=\gamma(p_{i})$. From
the proof of Lemma \ref{Lemma:-eta-gamma}, $J(\gamma)=(1-c)/(km''(\eta(k)))$
is strictly decreasing in $c$ and in $k$ (the latter because $km''(\eta(k))$
is strictly increasing in $k$), so $\gamma_{2}>\gamma_{1}$; since
$\eta_{2}\le\eta_{1}<\gamma_{1}<\gamma_{2}$, $\eta_{2}<\gamma_{2}$.

Since $\delta^{R}(\tau')\le\tau'<\eta_{1}$, $\partial u^{R}(\gamma_{1};\tau';p_{1})/\partial\tau'=-\left[1-k_{1}m'(\delta^{R}(\tau'))\right]<0$
for $\tau'<\eta_{1}$, so $u^{R}(\gamma_{1};\cdot;p_{1})$ is strictly
decreasing on $[0,\eta_{1}]$; with $u^{R}(\gamma_{1};\underline{\tau}(p_{1});p_{1})\le0$
by the definition of $\underline{\tau}$ in (\ref{eq:tau_underline})
and continuity, $\tau\in(\underline{\tau}(p_{1}),\eta_{1})$ gives
$u^{R}(\gamma_{1};\tau;p_{1})<0$.

For every $s$, $k_{2}m'(\delta^{R}(s;p_{2}))\ge k_{1}m'(\delta^{R}(s;p_{1}))$:
in the $c$-direction this follows from Claim \ref{clm:pointwise}(ii)
with $L=0$ and $m'$ increasing; in the $k$-direction from Claim
\ref{clm:pointwise}(iii) with $L=0$ (note $\hat{J}(s;0)=J(s)>0$
for $s<1$). Hence $u^{R}(\theta;\tau;p_{2})\le u^{R}(\theta;\tau;p_{1})$
for every $\theta$. Moreover, for $s\in(\gamma_{1},\gamma_{2}]$,
$\delta^{R}(s;p_{2})\ge\eta_{2}$: if $\delta^{R}(s;p_{2})=s$ this
follows from $s>\gamma_{1}>\eta_{1}\ge\eta_{2}$; if interior, $d^{R}(\cdot;p_{2})$
is strictly decreasing (proof of Lemma \ref{lem:Relaxed-program})
with $d^{R}(\gamma_{2};p_{2})=\eta_{2}$ (proof of Lemma \ref{Lemma:-eta-gamma}).
Hence $1-k_{2}m'(\delta^{R}(s;p_{2}))\le0$ on $(\gamma_{1},\gamma_{2}]$,
and 
\[
u^{R}(\gamma_{2};\tau;p_{2})\le u^{R}(\gamma_{1};\tau;p_{2})\le u^{R}(\gamma_{1};\tau;p_{1})<0.
\]
The last statement of the claim follows because $\{\tau'\in[0,\eta_{2}]:u^{R}(\gamma_{2};\tau';p_{2})<0\}=(\underline{\tau}(p_{2}),\eta_{2}]$
by the same monotonicity-in-$\tau'$ argument applied at $p_{2}$. 
\end{proof}
\begin{claim}
\label{clm:multiplier}Under the hypotheses of Claim \ref{clm:irmonotone},
and supposing additionally that the Lemma \ref{lem:p-tau-case-3}
case applies to \ref{eq:Program-v} at $p_{2}$, $L^{0}(\phi^{\tau}(p_{2});p_{2})>L^{0}(\phi^{\tau}(p_{1});p_{1})$. 
\end{claim}
\begin{proof}
By Lemma \ref{lem:p-tau-case-3}, $\phi_{i}:=\phi^{\tau}(p_{i})\in(\eta_{i},\gamma_{i})$
is well defined, and with $\ell_{i}:=L^{0}(\phi_{i};p_{i})$ and $\delta_{i}(s;L):=\delta(s;L;p_{i})$,
\begin{equation}
U(\phi_{i};\tau)=\int_{\tau}^{\phi_{i}}\left[1-k_{i}m'(\delta_{i}(s;\ell_{i}))\right]ds=0.\label{eq:CS-Prop-4}
\end{equation}
Also $\delta_{i}(\phi_{i};\ell_{i})=\eta_{i}$: $\hat{J}(\phi_{i};\ell_{i})=J(\gamma_{i})=G(\eta_{i};p_{i})$
by the definitions of $L^{0}$ and $\gamma$, and $\phi_{i}>\eta_{i}$
rules out the corner in (\ref{eq:CS-Prop-1}). Since $\hat{J}(\cdot;\ell_{i})$
is nonincreasing on $[\tau,\phi_{i}]$ (appendix \ref{OA:Adjusted-hazardrate}),
$\hat{J}(s;\ell_{i})\ge J(\gamma_{i})>0$ there.

Suppose toward a contradiction that $\ell_{2}\le\ell_{1}$. Since
$L^{0}(\theta;p)=1-F(\theta)-f(\theta)J(\gamma(p))$ is strictly increasing
in $p$ pointwise ($J(\gamma(p))$ is strictly decreasing, as in Claim
\ref{clm:irmonotone}) and strictly decreasing in $\theta$ on $[0,\gamma(p))$
(proof of Lemma \ref{lem:p-tau-case-2}, with $h$ strictly increasing),
$L^{0}(\phi_{2};p_{2})=\ell_{2}\le\ell_{1}=L^{0}(\phi_{1};p_{1})<L^{0}(\phi_{1};p_{2})$
implies $\phi_{2}>\phi_{1}$.

For $s\in(\tau,\phi_{1}]$, combining Claim \ref{clm:pointwise}(i)
($\delta_{2}(s;\ell_{2})\ge\delta_{2}(s;\ell_{1})$, as $\ell_{2}\le\ell_{1}$)
with Claim \ref{clm:pointwise}(ii)--(iii), 
\[
k_{2}m'(\delta_{2}(s;\ell_{2}))\ge k_{2}m'(\delta_{2}(s;\ell_{1}))\ge k_{1}m'(\delta_{1}(s;\ell_{1})).
\]
The second inequality is strict on $\{s\in(\tau,\phi_{1}]:\delta_{1}(s;\ell_{1})<s\}$,
which is an interval of positive length ending at $\phi_{1}$: the
corner condition $G(s;p_{1})\le\hat{J}(s;\ell_{1})$ crosses once
($G$ strictly increasing, $\hat{J}(\cdot;\ell_{1})$ nonincreasing),
and $\delta_{1}(\phi_{1};\ell_{1})=\eta_{1}<\phi_{1}$ shows $\delta_{1}$
is interior at $\phi_{1}$. (In the $k$-direction the inequality
is strict at every $s>0$.) Hence 
\[
\int_{\tau}^{\phi_{1}}\left[1-k_{2}m'(\delta_{2}(s;\ell_{2}))\right]ds<\int_{\tau}^{\phi_{1}}\left[1-k_{1}m'(\delta_{1}(s;\ell_{1}))\right]ds=0.
\]
For $s\in(\phi_{1},\phi_{2}]$, $\delta_{2}(s;\ell_{2})\ge\eta_{2}$:
if $\delta_{2}(s;\ell_{2})=s$, this follows from $s>\phi_{1}>\eta_{1}\ge\eta_{2}$;
if interior, $\delta_{2}(\cdot;\ell_{2})$ is nonincreasing on its
interior region ($G$ strictly increasing, $\hat{J}(\cdot;\ell_{2})$
nonincreasing on $[\tau,\phi_{2}]$ by appendix \ref{OA:Adjusted-hazardrate})
with $\delta_{2}(\phi_{2};\ell_{2})=\eta_{2}$. Hence $1-k_{2}m'(\delta_{2}(s;\ell_{2}))\le0$
on $(\phi_{1},\phi_{2}]$. Summing the two pieces gives $U(\phi_{2};\tau)<0$
at $p_{2}$, contradicting (\ref{eq:CS-Prop-4}). Hence $\ell_{2}>\ell_{1}$. 
\end{proof}
\begin{claim}
\label{clm:rent}If $\tau\in(0,\eta)$ and $X$ satisfies $\tau-km(\tau)=\left[1-km'(\tau)\right]X$,
then $m'(\tau)X-m(\tau)=\left[\tau m'(\tau)-m(\tau)\right]/\left[1-km'(\tau)\right]>0$. 
\end{claim}
\begin{proof}
Substitute $X$ and simplify; strict convexity with $m(0)=0$ gives
$m(\tau)<\tau m'(\tau)$, and $1-km'(\tau)>0$ since $\tau<\eta$. 
\end{proof}
\emph{Part 1: Proof for the effect of $c$ on $\tau^{v}$}
\begin{proof}
Fix $k$ and let $c_{1}<c_{2}$ lie in the maintained parameter region;
write $\tau_{i}:=\tau^{v}(c_{i})$ and $\Lambda_{i}:=\Lambda^{\tau_{i}}(\tau_{i};c_{i})$.
Note $\Psi_{0}$ does not depend on $c$.

Suppose first $\Lambda_{1}=0$. By Claim \ref{clm:psi0}, $\tau_{1}=\tau_{0}$,
and since $\tau_{0}<\eta$, either the Lemma \ref{lem:p-tau-case-1}
or the Lemma \ref{lem:p-tau-case-3} case applies to $\mathcal{P}^{v}(\tau_{0})$
at $c_{2}$. In the Lemma \ref{lem:p-tau-case-1} case, $\tau_{0}<b^{R}(c_{2})$
(Claim \ref{clm:psi0}(i)) gives $\Psi^{v}(\tau_{0};c_{2})=\psi^{v}(a^{R}(\tau_{0};c_{2}),\tau_{0})=\psi^{v}(0,\tau_{0})=\Psi_{0}(\tau_{0})=0$,
so $\tau_{2}=\tau_{0}=\tau_{1}$ by Lemma \ref{Lemma:Pi-v(v)-quasiconcave}.
In the Lemma \ref{lem:p-tau-case-3} case, $\Psi^{v}(\tau_{0};c_{2})=\Psi_{0}(\tau_{0})+\left[1-km'(\tau_{0})\right]L^{0}(\phi^{\tau_{0}}(c_{2});c_{2})/f(\tau_{0})>0$,
since $L^{0}(\phi;c_{2})>0$ for $\phi<\gamma(c_{2})$ and $km'(\tau_{0})<1$;
hence $\tau_{2}<\tau_{1}$.

Suppose instead $\Lambda_{1}>0$, so the Lemma \ref{lem:p-tau-case-3}
case applies at $(\tau_{1},c_{1})$ and $\Lambda_{1}=L^{0}(\phi^{\tau_{1}}(c_{1});c_{1})$.
By Claim \ref{clm:irmonotone} ($c$-direction, with $\tau_{1}<\eta$),
the Lemma \ref{lem:p-tau-case-3} case applies at $(\tau_{1},c_{2})$;
by Claim \ref{clm:multiplier}, $\ell_{2}:=L^{0}(\phi^{\tau_{1}}(c_{2});c_{2})>\Lambda_{1}$.
Hence 
\[
\Psi^{v}(\tau_{1};c_{2})=\Psi_{0}(\tau_{1})+\left[1-km'(\tau_{1})\right]\frac{\ell_{2}}{f(\tau_{1})}>\Psi_{0}(\tau_{1})+\left[1-km'(\tau_{1})\right]\frac{\Lambda_{1}}{f(\tau_{1})}=\Psi^{v}(\tau_{1};c_{1})=0,
\]
so $\tau_{2}<\tau_{1}$. In all cases $\tau_{2}\le\tau_{1}$: $\tau^{v}$
is weakly decreasing in $c$.

For strictness, if the optimal mechanism at $c_{2}$ has a positive
measure of included types with binding IR, then $\Lambda_{2}>0$ and
$\tau_{2}<\tau_{0}$ (Claim \ref{clm:psi0}(iii)); if $\Lambda_{1}=0$
then $\tau_{1}=\tau_{0}>\tau_{2}$, and if $\Lambda_{1}>0$ the display
above gives $\tau_{2}<\tau_{1}$---the decrease is strict either
way. Conversely, if no such positive-measure set exists at $c_{2}$,
then $\tau_{2}=\tau_{0}$ (Claim \ref{clm:psi0}(iii)); $\Lambda_{1}>0$
would give $\tau_{2}<\tau_{1}<\tau_{0}$, a contradiction, so $\Lambda_{1}=0$
and $\tau_{2}=\tau_{1}=\tau_{0}$.
\end{proof}
\emph{Part 2: Proof for the effect of $c$ on the seller's profit}
\begin{proof}
Let $(\alpha_{2},u_{2})$ with cutoff $\tau_{2}$ be optimal at $c_{2}$.
IC and IR do not involve $c$, so the mechanism is feasible at $c_{1}$,
and 
\begin{align*}
\Pi^{v}(\tau^{v}(c_{1});c_{1}) & \ge\int_{\tau_{2}}^{1}\left[\theta-kM(\alpha_{2}(\theta),\theta)-c_{1}\alpha_{2}(\theta)-u_{2}(\theta)\right]f(\theta)d\theta\\
 & =\Pi^{v}(\tau_{2};c_{2})+(c_{2}-c_{1})\int_{\tau_{2}}^{1}\alpha_{2}(\theta)f(\theta)d\theta.
\end{align*}
By Corollary \ref{Corollary:V-Fixed-tau}, $\alpha_{2}=a^{R}(\cdot;c_{2})$
on $[\max\{\tau_{2},\gamma(c_{2})\},1]$; since $b^{R}<1$ (Lemma
\ref{lem:Relaxed-program}), $b^{R}<\gamma$ whenever $\eta<\gamma$
(Lemma \ref{Lemma:-eta-gamma}), and $a^{R}>0$ on $(b^{R},1]$, the
integral is strictly positive. Hence $\Pi^{v}(\tau^{v}(c_{2});c_{2})<\Pi^{v}(\tau^{v}(c_{1});c_{1})$. 
\end{proof}
\emph{Part 3: Proof for the effect of $k$ on $\tau^{v}$}
\begin{proof}
Fix $c$ and let $k_{1}<k_{2}$ lie in the maintained region; write
$\tau_{1}:=\tau^{v}(k_{1})$, $\eta_{i}:=\eta(k_{i})$, $\gamma_{i}:=\gamma(k_{i})$,
and $\Lambda_{1}:=\Lambda^{\tau_{1}}(\tau_{1};k_{1})$. By Claim \ref{clm:psi0},
$\tau_{1}\in(0,\eta_{1})$ and 
\begin{equation}
0=\Psi^{v}(\tau_{1};k_{1})=\Psi_{0}(\tau_{1};k_{1})+\left[1-k_{1}m'(\tau_{1})\right]\frac{\Lambda_{1}}{f(\tau_{1})},\label{eq:CS-Prop-5}
\end{equation}
and, for any $\tau$, 
\begin{equation}
\Psi_{0}(\tau;k_{2})-\Psi_{0}(\tau;k_{1})=(k_{2}-k_{1})\left[m'(\tau)J(\tau)-m(\tau)\right].\label{eq:CS-Prop-6}
\end{equation}
Recall $\gamma_{2}>\gamma_{1}$ and $\eta_{2}<\eta_{1}$ (Claim \ref{clm:irmonotone}
and the proof of Lemma \ref{Lemma:-eta-gamma}). By Corollary \ref{Corollary:V-Fixed-tau},
exactly one of three cases applies to $\mathcal{P}^{v}(\tau_{1})$
at $k_{2}$.

\emph{Case A: the Lemma }\ref{lem:p-tau-case-1}\emph{ case applies
at $(\tau_{1},k_{2})$.} Then $\Lambda_{1}=0$. Indeed, suppose $\Lambda_{1}>0$,
so the Lemma \ref{lem:p-tau-case-3} case applies at $(\tau_{1},k_{1})$
and $\eta_{1}<\gamma_{1}$. If $\eta_{2}\ge\gamma_{2}$, then $\gamma_{2}-\eta_{2}\le0<\gamma_{1}-\eta_{1}$
contradicts $\gamma-\eta$ strictly increasing in $k$ (Lemma \ref{Lemma:-eta-gamma}).
So $\eta_{2}<\gamma_{2}$, and the Lemma \ref{lem:p-tau-case-1} case
requires $\tau_{1}\in[0,\underline{\tau}(k_{2})]\cup[\gamma_{2},1)$;
but $\tau_{1}<\eta_{1}<\gamma_{1}<\gamma_{2}$ rules out the second
interval, while Claim \ref{clm:irmonotone} gives $u^{R}(\gamma_{2};\tau_{1};k_{2})<0$
and hence $\tau_{1}>\underline{\tau}(k_{2})$, ruling out the first.
Thus $\Lambda_{1}=0$ and, by (\ref{eq:CS-Prop-5}), $\Psi_{0}(\tau_{1};k_{1})=0$,
i.e., $\tau_{1}-k_{1}m(\tau_{1})=\left[1-k_{1}m'(\tau_{1})\right]J(\tau_{1})$.
Claim \ref{clm:rent} with $X=J(\tau_{1})$ and (\ref{eq:CS-Prop-6})
give $\Psi_{0}(\tau_{1};k_{2})>0$. Since $a^{R}(\tau_{1};k_{2})$
maximizes $\psi^{v}(\cdot,\tau_{1};k_{2})$ on $Y$ (Lemma \ref{lem:Relaxed-program}),
\[
\Psi^{v}(\tau_{1};k_{2})=\psi^{v}(a^{R}(\tau_{1};k_{2}),\tau_{1};k_{2})\ge\psi^{v}(0,\tau_{1};k_{2})=\Psi_{0}(\tau_{1};k_{2})>0.
\]

\emph{Case B: the Lemma }\ref{lem:p-tau-case-2}\emph{ case applies
at $(\tau_{1},k_{2})$, i.e., $\eta_{2}<\gamma_{2}$ and $\tau_{1}\in[\eta_{2},\gamma_{2})$.}
By the Case 2(ii) computation in the proof of Lemma \ref{Lemma:Pi-v(v)-quasiconcave},
$\Psi^{v}(\tau_{1};k_{2})=(1-c)\tau_{1}+c\eta_{2}-k_{2}m(\eta_{2})$.
Since $k_{2}m(\eta_{2})<\eta_{2}k_{2}m'(\eta_{2})=\eta_{2}$ and $c<1$,
\[
\Psi^{v}(\tau_{1};k_{2})>(1-c)\tau_{1}+c\eta_{2}-\eta_{2}=(1-c)(\tau_{1}-\eta_{2})\ge0.
\]

\emph{Case C: the Lemma }\ref{lem:p-tau-case-3}\emph{ case applies
at $(\tau_{1},k_{2})$,} so $\tau_{1}<\eta_{2}$ and 
\begin{equation}
\Psi^{v}(\tau_{1};k_{2})=\Psi_{0}(\tau_{1};k_{2})+\left[1-k_{2}m'(\tau_{1})\right]\frac{\ell_{2}}{f(\tau_{1})},\qquad\ell_{2}:=L^{0}(\phi^{\tau_{1}}(k_{2});k_{2})>0.\label{eq:CS-Prop-7}
\end{equation}
If $\Lambda_{1}=0$, then $\Psi_{0}(\tau_{1};k_{1})=0$, Claim \ref{clm:rent}
and (\ref{eq:CS-Prop-6}) give $\Psi_{0}(\tau_{1};k_{2})>0$, and
both terms in (\ref{eq:CS-Prop-7}) are positive ($\tau_{1}<\eta_{2}$),
so $\Psi^{v}(\tau_{1};k_{2})>0$. If $\Lambda_{1}>0$, the Lemma \ref{lem:p-tau-case-3}
case applies at $(\tau_{1},k_{1})$ by Claim \ref{clm:psi0}(ii),
so Claim \ref{clm:multiplier} ($k$-direction) gives $\ell_{2}>\Lambda_{1}$,
and by (\ref{eq:CS-Prop-6}) and (\ref{eq:CS-Prop-5}), 
\[
\Psi^{v}(\tau_{1};k_{2})>\Psi_{0}(\tau_{1};k_{2})+\left[1-k_{2}m'(\tau_{1})\right]\frac{\Lambda_{1}}{f(\tau_{1})}=\Psi^{v}(\tau_{1};k_{1})+(k_{2}-k_{1})\,m'(\tau_{1})\left[\tilde{X}-\frac{m(\tau_{1})}{m'(\tau_{1})}\right],
\]
where $\tilde{X}:=J(\tau_{1})-\Lambda_{1}/f(\tau_{1})$. By (\ref{eq:CS-Prop-5}),
$\tau_{1}-k_{1}m(\tau_{1})=\left[1-k_{1}m'(\tau_{1})\right]\tilde{X}$,
so Claim \ref{clm:rent} gives $m'(\tau_{1})\tilde{X}-m(\tau_{1})>0$;
hence $\Psi^{v}(\tau_{1};k_{2})>0$.

In every case $\Psi^{v}(\tau_{1};k_{2})>0$. Since $\Psi^{v}(\cdot;k_{2})$
is continuous and strictly increasing with unique zero $\tau^{v}(k_{2})$
(Lemma \ref{Lemma:Pi-v(v)-quasiconcave}), $\tau^{v}(k_{2})<\tau^{v}(k_{1})$.
\end{proof}
\emph{Part 4: Proof for the effect of $k$ on the seller's profit}
\begin{proof}
Let $c=0$ and $0<k_{1}<k_{2}$. Let $(\alpha_{1},u_{1})$ with cutoff
$\tau_{1}:=\tau^{v}(k_{1})$ be optimal at $k_{1}$, with $\delta_{1}(\theta):=\theta-\alpha_{1}(\theta)\in[0,\theta]$
and $u_{1}(\theta)=\int_{\tau_{1}}^{\theta}\left[1-k_{1}m'(\delta_{1}(s))\right]ds$.
Define $T(z):=(m')^{-1}((k_{1}/k_{2})m'(z))$ for $z>0$ and $T(0):=0$;
$T$ is continuous and strictly increasing with $T(z)<z$ for $z>0$.
Construct a mechanism at $k_{2}$ with inclusion set $[\tau_{1},1]$,
mismatch $\delta_{2}:=T\circ\delta_{1}$, allocation $\alpha_{2}(\theta):=\theta-\delta_{2}(\theta)\in[\alpha_{1}(\theta),\theta]$,
utility $u_{2}(\theta):=\int_{\tau_{1}}^{\theta}\left[1-k_{2}m'(\delta_{2}(s))\right]ds$,
and transfer $t_{2}(\theta):=\theta-k_{2}m(\delta_{2}(\theta))-u_{2}(\theta)$.

First, $k_{2}m'(\delta_{2})=k_{1}m'(\delta_{1})$ pointwise by construction,
so $u_{2}=u_{1}\ge0$: IR holds and $u_{2}(\tau_{1})=0$.

Second, $\alpha_{2}$ is nondecreasing. For $z>0$, $T'(z)=k_{1}m''(z)/\left[k_{2}m''(T(z))\right]\le1$,
because $T(z)<z$ and log-concavity of $m'$ give 
\[
k_{2}m''(T(z))=k_{1}m'(z)\cdot\frac{m''(T(z))}{m'(T(z))}\ge k_{1}m'(z)\cdot\frac{m''(z)}{m'(z)}=k_{1}m''(z).
\]
Hence $z-T(z)$ is nondecreasing. For $\theta<\theta'$: if $\delta_{1}(\theta')\ge\delta_{1}(\theta)$,
then $\delta_{2}(\theta')-\delta_{2}(\theta)\le\delta_{1}(\theta')-\delta_{1}(\theta)\le\theta'-\theta$,
the last step because $\alpha_{1}$ is nondecreasing; if $\delta_{1}(\theta')<\delta_{1}(\theta)$,
then $\delta_{2}(\theta')-\delta_{2}(\theta)<0\le\theta'-\theta$.
Either way $\alpha_{2}(\theta')\ge\alpha_{2}(\theta)$.

Third, since $\delta_{2}\ge0$, $u_{2}$ satisfies the envelope condition
(\ref{eq:u'-v}); with $\alpha_{2}$ nondecreasing and $u_{2}(\tau_{1})=0$,
IC holds.

Fourth, since $c=0$, the seller's payoff $\Pi_{2}$ from the constructed
mechanism satisfies 
\[
\Pi_{2}-\Pi^{v}(\tau_{1};k_{1})=\int_{\tau_{1}}^{1}\left[k_{1}m(\delta_{1}(\theta))-k_{2}m(\delta_{2}(\theta))\right]f(\theta)d\theta.
\]
Fix $\theta$ with $\delta_{1}(\theta)>0$, set $x:=k_{1}m'(\delta_{1}(\theta))>0$,
and for $k\in[k_{1},k_{2}]$ let $\delta(k):=(m')^{-1}(x/k)$ and
$\phi(k):=k\,m(\delta(k))$, so $\phi(k_{1})=k_{1}m(\delta_{1}(\theta))$
and $\phi(k_{2})=k_{2}m(\delta_{2}(\theta))$. Then 
\[
\phi'(k)=m(\delta(k))-\frac{m'(\delta(k))^{2}}{m''(\delta(k))}<0,
\]
because strict log-concavity of $m'$ implies $m(z)m''(z)<m'(z)^{2}$
for $z>0$: $m''/m'$ strictly decreasing gives $m''(z)m'(s)<m''(s)m'(z)$
for $0<s<z$, and integrating over $s\in(0,z)$ yields $m''(z)m(z)<m'(z)\left[m'(z)-m'(0)\right]=m'(z)^{2}$.
Hence $k_{2}m(\delta_{2}(\theta))<k_{1}m(\delta_{1}(\theta))$ whenever
$\delta_{1}(\theta)>0$, with equality when $\delta_{1}(\theta)=0$.
By Corollary \ref{Corollary:Basic}, $\delta_{1}(\theta)=\theta\ge\tau_{1}>0$
on a positive-measure set, so $\Pi_{2}>\Pi^{v}(\tau_{1};k_{1})$.
Since the constructed mechanism is feasible at $k_{2}$, $\Pi^{v}(\tau^{v}(k_{2});k_{2})\ge\Pi_{2}>\Pi^{v}(\tau^{v}(k_{1});k_{1})$.
As $k_{1}<k_{2}$ were arbitrary, profit is strictly increasing in
$k$ when $c=0$.
\end{proof}
\newpage

\section{\protect\label{OA:Supporting-Anaylsis}Supporting Analysis}

\subsection{\protect\label{OA:Lemma-H-Cutoff}For Lemma \ref{lem:h-IC}}

\emph{Necessary conditions for IC.}
\begin{proof}
The envelope condition (\ref{eq:u'-horizontal})\emph{ }follows from
the envelope theorem.

\emph{Monotonicity.} Fix $\theta_{1}<\theta_{2}$ in $[0,\tau]$.
The two IC constraints
\[
u(\theta_{1})\ge v_{0}-kM(\alpha(\theta_{2}),\theta_{1})-t(\theta_{2}),\qquad u(\theta_{2})\ge v_{0}-kM(\alpha(\theta_{1}),\theta_{2})-t(\theta_{1})
\]
add, after substituting $t(\theta_{i})=v_{0}-kM(\alpha(\theta_{i}),\theta_{i})-u(\theta_{i})$,
to
\[
M(\alpha(\theta_{2}),\theta_{2})-M(\alpha(\theta_{2}),\theta_{1})\le M(\alpha(\theta_{1}),\theta_{2})-M(\alpha(\theta_{1}),\theta_{1}),
\]
i.e. $\int_{\theta_{1}}^{\theta_{2}}M_{2}(\alpha(\theta_{2}),\theta)\,d\theta\le\int_{\theta_{1}}^{\theta_{2}}M_{2}(\alpha(\theta_{1}),\theta)\,d\theta$.
If $\alpha(\theta_{2})<\alpha(\theta_{1})$, then $M_{2}(\alpha(\theta_{2}),\theta)>M_{2}(\alpha(\theta_{1}),\theta)$
for every $\theta$, reversing the inequality; hence $\alpha(\theta_{1})\le\alpha(\theta_{2})$,
so $\alpha$ is nondecreasing.

$u\left(\tau\right)=0$\emph{ if $\tau<1$}. Since $u$ is continuous,
if $u\left(\tau\right)>0$ and $\tau<1$, then excluded types just
above $\tau$ would obtain positive utility from $\tau$'s contract,
which violates their IC to be excluded.

\bigskip
\end{proof}
\emph{Sufficient conditions for IC.}
\begin{proof}
\emph{IC within $[0,\tau]$.} For $\theta,\hat{\theta}\in[0,\tau]$,
the payoff to $\theta$ from reporting $\hat{\theta}$ is
\[
U(\theta,\hat{\theta})=v_{0}-kM(\alpha(\hat{\theta}),\theta)-t(\hat{\theta})=u(\hat{\theta})+k\bigl[M(\alpha(\hat{\theta}),\hat{\theta})-M(\alpha(\hat{\theta}),\theta)\bigr].
\]
Using $u(\theta)-u(\hat{\theta})=-\int_{\hat{\theta}}^{\theta}kM_{2}(\alpha(s),s)\,ds$
and $M(\alpha(\hat{\theta}),\hat{\theta})-M(\alpha(\hat{\theta}),\theta)=-\int_{\hat{\theta}}^{\theta}M_{2}(\alpha(\hat{\theta}),s)\,ds$,
\[
u(\theta)-U(\theta,\hat{\theta})=k\int_{\hat{\theta}}^{\theta}\bigl[M_{2}(\alpha(\hat{\theta}),s)-M_{2}(\alpha(s),s)\bigr]\,ds.
\]
If $\theta>\hat{\theta}$, then $\alpha(s)\ge\alpha(\hat{\theta})$
on $[\hat{\theta},\theta]$, so the integrand is nonnegative; if $\theta<\hat{\theta}$,
then $\alpha(s)\le\alpha(\hat{\theta})$ on $[\theta,\hat{\theta}]$,
the integrand is nonpositive, and reversing the limits restores a
nonnegative value. Either way $u(\theta)\ge U(\theta,\hat{\theta})$,
so no included type gains by misreporting. By IR, $u(\theta)\ge0$,
so no included type prefers an excluded type's contract.

\emph{IC of excluded types.} Suppose $\tau<1$. For $\hat{\theta}\in[0,\tau]$,
let $\Delta(\hat{\theta},\theta):=v_{0}-kM(\alpha(\hat{\theta}),\theta)-t(\hat{\theta})$
be the payoff that type $\theta$ obtains from the included contract
$\hat{\theta}$; since $M(\alpha(\hat{\theta}),\cdot)$ is convex,
$\Delta(\hat{\theta},\cdot)$ is concave. I claim that $\partial_{\theta}\Delta(\hat{\theta},\tau)\le0$
for every $\hat{\theta}\in[0,\tau]$. 

At $\hat{\theta}=\tau$: $\partial_{\theta}\Delta(\tau,\tau)=-kM_{2}(\alpha(\tau),\tau)=u'(\tau^{-})\le0$,
since IR and $u(\tau)=0$ give $u(\theta)\ge0=u(\tau)$ for $\theta<\tau$,
hence a nonpositive left derivative at $\tau$. 

At $\hat{\theta}<\tau$: suppose toward a contradiction that $\partial_{\theta}\Delta(\hat{\theta},\tau)>0$.
By concavity, $\Delta(\hat{\theta},\cdot)$ is then strictly increasing
on $[0,\tau]$, so $\Delta(\hat{\theta},\hat{\theta})<\Delta(\hat{\theta},\tau)$.
But $\Delta(\hat{\theta},\hat{\theta})=u(\hat{\theta})$, and by the
IC within $[0,\tau]$ just established (applied at type $\tau$),
$\Delta(\hat{\theta},\tau)=U(\tau,\hat{\theta})\le u(\tau)=0$. Hence
$u(\hat{\theta})<0$, contradicting IR. Therefore $\partial_{\theta}\Delta(\hat{\theta},\tau)\le0$.

Concavity then makes $\Delta(\hat{\theta},\cdot)$ nonincreasing on
$[\tau,1]$, so for every excluded $\theta>\tau$ and every $\hat{\theta}\in[0,\tau]$,
$\Delta(\hat{\theta},\theta)\le\Delta(\hat{\theta},\tau)\le u(\tau)=0.$
Thus no excluded type gains from an included contract, and exclusion
(utility $0$) is optimal for them. This establishes IC.
\end{proof}

\subsection{\protect\label{OA:Lemma-H-Step-1}For Lemma \ref{lem:h-step-1}}

\emph{$\alpha^{h}$ is nondecreasing.}
\begin{proof}
Differentiate $\partial\psi^{h}/\partial y$ in (\ref{eq:=00005Cpsi^h-derivative})
with respect to $\theta$:
\[
\frac{\partial^{2}\psi^{h}}{\partial y\,\partial\theta}=km''(\theta-y)\left[1+(F/f)'(\theta)\right]+km'''(\theta-y)\frac{F(\theta)}{f(\theta)}.
\]
Strict log-concavity of $F$ gives $(F/f)'>0$; with $m''>0$ and
$m'''\geq0$, the cross-partial is strictly positive for $y<\theta$,
so $\psi^{h}$ has strictly increasing differences in $(y,\theta)$.
Take $\theta_{1}<\theta_{2}$ with $y_{i}:=\alpha^{h}(\theta_{i})$,
and suppose $y_{1}>y_{2}$. Then $y_{2}<y_{1}\leq\theta_{1}<\theta_{2}$,
so both allocations are feasible at both types; optimality yields
\[
\psi^{h}(y_{1},\theta_{1})\geq\psi^{h}(y_{2},\theta_{1})\quad\text{and}\quad\psi^{h}(y_{2},\theta_{2})\geq\psi^{h}(y_{1},\theta_{2}),
\]
whose sum gives $g(y_{2})\geq g(y_{1})$ for $g(y):=\psi^{h}(y,\theta_{2})-\psi^{h}(y,\theta_{1})=\int_{\theta_{1}}^{\theta_{2}}\partial_{\theta}\psi^{h}(y,\theta)\,d\theta$.
But $g'(y)=\int_{\theta_{1}}^{\theta_{2}}\partial_{y\theta}^{2}\psi^{h}(y,\theta)\,d\theta>0$
on $[y_{2},y_{1}]$ (where $y\leq\theta_{1}\leq\theta$ throughout),
so $g(y_{1})>g(y_{2})$, a contradiction. Hence $\alpha^{h}(\theta_{1})\leq\alpha^{h}(\theta_{2})$.
\end{proof}

\subsection{\protect\label{OA:Lemma-H-Step-2}For Lemma \ref{lem:h-step-2}}

\emph{$\Psi^{h}$ is strictly decreasing on $(0,1]$.}
\begin{proof}
$\Psi^{h}$ is continuous on $[0,1]$ because $\alpha^{h}$ is, so
it suffices to separately show that $\Psi^{h}$ is strictly decreasing
in each of the three regions. 

In the pooling region $(0,\beta^{h}]$, $\alpha^{h}(\tau)=0$, so
$\Psi^{h}(\tau)=\psi^{h}(0,\tau)=v_{0}-km(\tau)-km'(\tau)F(\tau)/f(\tau)$
and
\[
(\Psi^{h})'(\tau)=-km'(\tau)-km''(\tau)\frac{F(\tau)}{f(\tau)}-km'(\tau)\left(\frac{F}{f}\right)'(\tau)<0,
\]
where all three terms are strictly negative, using $m',m''>0$ on
$\mathbb{R}_{++}$ and $(F/f)'>0$ from strict log-concavity of $F$. 

In the interior region $(\beta^{h},\theta^{*})$ (with $\theta^{*}$
replaced by $1$ when it does not exist), $\alpha^{h}(\tau)\in(0,\tau)$,
so neither constraint of the inner program binds and the envelope
theorem gives
\begin{align*}
(\Psi^{h})'(\tau)=\frac{\partial\psi^{h}(\alpha^{h}(\tau),\tau)}{\partial\theta} & =-km'(d^{h}(\tau))-km''(d^{h}(\tau))\frac{F(\tau)}{f(\tau)}-km'(d^{h}(\tau))\left(\frac{F}{f}\right)'(\tau).\\
 & =-c-km'(d^{h}(\tau))\left(\frac{F}{f}\right)'(\tau)<0.
\end{align*}
where the second line follows from (\ref{eq:delta-h}).

In the ideal-product region $[\theta^{*},1]$, which is nonempty if
and only if $m''(0)/f(1)\geq c/k$, the constraint $y\leq\tau$ binds.
Since $\alpha^{h}(\tau)=\tau$, using $m(0)=m'(0)=0$,
\[
\Psi^{h}(\tau)=\psi^{h}(\tau,\tau)=v_{0}-c\tau,\qquad(\Psi^{h})'(\tau)=-c<0.
\]

Therefore $\Psi^{h}$ is strictly decreasing on $(0,1]$ in all three
regions.
\end{proof}

\subsection{\protect\label{OA:Lemma-Vertical}For Lemma \ref{lem:v-CutoffMech}}

\emph{Monotonicity of $\alpha$ and the envelope condition }(\ref{eq:u'-v})\emph{
are jointly necessary and sufficient for IC for $\left[\tau,1\right]$.}
\begin{proof}
For $\theta_{1}<\theta_{2}$ in $[\tau,1]$, IC requires
\[
u(\theta_{1})\geq\theta_{1}-kM(\alpha(\theta_{2}),\theta_{1})-t(\theta_{2}),\qquad u(\theta_{2})\geq\theta_{2}-kM(\alpha(\theta_{1}),\theta_{2})-t(\theta_{1}).
\]
Adding and rearranging (the vertical terms cancel),
\[
kM(\alpha(\theta_{2}),\theta_{1})+kM(\alpha(\theta_{1}),\theta_{2})\geq kM(\alpha(\theta_{1}),\theta_{1})+kM(\alpha(\theta_{2}),\theta_{2}).
\]
Strict convexity of $m$ then forces $\alpha(\theta_{1})\leq\alpha(\theta_{2})$;
otherwise, the inequality is reversed. 

To show sufficiency, let the payoff to type $\theta\in[\tau,1]$ from
reporting $\hat{\theta}\in[\tau,1]$ be
\[
U(\theta,\hat{\theta})=\theta-kM(\alpha(\hat{\theta}),\theta)-t(\hat{\theta})=u(\hat{\theta})+(\theta-\hat{\theta})-k\left[M(\alpha(\hat{\theta}),\theta)-M(\alpha(\hat{\theta}),\hat{\theta})\right].
\]
Since $m'(0)=0$, $M(y,\cdot)$ is continuously differentiable with
$M_{2}(y,s)=m'(s-y)$ for $s\geq y$ and $M_{2}(y,s)=-m'(y-s)$ for
$s\leq y$, the two branches agreeing at $s=y$. Hence
\[
M(\alpha(\hat{\theta}),\theta)-M(\alpha(\hat{\theta}),\hat{\theta})=\int_{\hat{\theta}}^{\theta}M_{2}(\alpha(\hat{\theta}),s)\,ds,
\]
and by\emph{ }(\ref{eq:u'-v}), $u(\theta)-u(\hat{\theta})=\int_{\hat{\theta}}^{\theta}\left[1-kM_{2}(\alpha(s),s)\right]ds$.
Combining,
\[
u(\theta)-U(\theta,\hat{\theta})=\int_{\hat{\theta}}^{\theta}k\left[M_{2}(\alpha(\hat{\theta}),s)-M_{2}(\alpha(s),s)\right]ds.
\]
For each fixed $s$, $M_{2}(y,s)$ is nonincreasing in $y$: $m'(s-y)$
decreases in $y$ on $\{y\leq s\}$, $-m'(y-s)$ decreases in $y$
on $\{y\geq s\}$, and both equal zero at $y=s$. If $\hat{\theta}<\theta$,
then $\alpha(\hat{\theta})\leq\alpha(s)$ for $s\in[\hat{\theta},\theta]$,
so the integrand is nonnegative and the integral is nonnegative. If
$\hat{\theta}>\theta$, then $\alpha(\hat{\theta})\geq\alpha(s)$
for $s\in[\theta,\hat{\theta}]$, so the integrand is nonpositive
on $[\theta,\hat{\theta}]$ and the integral from $\hat{\theta}$
to $\theta$ is again nonnegative. In both cases $u(\theta)\geq U(\theta,\hat{\theta})$,
establishing IC within $[\tau,1]$.
\end{proof}
\bigskip

\emph{$u\left(\tau\right)=0$ is necessary and sufficient for IC for
$\left[0,\tau\right]$.}
\begin{proof}
This is vacuous if $\tau=0$; thus, assume $\tau>0$. 

IC for excluded type $\tau-\varepsilon$ gives $(\tau-\varepsilon)-kM(\alpha(\tau),\tau-\varepsilon)-t(\tau)\leq0$;
taking $\varepsilon\downarrow0$ yields $u(\tau)\leq0$, which combined
with IR forces $u(\tau)=0$.

Define $\Delta(\theta',\theta):=\theta-kM(\alpha(\theta'),\theta)-t(\theta')$,
the utility type $\theta$ obtains from mimicking $\theta'$. Since
$M(y,\cdot)$ is convex, $\Delta(\theta',\cdot)$ is concave. I claim
that $\Delta_{2}(\theta',\tau)\geq0$ for all $\theta'\in[\tau,1]$.

At $\theta'=\tau$: $\Delta_{2}(\tau,\tau)=u'(\tau^{+})\geq0$, since
$u(\tau)=0$ and $u\left(\theta\right)\geq0$ on $[\tau,1]$ .

At $\theta'>\tau$: Suppose toward a contradiction that $\Delta_{2}(\theta',\tau)<0$.
By concavity, $\Delta(\theta',\theta)<\Delta(\theta',\tau)$ for all
$\theta>\tau$; setting $\theta=\theta'$ gives 
\[
u(\theta')=\Delta(\theta',\theta')<\Delta(\theta',\tau)\leq u(\tau)=0,
\]
where the second inequality is IC for $\tau$. This contradicts IR
of $\theta'$. Hence $\Delta(\theta',\cdot)$ is nondecreasing at
$\tau$. By concavity, for any $\theta_{e}<\tau$, 
\[
\Delta(\theta',\theta_{e})\leq\Delta(\theta',\tau)\leq u(\tau)=0,
\]
where the second inequality is IC for $\tau$. This is exactly IC
for the excluded type $\theta_{e}<\tau$.
\end{proof}

\subsection{\protect\label{OA:Adjusted-hazardrate}For Lemma \ref{lem:p-tau-case-3}}

\emph{$\left(1-F-L^{0}\left(\phi\right)\right)/f$ is nonincreasing
on $\left[\tau,\phi\right]$. }
\begin{proof}
Differentiating,
\[
\frac{d}{d\theta}\left[\frac{1-F\left(\theta\right)-L^{0}\left(\phi\right)}{f\left(\theta\right)}\right]=-\frac{f\left(\theta\right)^{2}+\left[1-F\left(\theta\right)-L^{0}\left(\phi\right)\right]f'\left(\theta\right)}{f\left(\theta\right)^{2}}.
\]
The ratio is nonincreasing if and only if $f\left(\theta\right)^{2}+\left[1-F\left(\theta\right)-L^{0}\left(\phi\right)\right]f'\left(\theta\right)\geq0$.
$\left(1-F\right)/f$ is decreasing (Assumption \ref{Assumption:F}),
which is equivalent to $f^{2}+\left(1-F\right)f'>0$. Write
\[
f^{2}+\left(1-F-L^{0}\left(\phi\right)\right)f'=\underbrace{\left[f^{2}+\left(1-F\right)f'\right]}_{>0}-L^{0}\left(\phi\right)f'.
\]
If $f'\left(\theta\right)\leq0$, the subtracted term is negative,
so the right-hand side is positive. If $f'\left(\theta\right)>0$,
divide through by $f'\left(\theta\right)$: the inequality $f^{2}+\left(1-F-L^{0}\left(\phi\right)\right)f'\geq0$
becomes $L^{0}\left(\phi\right)\leq f^{2}/f'+\left(1-F\right)$. Since
$f^{2}/f'>0$ and $L^{0}\left(\phi\right)\leq1-F\left(\phi\right)\leq1-F\left(\theta\right)$
for $\theta\leq\phi$, this also holds.
\end{proof}

\subsection{\protect\label{OA:Quasiconcavity-Pi-v}For Lemma \ref{Lemma:Pi-v(v)-quasiconcave}}

\emph{$\Psi^{v}$ is continuous and strictly increasing.}
\begin{proof}
\emph{Case 1: $\eta\ge\gamma$.} In this case, $\Psi^{v}\left(\tau\right)=\psi^{v}\left(a^{R}\left(\tau\right),\tau\right)=:\Psi^{R}\left(\tau\right)$.
Let $J=\left(1-F\right)/f$. By the envelope theorem, $\Psi^{R'}(\tau)=[1-km'(\delta^{R}(\tau))][1-J'(\tau)]+km''(\delta^{R}(\tau))J(\tau).$
Since $1-km'(\delta^{R}(\tau))>0$ almost everywhere and $J'<0$,
$\Psi^{R}$ is strictly increasing.

\emph{Case 2: $\eta<\gamma$. }I split this case into three subcases
spanning Lemmas \ref{lem:p-tau-case-1}-\ref{lem:p-tau-case-3}: (i)
$\tau\notin(\underline{\tau},\gamma)$, (ii) $\tau\in[\eta,\gamma)$,
and (iii) $\tau\in\left(\underline{\tau},\eta\right)$. I show that
$\Psi^{v}$ is continuous and strictly increasing in each segment
and also continuous at the boundary, which then implies that $\Psi^{v}$
is continuous and strictly increasing on $\left[0,1\right]$.

\emph{Case 2(i):} $\tau\notin(\underline{\tau},\gamma)$\emph{. }By
Corollary \ref{Corollary:V-Fixed-tau}, $\alpha^{\tau}(\tau)=a^{R}\left(\tau\right)$.\emph{
}Following the argument for Case 1, $\Psi^{v}$ is strictly increasing.

\emph{Case 2(ii):} $\tau\in[\eta,\gamma)$\emph{. }By Corollary \ref{Corollary:V-Fixed-tau},\emph{
}$\alpha^{\tau}(\tau)=\tau-\eta$ and $kM_{2}(\tau-\eta,\tau)=km'(\eta)=1$,
so the information-rent term vanishes: $\Psi^{v}(\tau)=(1-c)\tau+c\eta-km(\eta)$,
which implies that $\Psi^{v'}(\tau)=1-c>0$. 

\emph{Case 2(iii):} $\tau\in(\underline{\tau},\eta)$\emph{. }By Corollary
\ref{Corollary:V-Fixed-tau},\emph{ }$\alpha^{\tau}(\tau)=a(\tau;\phi^{\tau})=0$
because $\beta^{\tau}>\tau$, and $\Lambda^{\tau}(\tau)=L^{0}(\phi^{\tau})$.
Let $\tilde{J}\left(\tau\right)=\left(1-F\left(\tau\right)-L^{0}\left(\phi^{\tau}\right)\right)/f\left(\tau\right)$.
Therefore, $\Psi^{v}(\tau)=\tau-km(\tau)-[1-km'(\tau)]\tilde{J}\left(\tau\right)$.
For fixed $\phi$, the expression $[1-F(\tau)-L^{0}(\phi)]/f(\tau)$
is decreasing in $\tau$. Differentiating $U(\phi^{\tau};\tau)=0$
gives $d\phi^{\tau}/d\tau=-U_{\tau}(\phi^{\tau};\tau)/U_{\phi}(\phi^{\tau};\tau)<0$
because $U_{\tau}<0$ and $U_{\phi}<0$ on $\left(\eta,\gamma\right)$.
Since $L^{0}$ is decreasing, this implies that $\tilde{J}'\left(\tau\right)<0$.
Hence, $\Psi^{v'}(\tau)=[1-km'(\tau)]+km''(\tau)\widetilde{J}(\tau)-[1-km'(\tau)]\widetilde{J}'(\tau)>0$
because $\tau<\eta$ implies $1-km'(\tau)>0$.

\emph{Continuity at $\underline{\tau},\eta,\gamma$: }At $\tau=\underline{\tau}$,
evaluating at case 2(iii), $\phi^{\underline{\tau}}=\gamma$ by $U(\gamma;\underline{\tau})=u^{R}(\gamma;\underline{\tau})=0$,
so $L^{0}(\phi^{\underline{\tau}})=0$, which recovers case 2(i).
At $\tau=\eta$, evaluating at case 2(ii), $a^{0}(\eta)=0=a(\eta;\eta)=a\left(\eta;\phi^{\eta}\right)$,
matching case 2(iii). At $\tau=\gamma$, evaluating at case 2(i),
$\alpha^{\tau}\left(\gamma\right)=a^{R}(\gamma)=\gamma-\eta=a^{0}(\gamma)$;
since $L^{0}(\gamma)=0$, it matches case 2(ii). 
\end{proof}
\newpage

\section{\protect\label{OA:Uniform-Distribution}Uniform Distribution and
Quadratic Mismatch Cost}

In this appendix section, I provide an example of the solution characterization
using the ``\emph{uniform-quadratic}'' parameterization, where 
\[
F(\theta)=\theta,\qquad f(\theta)=1,\qquad m(z)=\frac{z^{2}}{2}.
\]
Then $M(y,\theta)=(\theta-y)^{2}/2$, $M_{1}=-(\theta-y)$, $M_{2}=\theta-y$,
$M_{12}=-1$, and $(1-F(\theta))/f(\theta)=1-\theta$. Assumptions
\ref{Assumption:F}, \ref{Assumption:m-1}, and \ref{Assumption:FB-ProductLine}
reduce to $c<k$. For brevity, write
\[
r:=c/k\in[0,1),\qquad S:=k+c.
\]
The optimal mechanisms are reported in Subsections \ref{OA:Uniform-Horizontal}
and \ref{OA:Uniform-Vertical}. Their derivations are provided in
Subsections \ref{OA:Uniform-Derivation-Horizontal} and \ref{OA:Uniform-Derivation-Vertical}.

\subsection{\protect\label{OA:Uniform-Horizontal}Optimal Mechanism for Section
\ref{Section:PureHorizontal}}

Assume first that $c>0$. Otherwise, Lemma \ref{Lemma:Horizontal-c=00003D0}
applies.

\emph{Coverage cutoff.} The optimal cutoff is 

\[
\tau^{h}=\begin{cases}
1, & v_{0}\ge c,\\[4pt]
\dfrac{v_{0}}{c}, & \dfrac{c^{2}}{k}\le v_{0}<c,\\[8pt]
\dfrac{2c-\sqrt{3c^{2}-2kv_{0}}}{k}, & \dfrac{3c^{2}}{8k}\le v_{0}<\dfrac{c^{2}}{k},\\[10pt]
\sqrt{\dfrac{2v_{0}}{3k}}, & 0<v_{0}<\dfrac{3c^{2}}{8k}.
\end{cases}
\]
\emph{Allocation and mismatch.} The pointwise rule is unchanged; only
the inclusion set $[0,\tau^{h}]$ shrinks. With $\delta^{h}(\theta):=\theta-\alpha^{h}(\theta)$,
for $\theta\in[0,\tau^{h}]$ the allocation is
\[
\alpha^{h}(\theta)=\max\bigl\{0,\,\min\{\theta,\,2\theta-r\}\bigr\}=\begin{cases}
0, & \theta\in[0,\,r/2],\\
2\theta-r, & \theta\in(r/2,\,r),\\
\theta, & \theta\in[r,\,1],
\end{cases}
\]
and the induced mismatch is
\[
\delta^{h}(\theta)=\min\bigl\{\theta,\,(r-\theta)_{+}\bigr\}=\begin{cases}
\theta, & \theta\in[0,\,r/2],\\
r-\theta, & \theta\in(r/2,\,r),\\
0, & \theta\in[r,\,1],
\end{cases}
\]
where $(x)_{+}:=\max\{x,0\}$.

\emph{Utility. }A compact way to report the utility schedule is to
define

\begin{equation}
H(x):=\int_{0}^{x}\delta^{h}(s)\,ds=\begin{cases}
\dfrac{x^{2}}{2}, & x\in[0,r/2],\\[6pt]
rx-\dfrac{x^{2}}{2}-\dfrac{r^{2}}{4}, & x\in(r/2,r),\\[8pt]
\dfrac{r^{2}}{4}, & x\in[r,1].
\end{cases}\label{eq:H-UniformQuadratic}
\end{equation}
Then

\[
u^{h}(\theta)=k\left[H(\tau^{h})-H(\theta)\right],\qquad\theta\in[0,\tau^{h}].
\]

\emph{Transfer.} The transfer is

\[
t^{h}(\theta)=v_{0}-\frac{k}{2}\left(\delta^{h}(\theta)\right)^{2}-u^{h}(\theta).
\]

\emph{Profit.} The seller's profit is $\Pi^{h}(\tau^{h})=\int_{0}^{\tau^{h}}\Psi^{h}(\theta)d\theta$,
which is 
\[
\Pi^{h}=\begin{cases}
v_{0}-\dfrac{c}{2}+\dfrac{c^{3}}{12k^{2}}, & v_{0}\geq c,\\[3mm]
\dfrac{v_{0}^{2}}{2c}+\dfrac{c^{3}}{12k^{2}}, & c^{2}/k\leq v_{0}<c,\\[3mm]
\dfrac{v_{0}r}{2}-\dfrac{kr^{3}}{16}+a\Bigl(\tau^{h}-\dfrac{r}{2}\Bigr)-c\Bigl((\tau^{h})^{2}-\dfrac{r^{2}}{4}\Bigr)+\dfrac{k}{6}\Bigl((\tau^{h})^{3}-\dfrac{r^{3}}{8}\Bigr), & 3c^{2}/(8k)\leq v_{0}<c^{2}/k,\\[3mm]
\dfrac{2v_{0}}{3}\,\tau^{h}=\dfrac{2v_{0}}{3}\sqrt{\dfrac{2v_{0}}{3k}}, & 0<v_{0}<3c^{2}/(8k),
\end{cases}
\]
where $a:=v_{0}+c^{2}/(2k)$ and $\tau^{h}$ is as above. The first
two values agree at $v_{0}=c$, and the ideal-region value $v_{0}^{2}/(2c)+c^{3}/(12k^{2})$
specializes to the full-coverage profit when coverage becomes complete.

\subsection{\protect\label{OA:Uniform-Vertical}Optimal Mechanism for Section
\ref{Section:Vertical}}

The following objects are determined by the parameters alone and follow
from straightforward algebra:
\begin{align*}
\eta & :=1/k,\qquad\gamma:=(S-1)/k=1-(1-c)/k,\qquad b^{R}:=(1+r)/2=S/(2k),\\
L^{0}(\theta) & :=\gamma-\theta\quad\text{on }[0,\gamma],\\
a^{R}(\theta) & :=\max\{0,\;2\theta-1-r\},\qquad d^{R}(\theta):=1+r-\theta\text{ for }\theta\ge b^{R}.
\end{align*}
Under this parameterization, $b^{R}$ can also be expressed as $b^{R}=(\eta+\gamma)/2$.

The condition $\eta\ge\gamma$ from Lemma \ref{Lemma:-eta-gamma}
is equivalent to $S\le2$; equivalently, $\eta<\gamma$ holds for
$S>2$.

For $S>2$ (equivalently, $\eta<\gamma$), define the IR-binding threshold
\[
\underline{\tau}:=\frac{2-\sqrt{2}(S-2)}{2k}=\eta-\frac{\sqrt{2}(S-2)}{2k}.
\]
Define the case-boundary curve
\[
c^{*}(k):=2-k+\frac{\sqrt{2}}{3}\bigl[\sqrt{k^{2}-2k+4}-(k-1)\bigr],
\]
and the constants
\[
q(c):=\frac{1}{(1-c)+\sqrt{(1-c)^{2}+3+2\sqrt{2}}},\qquad\tau_{I}^{v}(k):=\frac{(2+k)-\sqrt{k^{2}-2k+4}}{3k}.
\]
Split the parameter region $\{(c,k):0\le c<k,\,k>0\}$ into two cases
\begin{align*}
\text{Case (I):} & \quad S\le2,\;\;\text{or}\;\;S>2\text{ and }c\le c^{*}(k);\\
\text{Case (II):} & \quad S>2\text{ and }c>c^{*}(k).
\end{align*}
Case (I) corresponds to when the relaxed solution is feasible for
\ref{eq:Program-v} at the optimal cutoff (so Lemma \ref{lem:p-tau-case-1}
applies), whereas Case (II) is the region in which the optimal cutoff
lies in the IR-correction region. Case (II) requires $k>1$.

The optimal mechanism has the following closed form.

\emph{Case (I). }
\[
\tau^{v}=\tau_{I}^{v}(k)
\]
\[
\alpha^{\tau}(\theta)=a^{R}(\theta)\ \ \ \text{ on }\left[\tau^{v},1\right]
\]
\[
u^{\tau}(\theta)=\begin{cases}
\dfrac{k}{2}\bigl[(\eta-\tau^{v})^{2}-(\eta-\theta)^{2}\bigr], & \theta\in[\tau^{v},b^{R}],\\[4pt]
\dfrac{k}{2}\bigl[(\eta-\tau^{v})^{2}+(\theta-\gamma)^{2}-\tfrac{1}{2}(\eta-\gamma)^{2}\bigr], & \theta\in[b^{R},1].
\end{cases}
\]

\emph{Case (II).}
\[
\tau^{v}=(1-q(c))/k
\]
\[
\alpha^{\tau}(\theta)=\begin{cases}
0, & \theta\in[\tau^{v},\beta^{\tau}],\\
2\theta-\eta-\phi^{\tau}, & \theta\in[\beta^{\tau},\phi^{\tau}],\\
\theta-\eta, & \theta\in[\phi^{\tau},\gamma],\\
2\theta-1-r, & \theta\in[\gamma,1],
\end{cases}\qquad u^{\tau}(\theta)=\begin{cases}
\dfrac{k}{2}\bigl[(\eta-\tau^{v})^{2}-(\eta-\theta)^{2}\bigr], & \theta\in[\tau^{v},\beta^{\tau}],\\[4pt]
\dfrac{k}{2}(\theta-\phi^{\tau})^{2}, & \theta\in[\beta^{\tau},\phi^{\tau}],\\[4pt]
0, & \theta\in[\phi^{\tau},\gamma],\\[4pt]
\dfrac{k}{2}(\theta-\gamma)^{2}, & \theta\in[\gamma,1].
\end{cases}
\]
where 
\[
\phi^{\tau}=\eta+\sqrt{2}(\eta-\tau^{v})\ \ \ \ ;\ \ \ \ \beta^{\tau}=(\eta+\phi^{\tau})/2
\]

In both cases the transfer is $t(\theta)=\theta-kM(\alpha^{\tau}(\theta),\theta)-u^{\tau}(\theta)$.

\subsection{\protect\label{OA:Uniform-Profit}Profit and $c$}

Proposition \ref{Proposition:Vertical-CS} provides comparative statics
on the seller's profit with respect to $k$ when $c=0$. The following
proposition provides the same comparative static allowing $c$ to
take any value under the uniform-quadratic parameterization.
\begin{prop}
Suppose that $v\left(\theta\right)=\theta$. Under the uniform-quadratic
parameterization:
\begin{itemize}
\item If $c<1/2$, there exists $\underline{k}\left(c\right)>c$ such that
$\Pi^{v}\left(\tau^{v}\right)$ is strictly increasing in $k$ for
all $k>\underline{k}\left(c\right)$. 
\item If $c>1/2$, there exists $\underline{k}\left(c\right)>c$ such that
$\Pi^{v}\left(\tau^{v}\right)$ is strictly decreasing in $k$ for
all $k>\underline{k}\left(c\right)$.
\end{itemize}
\end{prop}
\begin{proof}
For every fixed $c$, sufficiently large $k$ places the optimum in
Case II. Substituting the Case II allocation and utility schedule
into the seller's objective gives $\Pi_{II}^{v}(c,k)=\{3(1-c)k^{2}+3(2c-1)k+Q(c)\}/(6k^{2})$,
where $Q(c):=-c^{3}+3c^{2}-6c+1-(1-c)q(c)^{2}+2q(c)$. Hence the exact
derivative is
\[
\frac{\partial\Pi_{II}^{v}(c,k)}{\partial k}=\frac{3(1-2c)k-2Q(c)}{6k^{3}}=\frac{1-2c}{2k^{2}}-\frac{Q(c)}{3k^{3}}.
\]
Since $Q(c)$ is finite for each fixed $c$, the sign of this derivative
for sufficiently large $k$ is governed by $1-2c$. If $c<1/2$, there
exists $\underline{k}(c)>c$ such that $\partial\Pi^{v}(\tau^{v})/\partial k>0$
for all $k>\underline{k}(c)$. If $c>1/2$, there exists $\underline{k}(c)>c$
such that $\partial\Pi^{v}(\tau^{v})/\partial k<0$ for all $k>\underline{k}(c)$.
\end{proof}

\subsection{\protect\label{OA:Uniform-Derivation-Horizontal}Derivation of Optimal
Mechanism for Section \ref{Section:PureHorizontal}}

For $y\in[0,\theta]$, define the downward mismatch $d:=\theta-y$.
Under the uniform-{}-quadratic parameterization, $M(y,\theta)=d^{2}/2$
and $M_{2}(y,\theta)=d$. The virtual surplus in Section \ref{Section:PureHorizontal}
therefore becomes

\[
\begin{aligned}\psi^{h}(y,\theta) & =v_{0}-\frac{k}{2}(\theta-y)^{2}-cy-k\theta(\theta-y)\\
 & =v_{0}-c\theta+cd-k\theta d-\frac{k}{2}d^{2}.
\end{aligned}
\]
Since $d\in[0,\theta]$, pointwise maximization is equivalent to maximizing
the last expression over $d\in[0,\theta]$. Its derivative with respect
to $d$ is $c-k\theta-kd=k(r-\theta-d)$. Hence $\delta^{h}(\theta)=\min\{\theta,(r-\theta)_{+}\},$
which yields the allocation rule reported in Subsection \ref{OA:Uniform-Horizontal}.
In particular, the pooling threshold is $\beta^{h}=r/2$, and the
allocation reaches the ideal product at $\theta^{*}=r$.

To characterize the cutoff, evaluate the marginal virtual surplus
at the pointwise allocation.

\[
\Psi^{h}(\theta):=\psi^{h}\bigl(\alpha^{h}(\theta),\theta\bigr)=\begin{cases}
v_{0}-\dfrac{3k}{2}\theta^{2}, & \theta\in[0,r/2],\\[5pt]
v_{0}+\dfrac{c^{2}}{2k}-2c\theta+\dfrac{k}{2}\theta^{2}, & \theta\in(r/2,r),\\[5pt]
v_{0}-c\theta, & \theta\in[r,1].
\end{cases}
\]
The first expression is decreasing on $[0,r/2]$, the second is decreasing
on $(r/2,r)$, and the third is decreasing on $[r,1]$.  Assumption
\ref{Assumption:FB-ProductLine} implies $c<k$. Consequently,

\[
\Psi^{h}(r/2)=v_{0}-\frac{3c^{2}}{8k},\qquad\Psi^{h}(r)=v_{0}-\frac{c^{2}}{k},\qquad\Psi^{h}(1)=v_{0}-c.
\]
Therefore, if $v_{0}\ge c$, then $\Psi^{h}(1)\ge0$ and $\tau^{h}=1$.
If $v_{0}<c$, the cutoff is the unique zero of $\Psi^{h}$, which
yields

\[
\tau^{h}=\begin{cases}
\dfrac{v_{0}}{c}, & \dfrac{c^{2}}{k}\le v_{0}<c,\\[8pt]
\dfrac{2c-\sqrt{3c^{2}-2kv_{0}}}{k}, & \dfrac{3c^{2}}{8k}\le v_{0}<\dfrac{c^{2}}{k},\\[10pt]
\sqrt{\dfrac{2v_{0}}{3k}}, & 0<v_{0}<\dfrac{3c^{2}}{8k}.
\end{cases}
\]

The envelope condition gives $u^{h'}(\theta)=-k\delta^{h}(\theta).$
Using $u^{h}(\tau^{h})=0$, 
\[
u^{h}(\theta)=\int_{\theta}^{\tau^{h}}k\delta^{h}(s)\,ds=k\left[H(\tau^{h})-H(\theta)\right],
\]
where $H$ is the antiderivative reported in (\ref{eq:H-UniformQuadratic}).
The transfer schedule then follows from $t^{h}(\theta)=v_{0}-\frac{k}{2}\bigl(\delta^{h}(\theta)\bigr)^{2}-u^{h}(\theta)$. 

Integrating the virtual surplus $\Psi^{h}$ over $\left[0,\tau^{h}\right]$
gives the profit expressions reported. In particular, 
\[
\Pi^{h}=\int_{0}^{\tau^{h}}\Psi^{h}(\theta)\,d\theta.
\]
When $v_{0}\geq c$ the cutoff is $\tau^{h}=1$ and 
\[
\Pi^{h}=\int_{0}^{r/2}\!\Bigl(v_{0}-\tfrac{3k}{2}\theta^{2}\Bigr)d\theta+\int_{r/2}^{r}\!\Bigl(v_{0}+\tfrac{c^{2}}{2k}-2c\theta+\tfrac{k}{2}\theta^{2}\Bigr)d\theta+\int_{r}^{1}\!\bigl(v_{0}-c\theta\bigr)d\theta=v_{0}-\frac{c}{2}+\frac{c^{3}}{12k^{2}}.
\]
For $v_{0}<c$ the upper limit is the interior cutoff $\tau^{h}<1$.
Integrating the same pieces up to $\tau^{h}$ yields, in the ideal
region $c^{2}/k\leq v_{0}<c$, 
\[
\Pi^{h}=\int_{0}^{r/2}\!\Psi^{h}+\int_{r/2}^{r}\!\Psi^{h}+\int_{r}^{\tau^{h}}\!\bigl(v_{0}-c\theta\bigr)d\theta=\frac{v_{0}^{2}}{2c}+\frac{c^{3}}{12k^{2}};
\]
in the interior region $3c^{2}/(8k)\leq v_{0}<c^{2}/k$, 
\[
\Pi^{h}=\int_{0}^{r/2}\!\Psi^{h}+\int_{r/2}^{\tau^{h}}\!\Psi^{h}=\frac{v_{0}r}{2}-\frac{kr^{3}}{16}+a\Bigl(\tau^{h}-\frac{r}{2}\Bigr)-c\Bigl((\tau^{h})^{2}-\frac{r^{2}}{4}\Bigr)+\frac{k}{6}\Bigl((\tau^{h})^{3}-\frac{r^{3}}{8}\Bigr),
\]
with $a:=v_{0}+c^{2}/(2k)$; and in the pooling region $0<v_{0}<3c^{2}/(8k)$,
\[
\Pi^{h}=\int_{0}^{\tau^{h}}\!\Bigl(v_{0}-\tfrac{3k}{2}\theta^{2}\Bigr)d\theta=v_{0}\tau^{h}-\frac{k}{2}(\tau^{h})^{3}=\frac{2v_{0}}{3}\tau^{h}.
\]

\subsection{\protect\label{OA:Uniform-Derivation-Vertical}Derivation of Optimal
Mechanism for Section \ref{Section:Vertical}}

\subsubsection*{Relaxed Allocation and Threshold Objects}

The relaxed mismatch satisfies $m'(d^{R})=c/k+m''(d^{R})(1-F)/f$,
which under quadratic-uniform reduces to $d^{R}(\theta)=r+1-\theta$.
The corner $d^{R}(b^{R})=b^{R}$ gives $b^{R}=(1+r)/2$. Hence $a^{R}(\theta)=0$
for $\theta\le b^{R}$ and $a^{R}(\theta)=2\theta-1-r$ for $\theta\ge b^{R}$.

The stationary points $\eta$ and $\gamma$ satisfy $m'(\eta)=1/k$
and $(1-F(\gamma))/f(\gamma)=(1-c)/[km''(\eta)]$. The first gives
$\eta=1/k$; the second gives $1-\gamma=(1-c)/k$, i.e., $\gamma=(S-1)/k$.
Both lie in $(0,1)$ when $k>1$ and $c<1$. Direct computation gives
\[
\eta-\gamma=\frac{2-S}{k},\qquad b^{R}=\frac{\eta+\gamma}{2},
\]
so $b^{R}$ is the midpoint of $\eta$ and $\gamma$ under this parameterization.

The cumulative-multiplier function $L^{0}(\theta)=(1-F(\theta))H(\theta)$
with $H(\theta)=1-h(\theta)(1-F(\gamma))/f(\gamma)$, $h:=f/(1-F)$,
becomes
\[
L^{0}(\theta)=(1-\theta)\Bigl[1-\tfrac{1-\gamma}{1-\theta}\Bigr]=(1-\theta)-(1-\gamma)=\gamma-\theta.
\]
The adjusted inverse-hazard rate is therefore $(1-F(\theta)-L^{0}(\theta))/f(\theta)=1-\gamma=(1-c)/k$,
constant in $\theta$.

\subsubsection*{The Threshold $\underline{\tau}$}

Suppose $S>2$, so $\eta<\gamma$. The relaxed utility schedule has
slope $u^{R\prime}(s)=1-ks$ on $[0,b^{R}]$ and $1-k(1+r-s)$ on
$[b^{R},1]$. Direct integration gives
\[
u^{R}(\gamma;\tau)=-\tau+\frac{k\tau^{2}}{2}-\frac{S^{2}-4S+2}{4k}.
\]
Setting this to zero yields $2k^{2}\tau^{2}-4k\tau-(S^{2}-4S+2)=0$,
with discriminant $8k^{2}(S-2)^{2}$. The relevant root (less than
$\eta$) is
\[
\underline{\tau}=\frac{2-\sqrt{2}(S-2)}{2k}.
\]
$\underline{\tau}>0$ iff $S<2+\sqrt{2}$; for $S\ge2+\sqrt{2}$ the
relaxed allocation violates IR for every $\tau\in(0,\eta)$. At $S=2$,
$\underline{\tau}=\eta$.

\subsubsection*{The Lemma \ref{lem:p-tau-case-3} Construction}

For $\phi\in(\tau,\gamma)$, the maximizer $a(\theta;\phi)$ of $\tilde{\psi}^{v}(\cdot,\theta;L^{0}(\phi))$
satisfies, by the FOC,
\[
\theta-a(\theta;\phi)=\frac{c}{k}+(1-\theta)-L^{0}(\phi)=\frac{c}{k}+(1-\gamma)+(\phi-\theta)=\eta+(\phi-\theta),
\]
hence $a(\theta;\phi)=2\theta-\eta-\phi$ whenever positive. The corner
threshold is $\beta(\phi)=(\eta+\phi)/2$. For $\theta\in[\beta(\phi),\phi]$,
$\theta-a(\theta;\phi)=\eta+\phi-\theta$, so $1-km'(\theta-a)=k(\theta-\phi)$.

The utility integral becomes
\[
U(\phi;\tau)=\int_{\tau}^{\beta(\phi)}(1-ks)\,ds+\int_{\beta(\phi)}^{\phi}k(s-\phi)\,ds.
\]
Carrying out both integrals and simplifying using $\beta(\phi)=(\eta+\phi)/2$,
\[
U(\phi;\tau)=\frac{\eta+\phi}{2}-\tau+\frac{k\tau^{2}}{2}-\frac{k(\eta^{2}+\phi^{2})}{4}.
\]
Multiplying by $4k$ and substituting $X:=k\phi-1$ (so $k\phi=X+1$,
$k^{2}\eta^{2}=1$), the equation $U(\phi^{\tau};\tau)=0$ reduces
to
\[
X^{2}=2(1-k\tau)^{2}.
\]
Taking $X>0$ (so $\phi^{\tau}>\eta$),
\[
\phi^{\tau}=\eta+\sqrt{2}(\eta-\tau),\qquad\beta^{\tau}=\frac{\eta+\phi^{\tau}}{2}=\eta+\frac{\sqrt{2}}{2}(\eta-\tau).
\]
At $\tau=\underline{\tau}$, $\phi^{\tau}=\gamma$ and $\beta^{\tau}=b^{R}$;
the Lemma \ref{lem:p-tau-case-3} region degenerates and the construction
merges into Lemma \ref{lem:p-tau-case-1}. At $\tau=\eta$, $\phi^{\tau}=\beta^{\tau}=\eta$.

\subsubsection*{Optimal Cutoff $\tau^{v}$}

\paragraph{Case (I).}

When Lemma \ref{lem:p-tau-case-1} applies at the marginal type, $\alpha^{\tau}(\tau)=0$
(since $\tau^{v}<b^{R}$ holds always, verified below) and $\Lambda^{\tau}(\tau)=0$.
The virtual surplus is
\[
\Psi^{v}(\tau)=\psi^{v}(0,\tau;0)=-\tfrac{3k}{2}\tau^{2}+(2+k)\tau-1.
\]
Setting $\Psi^{v}(\tau^{v})=0$ and selecting the smaller root,
\[
\tau_{I}^{v}=\frac{(2+k)-\sqrt{k^{2}-2k+4}}{3k}.
\]
This is independent of $c$. Direct calculation shows $\tau_{I}^{v}\le1/2\le b^{R}$
for all $k>0$.

\paragraph{Case (II).}

When Lemma \ref{lem:p-tau-case-3} applies at $\tau^{v}$, $\alpha^{\tau}(\tau)=0$
(since $\tau<\beta^{\tau}$) and $\Lambda^{\tau}(\tau)=L^{0}(\phi^{\tau})=\gamma-\phi^{\tau}$.
The virtual surplus is
\[
\Psi^{v}(\tau)=\tau-\frac{k\tau^{2}}{2}-\bigl[1-\tau-\gamma+\phi^{\tau}\bigr](1-k\tau).
\]
Substituting $q:=1-k\tau$, $\tau=(1-q)/k$, and $\phi^{\tau}=\eta+\sqrt{2}(\eta-\tau)$
gives $1-\tau-\gamma+\phi^{\tau}=[(1+\sqrt{2})q+(1-c)]/k$. The equation
$\Psi^{v}(\tau^{v})=0$ reduces to
\[
(3+2\sqrt{2})q^{2}+2(1-c)q-1=0.
\]
The positive root is
\[
q(c)=\frac{1}{(1-c)+\sqrt{(1-c)^{2}+3+2\sqrt{2}}},\qquad\tau_{II}^{v}=\frac{1-q(c)}{k}.
\]
Notably, $q(c)$ depends only on $c$: in this regime, $k\tau^{v}=1-q\left(c\right)$
is determined by $c$ alone, and coverage cutoff scales linearly in
$1/k$.

\subsubsection*{Case Boundary $c^{*}(k)$}

Strict quasiconcavity of $\Pi^{v}$ (Lemma \ref{Lemma:Pi-v(v)-quasiconcave})
implies that Case (I) applies iff $\tau_{I}^{v}\le\underline{\tau}$.
Setting equality,
\[
\frac{(2+k)-\sqrt{k^{2}-2k+4}}{3k}=\frac{2-\sqrt{2}(S-2)}{2k},
\]
and solving for $S$,
\[
S^{*}(k)=2+\frac{\sqrt{2}}{3}\bigl[\sqrt{k^{2}-2k+4}-(k-1)\bigr],
\]
which gives $c^{*}(k)=S^{*}(k)-k$ as in the statement. For $k\le1$,
$c^{*}(k)\ge1$, so Case (II) is vacuous (since $c<1$). For $k$
large, $c^{*}(k)<0$, so every admissible $c$ falls in Case (II).

\subsubsection*{Utility Schedules}

The utility formulas in the theorem follow from integrating the slope
$u^{*\prime}(\theta)=1-kM_{2}(\alpha^{\tau}(\theta),\theta)$ and
the boundary conditions $u^{\tau}(\tau^{v})=0$, $u^{\tau}(\phi^{\tau})=u^{\tau}(\gamma)=0$
(the latter two in Case (II) only).

\paragraph{Pooled region ($\alpha^{\tau}=0$).}

Slope is $1-k\theta=-k(\theta-\eta)$, so $u^{\tau}(\theta)+(k/2)(\theta-\eta)^{2}$
is constant. Using $u^{\tau}(\tau^{v})=0$ gives
\[
u^{\tau}(\theta)=\frac{k}{2}\bigl[(\eta-\tau^{v})^{2}-(\eta-\theta)^{2}\bigr].
\]
This applies on $[\tau^{v},b^{R}]$ in Case (I) and $[\tau^{v},\beta^{\tau}]$
in Case (II).

\paragraph{Screening region $[b^{R},1]$ in Case (I).}

Slope is $1-k(1+r-\theta)=k(\theta-\gamma)$, so $u^{\tau}(\theta)-(k/2)(\theta-\gamma)^{2}$
is constant. Using the value at $b^{R}$ from the pooled region and
the identity $(\eta-b^{R})^{2}=(b^{R}-\gamma)^{2}=(\eta-\gamma)^{2}/4$
gives the second formula in the theorem.

\paragraph{Case (II), Lemma \ref{lem:p-tau-case-3} region $[\beta^{\tau},\phi^{\tau}]$.}

Slope is $k(\theta-\phi^{\tau})$, so $u^{\tau}(\theta)-(k/2)(\theta-\phi^{\tau})^{2}$
is constant. The endpoint condition $u^{\tau}(\phi^{\tau})=0$ gives
$u^{\tau}(\theta)=(k/2)(\theta-\phi^{\tau})^{2}$. Continuity at $\beta^{\tau}$
is verified by $u^{\tau}(\beta^{\tau})=(k/2)(\beta^{\tau}-\phi^{\tau})^{2}=k(\eta-\tau^{v})^{2}/4$,
matching the pooled formula evaluated at $\theta=\beta^{\tau}$ since
$\eta-\beta^{\tau}=-(\eta-\tau^{v})/\sqrt{2}$.

\paragraph{Case (II), regions $[\phi^{\tau},\gamma]$ and $[\gamma,1]$.}

On $[\phi^{\tau},\gamma]$, $\alpha^{\tau}=\theta-\eta$ gives $\theta-\alpha^{\tau}=\eta$,
so the slope $1-k\eta=0$ and $u^{\tau}\equiv0$. On $[\gamma,1]$,
$\alpha^{\tau}=2\theta-1-r$ gives the same slope $k(\theta-\gamma)$
as in Case (I), and the endpoint $u^{\tau}(\gamma)=0$ gives $u^{\tau}(\theta)=(k/2)(\theta-\gamma)^{2}$.
\end{document}